\documentclass[12pt]{article}

\usepackage[T1]{fontenc}
\usepackage[utf8]{inputenc}
\usepackage{lmodern}
\usepackage{setspace}
\usepackage{amsmath}
\usepackage{amssymb}
\usepackage{amsfonts}
\usepackage{amsthm}
\newtheorem{assumption}{Assumption}
\usepackage{bm}
\theoremstyle{remark}
\newtheorem{remark}{Remark}[section]
\usepackage{graphicx}
\usepackage{float}
\usepackage{booktabs}
\usepackage{subcaption}
\usepackage{array}
\usepackage{natbib}
\usepackage{algorithm}
\usepackage{algpseudocode}
\newtheorem{corollary}{Corollary}
\usepackage{booktabs}
\usepackage{tikz}
\usetikzlibrary{arrows, chains, positioning, quotes, shapes.geometric}
\usetikzlibrary{arrows.meta}
\tikzset{
node distance = 8mm and 20mm,
start chain = going below,
arrow/.style = {thick,-stealth},
base/.style = {
    draw, thick,
    minimum width=30mm,
    minimum height=10mm,
    align=center,
    inner sep=1mm,
    outer sep=0mm,
    on chain
},
decision/.style = {diamond, base, aspect=2, inner xsep=0mm},
process/.style = {rectangle, base},
startstop/.style = {rectangle, rounded corners, base},
}
\usepackage{subcaption}

\newcommand{\nie}{\mbox{NIE}}
\tikzstyle{line} = [draw, -latex']

\def\bSig\mathbf{\Sigma}

\usepackage{setspace} 
\newtheorem{suppassumption}{Assumption}

\usepackage[a4paper,margin=1in]{geometry}
\def\bSig\mathbf{\Sigma}

\newtheorem{proposition}{Proposition}
\title{Privacy-Preserving Causal Meta-Mediation Analysis with Survival Outcomes}
\author{
Marie-Félicia Beclin\\
\small EPILOGY, Institut Mondor of Biomedical Research, INSERM U955,\\
\small Université Paris-Est Créteil, France\\
\small \texttt{mariefelicia.beclin@gmail.com}
\and
Tat-Thang Vo\\
\small EPILOGY, Institut Mondor of Biomedical Research, INSERM U955,\\
\small Université Paris-Est Créteil, France\\
\small \texttt{tat-thang.vo@u-pec.fr}
}

\begin{document}
\maketitle

\bigskip

\begin{abstract}
Privacy and data-governance constraints often prevent pooling individual-level data across studies, limiting the use of conventional approaches for causal mediation analysis in multicenter settings. We propose a federated causal meta-mediation framework for right-censored time-to-event outcomes that enables collaborative estimation without sharing individual-level data. Our framework targets natural indirect effects in a prespecified population by combining information on mediator and outcome mechanisms across distributed data sources. A site-by-site identification strategy further allows heterogeneity across data sources to be characterized, with a variance decomposition separating outcome-related, mediator-related, and interaction components. We develop federated one-step and targeted maximum likelihood estimators that accommodate data-adaptive and machine-learning methods for nuisance-function estimation. The finite-sample performance of the proposed estimators is evaluated through numerical simulations. To illustrate the practical utility of the framework, we apply it on data from the French National Health Data System to evaluate the role of methotrexate coprescription in explaining the effect of TNFi versus IL-12/23 inhibitor therapy on treatment persistence among psoriatic patients.
\end{abstract}


\maketitle
\section{Introduction}
The increasing availability of multicenter and real-world data has created unprecedented opportunities for causal treatment evaluation in health research \citep{deng2026survey}. However, growing concerns regarding privacy, regulatory constraints, and institutional data governance often restrict the direct sharing of individual-level data, notably under regulations such as the General Data Protection Regulation in Europe \citep{gdpr} and the U.S. Health Insurance Portability and Accountability Act \citep{hipaa}. Federated learning \citep{mcmahan2017communication, li2020federated} has emerged as a promising paradigm that enables collaborative analysis by exchanging only summary-level information or model updates. While methodological developments in federated learning have largely focused on predictive modeling tasks such as classification, regression, and deep learning, increasing attention has also been devoted to causal inference, where the goal is estimation of interpretable causal parameters such as the average treatment effect (ATE) \citep{han2025federated,khellaf2025federated,xiong2023federated}. 
Unlike predictive objectives, estimating the ATE requires accounting for confounding structures, treatment assignment mechanisms and case-mix differences across distributed data sources \citep{vo2019novel}. Several federated methods have been proposed for ATE estimation with continuous and discrete outcomes \citep{vo2022bayesian, han2025federated, xiong2023federated}. In contrast, relatively little work has addressed time-to-event outcomes. Recently, \citet{liuprivacy} proposed a federated One-Step (OS) estimator for estimating the ATE on survival outcomes in a target population by reweighting source-site contributions to account for distributional shifts. However, this approach relies on a parametric exponential tilt model for estimating the density ratio between participating sites, rendering it vulnerable to model misspecification. In addition, Targeted Maximum Likelihood Estimation (TMLE), which enjoys several theoretical and practical advantages, including respecting the natural parameter space and often exhibiting superior finite-sample performance \citep{van2006targeted}, has not yet been systematically studied in the federated setting. 

Beyond the ATE, little attention has been devoted to evaluating causal mechanisms in distributed environments, despite its central role in elucidating pathways through which the treatment exerts its effect on the outcome \citep{jang2026privacy}. To the best of our knowledge, no federated methods currently exist for the estimation of natural direct and indirect treatment effects in the presence of right-censored outcomes. Developing privacy-preserving, statistically efficient methods for decentralized mediation analysis with survival data  remains an open, unaddressed research question.

In this paper, we aim to develop a horizontal federated causal meta-mediation framework for time-to-event outcomes. We consider a collaborative setting in which multiple institutions jointly estimate mediation effects in an external target population without sharing individual-level data. Our approach integrates causal mediation analysis, transportability, survival modeling, semiparametric theory and federated estimation into a unified framework that supports efficient distributed inference. Leveraging efficient influence functions, we derive federated, privacy-preserving OS and TMLE estimators for the NDIEs on survival outcomes. Throughout, privacy-preserving refers to keeping individual-level observations local; we do not claim formal differential-privacy guarantees. Both estimators accommodate flexible, data-adaptive machine learning methods for estimating nuisance functions while retaining desirable asymptotic properties. The finite-sample performance of the proposed estimators is
evaluated through numerical simulations. To illustrate the practical utility of the framework, we apply it on data
from the French National Health Data System to evaluate the role of methotrexate coprescription in explaining
the effect of TNFi versus IL-12/23 inhibitor therapy on treatment persistence among psoriatic patients.

\section{Transportability of natural direct and indirect effects} \label{sec:2}
Consider $K$ studies ($S=1,\ldots,K$) investigating the mediating role of a vector of intermediate variables $M$ in explaining the causal effect of a binary treatment $A$ on a right-censored time-to-event outcome $T\in\mathcal{T}=\{t_0,\ldots,t_{h}\}\cup\{\infty\}$, where $T=\infty$ indicates that the event does not occur during the follow-up period. Let $C\in\mathcal{T}$ denote the censoring time and $L$ a set of baseline covariates measured consistently across all studies. Let $\tilde{T} =  \text{min}(T, C)$ denote the observed follow-up time and $\delta = I(T \leq C)$ denote the event indicator.
We assume that the mediator is measured shortly after treatment assignment at a fixed time $t_M<t_0$, so that it is not subject to truncation by the event of interest. Examples include acute inflammatory biomarkers measured shortly after vaccine administration, or clinical interventions such as surgery or chemotherapy that occur before the start of follow-up \citep{rochon2014mediation}. 

Let $M(a)$ denote the potential value of the mediator under treatment level $A=a$, and let $T(a,m)$ denote the potential event time had treatment been set to $a$ and the mediator to $m$. Our objective is to estimate the natural indirect effect (NIE) on survival at a clinically relevant time point $\nu\in\mathcal{T}$ in a target population $S=0$, for which only baseline covariates $L$ are available. Specifically,
\[
\operatorname{NIE}^{a}(\nu)
=
\Pr\!\left\{T\bigl(a,M(1)\bigr)>\nu\mid S=0\right\}
-
\Pr\!\left\{T\bigl(a,M(0)\bigr)>\nu\mid S=0\right\}.
\]

To identify the NIE within each study, we adopt the standard assumptions for causal mediation analysis.

\begin{assumption}\label{ass:nie}(Identification of $\nie^{a}(\nu)$) \begin{itemize} 
\item [(i)] (Ignorability) $T(a,m)\ \perp A \mid L, S$ and $M(a)\ \perp\ A \mid L,S.$ \item [(ii)] (Consistency) $T = T(a,m)$ if $(M,A)=(m,a)$ and $M = M(a)$  if $A=a$. \item [(iii)] (Cross-world independence) $T(a,m)\ \perp\ M(a^*)\mid L, S.$ \item [(iv)] (Positivity) $   0<\mathbb{P}(A=a\mid L=l,S=s), \mathbb{P}(M=m\mid A=a,L=l,S=s),\mathbb{P}(S=s\mid L=l) <1,$   $\mathbb{P}(C \ge t \mid A=a, M=m, L=l, S=s) >0$ and $\mathbb{P}(R=0 \mid A=a,L=l,S=s) >0$ for all relevant $(a,m,s)$,  $t\leq\nu$. 
\end{itemize} 
\end{assumption}

These assumptions require that all common causes of the treatment, mediator, and outcome are measured within each study and that no mediator--outcome confounder is itself affected by the treatment, i.e. there is no treatment-induced mediator--outcome confounding. In practice, this is more plausible when there is relatively little time between $A$ and $M$, leaving limited opportunity for post-baseline variables affecting both $M$ and $T$ to arise \citep{vanderweele2016mediation}. 
To transport study-specific effects to a common target population, we further assume that all variables modifying the treatment--mediator and mediator--outcome relationships that differ across studies are included in $L$.
\begin{assumption}[Transportability]
\label{ass:transportability}
$T(a,m)\perp S\mid L
\quad\text{and}\quad
M(a)\perp S\mid L.$
\end{assumption}
This assumption implies that, conditional on $L$, the mediator and outcome generating mechanisms are invariant across the source and target populations, i.e.,
$\mathbb{P}(M=m\mid A,L,S=s)$ and $\mathbb{P}(T>t\mid A,M,L,S=s)$ do not depend on $s$, allowing information from the source studies to be transported to estimate $\operatorname{NIE}^{a}(\nu)$ in the target population.

In many applications, mediator values may be missing despite not being truncated by the occurrence of the event. In these applications, we assume that the missingness mechanism depends only on variables observed prior to mediator measurement. Specifically, let $R$ denote the missingness indicator for the mediator $M$, with $R=1$ if $M$ is missing and $R=0$ otherwise. Then:

\begin{assumption}[Missing at Random]
$R\perp(M,T,C)\mid A,L,S=k,
\qquad
k=1,\ldots,K.$
\end{assumption}

Finally, we assume that censoring is conditionally independent of the event time given the observed treatment, baseline covariates, and mediator among individuals with observed mediator values.

\begin{assumption}[Conditional Independence of Censoring] \label{ass:cens}
$C\perp T\mid A,L,M,R=0,S=k$ for $k=1,\ldots,K.$
\end{assumption}

\begin{figure} \label{fig:dag}
\centering
\begin{tikzpicture}[node distance=1cm and 1cm]

%=====================
% Nodes
%=====================

\node (a) at (1,1) {$A$};
\node (t) at (4,1) {$T$};
\node (m) at (2.5,0.5) {$M$};
\node (l) at (1.5,2) {$L$};
\node (s) at (0,2) {$S$};
\node (delta) at (5.5,0.5) {$\delta$};
\node (ttilde) at (5.5, 2.5)  {$\tilde{T}$};
\node (cens) at (4,2.5) {$C$};

%=====================
% Core structure
%=====================

\draw[->] (a) -- (t);
\draw[->] (a) -- (m);
\draw[->] (m) -- (t);
\draw[->] (m) -- (cens);

\draw[->] (l) -- (a);
\draw[->] (l) -- (m);
\draw[->] (l) -- (t);
\draw[-] (l) -- (s);

\draw[->] (s) -- (a);
%\draw[->,  bend left=55] (s) to (t);
%\draw[->,  bend right=55] (s) to (m);

%=====================
% Censoring and observation
%=====================

\draw[->] (t) -- (delta);
\draw[->] (cens) -- (delta);
\draw[->] (t) -- (ttilde);
\draw[->] (cens) -- (ttilde);

\draw[->] (a) -- (cens);
\draw[->] (l) -- (cens);
\draw[->, bend left=35] (s) to (cens);

\end{tikzpicture}

\caption{Causal diagram. The arrows $S\rightarrow A$;  $S\rightarrow L$ and $S\rightarrow C$ reflects differences across sites in treatment assignment mechanism, covariate distribution and censoring mechanism, respectively.
}
\end{figure}
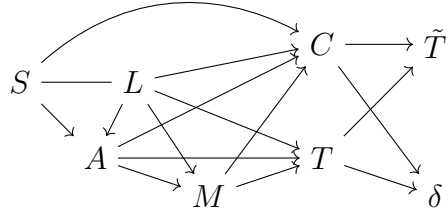

The proposed assumptions are satisfied under the causal diagram depicted in Figure~1, which encodes a nonparametric structural equation model with independent error terms. %Consider $S = \{1, . . . , K\}$ mediation studies that evaluate the role of an intermediate variable $M$ in explaining the causal effect of a binary exposure $A$ on an outcome $Y$. In addition, consider $S = \{K + 1, . . . , P\}$ studies that evaluate the causal effect of $A$ on the mediator $M$. 
Under Assumptions~\ref{ass:nie}--\ref{ass:cens}, the counterfactual survival probability $\theta^{a, a^{*}}(\nu):=\mathbb{P}[T\bigl(a, M(a^{*})\bigr)>\nu\mid S=0]$ is identified by the functional $\theta^{a,a^*}_{k,p}(\nu)$, where:
\begin{eqnarray}
    \theta^{a, a^{*}}_{k,p}(\nu) &=& \mathbb{E}\biggl[\mathbb{E} \bigl[ \mathbb{P}(T>\nu\mid a,M,L,S=k)\mid a^*,L,S=p \bigr]\mid S=0\biggr]\\
    &=&\mathbb{E} \Biggl[  \mathbb{E} \biggl[ \prod_{t\leq \nu} \Bigl(1-  \mathbb{P}(\tilde{T}=t, \delta=1 \mid a,k,  L, M, R=0, \tilde{T}\geq t  )\Bigr) ~\bigg|~ a^{*}, p, R=0 \biggr] ~\bigg| S=0 \Biggr]. \nonumber
\end{eqnarray} 
%The target parameter $\mathrm{NIE}^{a}(\nu)$ can then be identified by $\zeta_{k,p} = \theta_{k,p}^{a,1} - \theta_{k,p}^{a,0}$.
For each pair of outcome and mediator source sites $(k,p)$, define $\zeta_{k,p} ^{a}(\nu) = \theta_{k,p}^{a,1}(\nu) - \theta_{k,p}^{a,0}(\nu)$. We suppress the dependence on a and $\nu$ when no ambiguity arises.
This identification result suggests a natural collaborative estimation strategy. Specifically, the treatment--mediator relationship may be learned from site $p$, while the mediator--outcome relationship may be learned from site $k$. These quantities are then transported to the target population by standardizing over the covariate distribution in the target site. By considering every pair of mediator and outcome sources, this strategy yields a collection of standardized effect estimates that allow one to assess potential between-site heterogeneity. 
Specifically, let $\zeta = \sum_{k,p} w_kq_p\zeta_{k,p}$
and $\tau^2 = \sum_{k,p}w_kq_p(\zeta_{k,p}-\zeta)^2,$ 
where $(w_k)$ and $(q_p)$ are prespecified weight functions for outcome and mediator sources, with $\sum_k w_k=1$ and $\sum_p q_p=1$. The total variance $\tau^2$ can then be decomposed as:
\begin{equation} \label{eq:anova}
    \tau^2 =
    \underbrace{\sum_k w_k(\zeta_{k,\cdot}-\zeta)^2}_{ \text{Outcome-related variation }(\xi^2)}
    +
    \underbrace{\sum_p q_p(\zeta_{\cdot,p}-\zeta)^2}_{ \text{Mediator-related variation }(\eta^2)}
    +    \underbrace{\sum_{k,p}w_kq_p(\zeta_{k,p}-\zeta_{k,\cdot}-\zeta_{\cdot,p}+\zeta)^2}_{\text{Interaction}}.
\end{equation}
where $\zeta_{k,\cdot}= \sum_{p} q_p \zeta_{k,p}$ and $\zeta_{\cdot,p}=\sum_{k} w_k \zeta_{k,p}$.
Under the proposed identification assumptions, $\zeta = \operatorname{NIE}^{a}(\nu)$ and all variance components are zero. However, Assumption \ref{ass:transportability} may be violated due to some heterogeneity in the outcome distribution across sites, i.e., $M(a)\perp\!\!\!\perp S\mid L$ but $T(a,m)\not\perp\!\!\!\perp S\mid L$. In this case, $\zeta_{k,p}$ may vary across $k$ but remains constant across $p$, i.e., $\zeta_{k,p} \ne \zeta_{k',p}$ but $\zeta_{k,p} = \zeta_{k,p'}$ for any $k,k',p,p'$. Consequently, the first variance component in (2) is nonzero, whereas the second variance component is zero. In contrast, when heterogeneity exists in the mediator distribution across sites, $\zeta_{k,p}$ may vary across $p$ but remains constant across $k$, i.e., $\zeta_{k,p} = \zeta_{k',p}$ but $\zeta_{k,p} \ne \zeta_{k,p'}$ for any $k,k',p,p'$. Consequently, the second variance component in (2) is nonzero, whereas the first is zero. 
When both outcome- and mediator-related heterogeneity are present, the interaction component captures the remaining variation that cannot be explained by either source separately. In particular, this component is zero when the variability within $\boldsymbol{S_k}=\{\zeta_{k,1},\ldots,\zeta_{k,K}\}$ does not depend on $k$ and the variability within $\boldsymbol{S_p}=\{\zeta_{1,p},\ldots,\zeta_{K,p}\}$ does not depend on $p$.

In practice, potential violations of Assumption \ref{ass:transportability} can thus be assessed by estimating each variance component in (2). When a variance component is substantially larger than what would be expected from sampling variability, this provides evidence suggestive of heterogeneity in the mediator- and/or outcome-generating mechanism and raises concerns about the validity of the transportability assumption. A site-by-site identification strategy, as in (1), thus provides a useful diagnostic for identifying and quantifying potential heterogeneity and encourages caution regarding unrecognized sources of heterogeneity across sites.
For completeness, the Online Supplementary Materials present an alternative identification strategy that first pools information across studies before estimating the nuisance functions. Although this approach may improve statistical efficiency when the transportability assumption holds across sites, it produces a single standardized estimate and  does not permit assessment of heterogeneity across mediator and outcome sources, as suggested above. 

In the Online Supplementary Materials, we further extend the proposed identification and heterogeneity assessment strategy to natural path-specific effects, when the interest lies in evaluating the fine-grained contribution of each individual mediator in $M$ in explaining the total treatment effect. 
\section{Federated estimation strategies} \label{sec:para}
We now describe several federated estimation strategies for $\theta^{a,a^*}_{k,p}(\nu)$ under the constraint that individual-level data cannot be shared across sites. To simplify the presentation, let $I(v)=I(V=v)$ denote the indicator function for a generic variable $V$, and define the nuisance parameters

%$\bm \eta = (U_s, V_s, \pi^a_k, \tau^{a}_{k}, \rho^{a}_{k},\lambda_{k}^{a}(t), g_{k}^{a}(t), S_{k}^{a}(t), G_{k}^{a}(t) b_{k,p}^{a, a^*})$, where:
%$U_p=\mathbb{P}(S=p\mid L)$, $\pi_k^a=\mathbb{P}(A=a\mid L,k)$,  $V_k=\mathbb{P}(S=k\mid L,M,R=0)$, $\tau_k^a=\mathbb{P}(A=a\mid L,M,R=0,k)$, $\rho_k^a=\mathbb{P}(R=0\mid a,L,k)$, $g_k^a(t)=\mathbb{P}(C=t\mid a,L,M,k,C\geq t)$, $\lambda_k^a(t)=\mathbb{P}(T=t\mid T\geq t,a,L,M,k)$, $G_k^a(t)=\mathbb{P}(C>t\mid a,L,M,k)$,$S_k^a(t)=\mathbb{P}(T>t\mid a,L,M,k)$, $ b_{k,p}^{a,a^*}=\mathbb{E}\!\left[S_k^a(\nu)\mid a^*,p,L\right].$

 \begin{align*}
    &U_p=\mathbb{P}(S=p \mid L) &
         &\pi^{a}_{k}=\mathbb{P}(A=a \mid L, k) \\
     &V_k=\mathbb{P}(S=k \mid L, M, R=0) &
     &\tau^{a}_{k}=\mathbb{P}(A=a \mid L, M, R=0,k)\\
     &\rho^{a}_{k} = \mathbb{P}(R=0 \mid a,L,k) & 
     &g_{k}^{a}(t) = \mathbb{P}(C=t\mid a, L, M,k, C \geq t)\\
     &\lambda_{k}^{a}(t) = \mathbb{P}(T=t\mid T\geq t, a, L, M, k) &
     &G_{k}^{a}(t) = \mathbb{P}(C>t\mid a, L, M, k) \\ 
     &S_{k}^{a}(t) = \mathbb{P}(T>t\mid a, L, M, k). & 
     &b_{k,p}^{a, a^*} =\mathbb{E}\Bigl[S_{k}^{a}(\nu) \mid a^*, p, L\Bigr]\\
 \end{align*}

The event hazard function $\lambda_k^a(t)$ and the censoring hazard function $g_k^a(t)$ can be expressed in terms of observable data under Assumption~\ref{ass:cens}, i.e.,
$$\lambda_k^a(t)=\mathbb{P}\left(
\tilde{T}=t,\delta=1
\mid a,k,L,M,R=0,\tilde{T}\geq t
\right) \text{ and},$$
$$g_k^a(t)=
\mathbb{P}\left(
\tilde{T}=t,\delta=0
\mid
a,k,L,M,R=0,
\{\tilde{T}>t\}\cup\{\tilde{T}=t, \delta=0 \}
\right).$$
The corresponding survival functions are obtained from the event and censoring hazard functions as
$
S_k^a(t)=
\prod_{t_j\leq t}
\left(1-\lambda_k^a(t_j)\right)
$ and $
G_k^a(t)=
\prod_{t_j\leq t}
\left(1-g_k^a(t_j)\right).$ 

The above nuisance functions can be naturally partitioned into three groups according to the degree of information sharing required for their estimation:
\begin{itemize}
    \item[(i)] \textbf{Federated parameters.} The first group, $\bm{\eta}_{k,p}^{\mathrm{fed}} = \{ U_p, V_k\}$, comprises nuisance parameters that require information from multiple studies and must  be estimated using a federated optimization algorithm. In such algorithms, each study computes updates based on its local data and transmits only summary quantities (e.g., gradients or model parameters) to a coordinating server, which aggregates them to update a shared model. This decentralized optimization is possible because many learning objectives satisfy an additive empirical risk decomposition,
    $\mathcal{L}(\eta)=\sum_{k=1}^K \mathcal{L}_k(\eta)$,
    implying that the global gradient can be expressed as the sum of study-specific gradients. Numerous federated optimization algorithms have been proposed, differing in their communication efficiency, computational complexity, and convergence guarantees \citep{kairouz2021advances, li2020federated}. Throughout this paper, we adopt the Federated Averaging (FedAvg) algorithm, owing to its communication efficiency and straightforward implementation \citep{mcmahan2017communication, li2020federated}. Algorithm 1 illustrates the use of FedAvg for fitting a regularized multinomial classification model. Nevertheless, the proposed methodology is compatible with any federated learning algorithm capable of estimating the required nuisance functions.
    \item[(ii)] \textbf{Local parameters.} The second group, $\bm{\eta}_{k,p}^{\mathrm{local}} = \{ S_{k}^{a}, G_{k}^{a}, \lambda_{k}^{a}, g_{k}^{a}, \tau_{k}^{a}, \pi_{k}^{a}, \rho_{k}^{a}\}$, consists of nuisance parameters that depend exclusively on data from a single study. These quantities can  be estimated independently within each site without any communication. Alternatively, they can be estimated using a federated algorithm that includes site membership as a covariate, thereby allowing the fitted nuisance functions to vary across sites while potentially borrowing information across them. This approach may improve finite-sample efficiency, particularly when some sites have limited sample sizes.  

    \item[(iii)] \textbf{Mixed parameters.} The final group, $\bm{\eta}_{k,p}^{\mathrm{mix}} = \{ b_{k,p}^{a, a^*} \}$, consists of nuisance parameters involving nested conditional expectations that cannot be estimated either entirely locally or through a single federated optimization procedure. Their estimation instead proceeds sequentially across studies. 
    Specifically, one source site is first used to estimate the inner conditional expectation. The resulting fitted model is then transferred to a second site, where it is used to generate predicted values that subsequently serve as pseudo-outcomes for estimating the outer conditional expectation. Consequently, estimation of $\bm{\eta}_{k,p}^{\mathrm{mix}}$ requires communication across studies, although each individual modeling step is performed locally.
\end{itemize}

\begin{algorithm}[H]

\small

\caption{Federated Multinomial Regression with Convergence-Based Stopping}

\label{alg:FedAvg}

\begin{algorithmic}[1]

\Require Source populations and target population $k=0,\dots,K$ with private data
$\mathcal{D}_k=\{(L_i,S_i)\}_{i=1}^{n_k}$,
$n=\sum_{k=0}^{K} n_k$;
multinomial probabilities
$f_s(\beta;L)=\mathbb{P}_{\beta}(S=s\mid L)$,
with $f_s(\beta;L)\geq 0$ and
$\sum_{s=0}^{K} f_s(\beta;L)=1$; individual negative log-likelihood contribution
$\ell_i(\beta)=-\log f_{S_i}(\beta;L_i)$;
penalty function $\mathcal{P}(\beta)$ with regularization parameter $\kappa\geq0$;
learning rate $\epsilon$,
local epochs $E$, convergence threshold $\alpha$

\Ensure Global parameters $\beta$

\State Initialize $\beta^{(0)}$
\State $t \leftarrow 0$
\State $\Delta \leftarrow +\infty$

\While{$\Delta > \alpha$}

    \State Server broadcasts $\beta^{(t)}$ to all participating clients

    \For{each $k =0\,\ldots,K$ \textbf{in parallel}}

        \State $\beta_k^{(t,0)} \leftarrow \beta^{(t)}$

        \For{$e=0,\dots,E-1$}

            \For{each $(L_i,S_i)\in\mathcal{D}_k$}

                \State
                $\ell_i\!\left(\beta_k^{(t,e)}\right)
                =
                -\log
                f_{S_i}\!\left(
                \beta_k^{(t,e)};L_i
                \right)$

            \EndFor

            \State
            $\beta_k^{(t,e+1)}
            \leftarrow
            \beta_k^{(t,e)}
            -
            \epsilon
            \left\{
            \frac{1}{n_k}
            \sum_{i=1}^{n_k}
            \nabla_{\beta}
            \ell_i\!\left(\beta_k^{(t,e)}\right)
            +
            \kappa
            \nabla_{\beta}
            \mathcal{P}\!\left(\beta_k^{(t,e)}\right)
            \right\}$

        \EndFor

        \State Client $k$ sends
        $\beta_k^{(t+1)}=\beta_k^{(t,E)}$
        to server

    \EndFor

    \State
    $\beta^{(t+1)}
    \leftarrow
    \sum_{k=0}^{K}
    \frac{n_k}{n}
    \beta_k^{(t+1)}$

    \State
    $\Delta
    \leftarrow
    \left\|
    \beta^{(t+1)}-\beta^{(t)}
    \right\|$

    \State $t \leftarrow t+1$

\EndWhile

\State \Return $\beta^{(t)}$ and number of communication rounds $t$

\end{algorithmic}

\end{algorithm}

\subsection{Parametric estimation strategies}
\paragraph{Parametric G-computation} The discrete-time hazard functions $\lambda_{k}^{a}(t_j \mid M,L, R=0)$ are estimated at source site $k$ using a parametric classification model, such as logistic regression. The fitted model parameters are then transmitted to site $p$, where they are used to compute the predicted survival probability
$
\hat  S_{k}^{a}(\nu \mid M,L, R=0)
=
\prod_{t_j \leq \nu}
\left(
1-\hat \lambda_{k}^{a}(t_j \mid M,L, R=0)
\right).
$
To obtain an estimate of the conditional expectation
$
\mathbb{E}_M\!\Bigl[S_{k}^{a}(\nu \mid M,L, R=0)\mid a^*,p, L\Bigr],
$
these predicted survival probabilities are regressed on baseline covariates among individuals with $A=a^*$ at site $p$, yielding
$\hat  b_{k,p}^{a,a^*}(L).$ The resulting G-formula estimator is
$
\hat {\theta}^{\mathrm{Gformula}}_n
=
\frac{1}{n_0}
\displaystyle \sum_{i:S_i=0}
\hat  b_{k,p}^{a,a^*}(L_i).
$

\paragraph{Parametric inverse-probability weighting} Alternatively, $\theta_{k,p}^{a,a^*}$ can be reexpressed as:
\begin{align}
   \theta^{a,a^{*}}_{k,p}(P)  &=   \frac{1}{\mathbb{P}(S=0)}
   \mathbb{E}\Biggl[ \frac{\Omega^{a,a^{*}}_{k,p}(L,M)}{G_{k}^{a}(\nu\mid L, M, R=0)} \, \, 
       I(\tilde{T} > \nu)\, I(a,k,R=0) \Biggr] \label{IPW_1_main}\\
  &= \frac{1}{\mathbb{P}(S=0)} \mathbb{E}\biggl[   W^{a^{*}}_{p}(L) \, \,I(a^*,p,R=0) \, \,  S_{k}^{a}(\nu \mid L, M, R=0)  \biggr] \label{IPW_2_main}
\end{align}
where
\begin{eqnarray*}
    \displaystyle \Omega^{a,a^{*}}_{k,p} (L,M) &=&\frac{ \tau^{a^*}_{p}(L,M) \,\,   V_p(L,M) \, \, U_0(L)  }{\tau^{a}_{k}(L,M) \,\,   V_k(L,M) \, \pi^{a^*}_{p}(L) \, U_p(L)  \,\,   \rho^{a^*}_{p}(L) } \\
\displaystyle W^{a^{*}}_{p}(L)&=& \frac{U_0(L) }{U_p(L) \,  \pi^{a^*}_{p}(L) \,\,   \rho^{a^*}_{p}(L)}. 
\end{eqnarray*}

Estimation based on the weighting-based representation (\ref{IPW_1_main}) requires parametric specification of the nuisance functions $\eta^{IPCW} = \{U_0; V_p,V_k,  \tau^{a}_{k},\tau^{a^*}_{p},\pi^{a^*}_{p}, g_{k}^{a},\rho^{a^*}_{p}\}$. 
By contrast, the alternative representation (\ref{IPW_2_main}) requires parametric specification of $\eta^{IPW} = \{U_0,U_p, \pi^{a^*}_{p}, \rho^{a^*}_{p},\lambda_{k}^{a}(t)\}$. The two representations  lead to distinct federated weighting estimators of $\theta_{k,p}^{a,a^*}$, characterized by different modeling requirements and robustness properties. 

Asymptotic standard errors of parametric approaches can be obtained using M-estimation theory.
Specifically, let $
\bm{\Psi}_{k,p,i}^{a,a^*}
(\bm \vartheta)
= \left(
\psi_{\theta,i}^{a,a^*},
\bm{\psi}^{\mathrm{local}}_{k,p,i},
\bm{\psi}^{\mathrm{fed}}_{k,p,i},
\bm{\psi}^{\mathrm{mix}}_{k,p,i}\right)^\top,$ be the estimating function of the target parameters and all nuisance functions, where $\bm \vartheta=(\theta_{k,p},\bm{\eta}_{k,p})$. Under standard regularity conditions, the asymptotic covariance matrix of
\(\hat{\bm{\vartheta}}\)
is estimated by the sandwich estimator $\widehat {\mathrm{Var}}(\hat{\bm{\vartheta}})
=
\hat A^{-1}
\hat B
(\hat A^{-1})^\top.$
%where
%$\hat A=-\frac{1}{n}\sum_{i=1}^{n}\frac{\partial\bm{\Psi}_{k,p,i}^{a,a^*}(\bm{\vartheta})}{\partial\bm{\vartheta}^{\top}}\bigg|_{\hat{\bm{\vartheta}}},$ and $\hat B= \frac{1}{n}\sum_{i=1}^{n}\bm{\Psi}_{k,p,i}^{a,a^*}(\hat{\bm{\vartheta}})\bm{\Psi}_{k,p,i}^{a,a^*}(\hat{\bm{\vartheta}})^{\top}$. % Because each component of \(\bm{\Psi}_{k,p,i}^{a,a^*}\) depends only on the data available in the study where it is evaluated, both matrices admit a federated decomposition. 
Each site \(s\) computes $\hat A_s
=
-\frac{1}{n}
\sum_{S_i=s}
\frac{\partial
\bm{\Psi}_{k,p,i}^{a,a^*}(\bm{\vartheta})}
{\partial\bm{\vartheta}^{\top}}
\bigg|_{\hat{\bm{\vartheta}}}$ 
and $
\hat B_s
=
\frac{1}{n}
\sum_{S_i=s}
\bm{\Psi}_{k,p,i}^{a,a^*}(\hat{\bm{\vartheta}})
\bm{\Psi}_{k,p,i}^{a,a^*}(\hat{\bm{\vartheta}})^{\top}.$
The site-specific matrices \(\hat A_s\) and \(\hat B_s\)
are then transmitted to the coordinating server, which aggregates them as
$
\hat A=\sum_s\hat A_s$ and $
\hat B=\sum_s\hat B_s$ to compute the final sandwich estimator. Consequently, variance estimation can be implemented without sharing individual-level observations: each site transmits only aggregated score cross-products and derivative matrices to the coordinating server.

%Variance can alternatively be estimated using a nonparametric bootstrap, which is valid for M-estimators under standard regularity conditions. In the federated setting, bootstrap resampling can be implemented without sharing individual-level data. Specifically, the target population site generates bootstrap samples by resampling subject indices and communicates the corresponding bootstrap indices to the participating sites. Each site then reconstructs the corresponding bootstrap sample locally, re-estimates the required nuisance parameters, and recomputes the federated estimator $\hat{\theta}_{k,p}^{a,a^*}$. Repeating this procedure over $B$ bootstrap replicates yields the empirical bootstrap distribution of the estimator, from which the variance and confidence intervals can be obtained.

\subsection{Federated data-adaptive estimation strategies}

Parametric approaches require correct specification of the nuisance models involved in the chosen identifying representation and may yield inconsistent estimates under model misspecification. To address this limitation, we develop efficient estimators that remain $\sqrt{n}$-consistent and asymptotically normal when nuisance functions are estimated using flexible, data-adaptive methods with slower convergence rates.

Our development is based on the efficient influence function (EIF), denoted by $\varphi_{k,p}^{a,a^*}(O,\bm{\eta}_{k,p})$, associated with the target parameter. Let $\tilde{\bm{\eta}}$ be an arbitrary fixed value of the nuisance parameters $\bm\eta$, corresponding for example to an initial estimator or its probability limit. The EIF yields the first-order expansion \begin{align}
\theta_{k,p}^{a,a^*}(\tilde{\bm{\eta}}_{k,p})
-
\theta_{k,p}^{a,a^*}(\bm{\eta}_{k,p})
=
-\mathbb{E}\!\left\{
\varphi_{k,p}^{a,a^*}(O,\tilde{\bm{\eta}}_{k,p})
\right\}
+
R(\bm{\eta}_{k,p},\tilde{\bm{\eta}}_{k,p}),
\end{align}
where $R(\bm{\eta}_{k,p},\tilde{\bm{\eta}}_{k,p})$ is a second-order remainder term. This remainder can typically be expressed as a sum of products of nuisance estimation errors of the form
$$
\mathbb{E}\!\left[
c(\bm{\eta}_{k,p},\tilde{\bm{\eta}}_{k,p})
\{f(\tilde{\bm{\eta}}_{k,p})-f(\bm{\eta}_{k,p})\}
\{g(\tilde{\bm{\eta}}_{k,p})-g(\bm{\eta}_{k,p})\}
\right].
$$ 
In contrast, the first-order terms $-\mathbb{E}\{\varphi_{k,p}^{a,a^*}(O,\tilde{\bm{\eta}}_{k,p})\}$ can be rewritten as $$\mathbb{P}_n \Big\{ \varphi_{k,p}^{a,a^*}(O,\bm{\eta}_{k,p}) - \varphi_{k,p}^{a,a^*}(O,\tilde{\bm{\eta}}_{k,p}) \Big\}
+ (\mathbb{P}_n - \mathbb{P})\Big\{ \varphi_{k,p}^{a,a^*}(O,\tilde{\bm{\eta}}_{k,p}) - \varphi_{k,p}^{a,a^*}(O,\bm{\eta}_{k,p}) \Big\}.
$$ Under certain conditions (discussed later) on the estimated nuisance parameters, the empirical process on the right-hand side is asymptotically negligible. The bias term $-\mathbb{E}\{\varphi_{k,p}^{a,a^*}(O,\tilde{\bm{\eta}}_{k,p})\}$ thus reduces to
$\mathbb{P}_n\{\varphi_{k,p}^{a,a^*}(O,\tilde{\bm{\eta}}_{k,p})\},$
typically called \textit{plug-in bias}. Combining the initial estimator with an EIF-based correction that removes this bias is the principle behind efficient estimators such as OS and TMLE. These estimators are consistent, asymptotically normal, and achieve the semiparametric efficiency bound under some regularity conditions. 

The EIF of $\theta_{k,p}^{a,a^*}$ in the nonparametric statistical model $\mathcal{M}$ is 
\begin{align}\label{eif} 
     \varphi^{a,a^*}_{k,p}(O,\bm{\eta}_{k,p})=&   
     -\Omega^{a,a^*}_{k,p}(L,M) P_0^{-1} I(a,k, R=0)\\ & \underbrace{\quad \quad \quad \quad  \times \sum_{t_d\leq \nu}  I( \tilde{T} \geq t_d)\mathcal{G}_{k}^{a,(t_d)}(\nu \mid L, M)  \big(I(\tilde{T}=t_d , \delta=1)  -\lambda_{k}^{a}(t_d\mid L,M) \big)}_{A_{k,p}^{a,a^*}} \notag \\
     & + \underbrace{I(a^*,p, R=0)P_0^{-1} W_{p}^{a^{*}}(L)  \big(S_{k}^{a}(\nu \mid L,M)-  b_{k,p}^{a,a^*}(L) \big)}_{B_{k,p}^{a,a^*}} \notag\\
     &+ I(S=0) P_0^{-1}  \big( b_{k,p}^{a,a^*}(L) -   \theta_{k,p}^{a,a^*}  \big) \notag
\end{align}
where $\mathcal{G}_{k}^{a,(t_d)}(\nu \mid L, M) = \prod_{  t_{d+1} \leq t_j \leq \nu } \bigl(1-\lambda_{k}^{a} (t_j \mid L,M) \bigr)/ \prod_{t_j\leq t_{d-1}} \bigl(1-g_{k}^{a} (t_j \mid L,M)\bigr)$ and $P_0 = \mathbb{P}(S=0)$. A proof based on the functional delta method is provided in the Supplementary Materials \citep{van1998functional}.
\paragraph{Federated One-Step Estimator} 
The One-Step estimator corrects the first-order bias by adding an EIF-based correction to the plug-in estimator, i.e.,
$\hat{\theta}_{k,p}^{a,a^*,\text{OS}}
=
\hat{\theta}_{k,p}^{a,a^*}
+
\hat{A}_{k,p}^{a,a^*}
+
\hat{B}_{k,p}^{a,a^*}
$. 
Here, $\hat{A}_{k,p}^{a,a^*}$ and $\hat{B}_{k,p}^{a,a^*}$ denote the empirical counterparts of the first two components of the EIF in (5).
The initial plug-in estimator \(\hat{\theta}_{k,p}^{a,a^*}\) is obtained using the G-computation representation described in Section~3.1, but with the nuisance functions \(\bigl(\lambda_{k}^{a}(t_j) \bigr)_{j\in \{1,\ldots, h\}} \) and \(b_{k,p}^{a,a^*}\) estimated using flexible data-adaptive methods. In practice, the two bias-correction components
$\hat  A_{k,p}^{a,a^*}$ and $\hat  B_{k,p}^{a,a^*}$
can be easily computed once the nuisance parameters are estimated in a federated manner. These aggregated quantities are then transmitted to the target site, where they are combined with the plug-in estimator to obtain $\hat{\theta}_{k,p}^{a,a^*,\text{OS}}$.

\paragraph{Federated TMLE Estimator}
\begin{algorithm}[ht]
\small
\caption{Targeted Updating of the Hazard Function
$\{\hat{\lambda}_{k}^{a}(t_j \mid L,M, R=0)\}_{j=1}^{h}$}
\label{alg:tmle_hazard_update}
\begin{algorithmic}[1]

\Require Data
$\{(\tilde T_i,\delta_i,L_i,M_i,A_i,S_i,R_i)\}_{i=1}^{n}$;
initial estimators
$\{\hat{\lambda}_{k}^{a}(t_j\mid L,M, R=0)\}_{j=1}^{h}$,
$\{\hat g(t_j\mid L,M, R=0)\}_{j=1}^{h}$,
and $\hat{\Omega}_{k,p}^{a,a^\ast}$.

\Ensure Updated hazard estimates
$\{\hat{\lambda}_{k}^{a,(1)}(t_j\mid L,M, R=0)\}_{j=1}^{h}$.

\vspace{0.2cm}
\State \textbf{Step 1: Targeting at the final time point $t_h$}

\State Compute the clever covariate
$
\hat w_i^{(h)}
=
\frac{\hat{\Omega}_{k,p}^{a,a^\ast}}
{\prod_{j < h}
\{1-\hat g(t_j\mid L_i,M_i)\}}.
$

\State Estimate the fluctuation parameter
$\hat\alpha^{(h)}$
from the logistic regression
$$
\text{logit}\{\lambda_k^a(t_h\mid L,M, R=0)\}
=
\text{logit}\{\hat\lambda_k^a(t_h\mid L,M, R=0)\}
+
\alpha^{(h)}\hat w_i^{(h)},
$$
using only observations satisfying
$A_i=a$, $S_i=k$, $R_i=0$ and
$\tilde T_i\ge t_h$.

\State Update the hazard estimate
$
\hat\lambda_k^{a,(1)}(t_h\mid L,M, R=0)
=
\text{expit}
\Bigl(
\text{logit}\{\hat\lambda_k^a(t_h\mid L,M, R=0)\}
+
\hat\alpha^{(h)}\hat w_i^{(h)}
\Bigr).
$

\State Recompute the corresponding survival quantity
$
\hat{\mathcal G}^{(t_{h-1},1)}(\nu\mid L,M, R=0).
$

\vspace{0.2cm}
\State \textbf{Step 2: Backward recursive targeting}

\For{$j=h-1,\ldots,1$}

    \State Compute the clever covariate
    $
    \hat w_i^{(j)}
    =
   \hat{\Omega}_{k,p}^{a,a^\ast}(L_i,M_i)\,
    \hat{\mathcal G}_{a,k}^{(t_j,1)}
    (\nu\mid L_i,M_i, R_i=0).
    $

    \State Estimate the fluctuation parameter
    $\hat\alpha^{(j)}$
    from
    $$
    \text{logit}\{\lambda_k^a(t_j\mid L,M, R=0)\}
    =
    \text{logit}\{\hat\lambda_k^a(t_j\mid L,M, R=0)\}
    +
    \alpha^{(j)}\hat w_i^{(j)}.
    $$

    \State Update
    $
    \hat\lambda_k^{a,(1)}(t_j\mid L,M, R=0)
    =
    \text{expit}
    \Bigl(
    \text{logit}\{\hat\lambda_k^a(t_j\mid L,M, R=0)\}
    +
    \hat\alpha^{(j)}\hat w_i^{(j)}
    \Bigr).
    $

    \State Recompute
    $
    \hat{\mathcal G}^{(t_{j-1},1)}(\nu\mid L,M, R=0).
    $

\EndFor

\Return
$\{\hat{\lambda}_{k}^{a,(1)}(t_j\mid L,M, R=0)\}_{j=1}^{h}$

\end{algorithmic}
\end{algorithm}

The TMLE procedure updates the initial nuisance parameter estimate $\hat{\boldsymbol{\eta}}$ along a least favorable submodel so that the empirical EIF equation is approximately solved. Specifically, the resulting updated estimate, $\hat{\boldsymbol{\eta}}^{\mathrm{TMLE}}$, approximately satisfies $\mathbb{P}_n \left\{\varphi(O;\hat {\boldsymbol{\eta}}^{\mathrm{TMLE}})\right \}=0.$ 
In the present setting, this requires neutralizing both terms $\hat {A}_{k,p}^{a,a^*}$ and $\hat {B}_{k,p}^{a,a^*}$, which can be done in two steps:
\begin{itemize}
    \item Neutralize $\hat {A}_{k,p}^{a,a^*}$, i.e., correct the bias for each $t_d \leq \nu$ induced by 
    $$\sum_{\substack{i:\,A_i=a \\ S_i=k,\ R_i=0}}\hat {\Omega}(L_i,M_i)  I(\tilde{T}_i \geq t_j)\hat {\mathcal{G}}_{a,k}^{(t_d)}(\nu \mid L_i,M_i) \Bigl(I(\tilde{T}_i=t_d , \delta_i=1)  -\hat {\lambda}_{k}^{a}(t_d\mid M_i,L_i)\Bigr).$$
    The hazard models must be targeted in an appropriate order because the clever covariate at time $t_d$ depends on hazard estimates at subsequent time points. Specifically, we  use a backward recursive targeting procedure, updating the hazards from the final time point $\nu $ to the first $t_d = 1$. This construction avoids repeated global iterations because, at each step, $\hat {\mathcal{G}}_{a,k}^{(t_d)}(\nu \mid L_i,M_i)$ and the clever covariate depends only on hazard estimates that have already been targeted \citep{landsiedel2025hazard}. Following this procedure avoids the need for iterative corrections, as each hazard function is updated through its own dedicated targeting step.
    \item  Neutralize $\hat {B}_{k,p}^{a,a^{*}}$, i.e., correct the bias induced by $$\sum_{i:S_{i} =p,A_{i}=a^{*}} \hat{W}_{p}^{a^{*}}(L)  \big(\hat{S}^{(1), a}_{k}(\nu \mid L_i,M_i)-  \hat b_{k,p}^{(1),a,a^{*}}(L_i) \big)$$ where $\hat{S}^{(1),a}_{k}$ and $ \hat b_{k,p}^{(1),a,a^{*}}$ denote the updated estimators of $S^a_k$ and $b_{k,p}^{a,a^*}$ respectively. These updated quantities are computed using the hazard-function estimates obtained after the first targeting step. In practice, the second step can be done via estimating the regression model
    $\text{logit}(b_{k,p}^{a,a^*,(2)})
    =\text{logit}(\hat  b_{k,p}^{a,a^*,(1)}) + \delta \hat  W_p^{a^*},$ using maximum likelihood. By doing so, the associated score equation corresponds to $\hat {B}_{k,p}^{a,a^*}$, thereby eliminating the second empirical bias term.
\end{itemize}
%In the current decentralized setting, the nuisance parameters are estimated in the same federated manner as described above; however, the construction of the final estimator requires an additional targeting steps to ensure efficient and doubly robust estimation. These steps is carried out through several iterative updates. Algorithm~\ref{alg:federated_tmle} summarizes the full estimation procedure for TMLE, which alternates between a federated phase and a local targeting phase.

In a federated setting, the above two-step approach is implemented in five phases (see Algorithm \ref{alg:federated_tmle}). Phase~1 estimates the shared nuisance 
parameters $\hat {\bm{\eta}}^{\mathrm{fed}}$ via FedAvg, exchanging only 
parameter updates across sites. Phase~2 estimates local nuisance parameters 
$\hat {\bm{\eta}}_{k}^{\mathrm{local},(0)}$ at each site, including propensity 
scores, missing data models, initial hazard and censoring functions, which are then broadcast to all sites. Each site subsequently fits the 
cross-site outcome model $\hat b^{a,a^*}_{k,p}$ to form the full initial 
nuisance vector $\hat {\bm{\eta}}_{k,p}^{\mathrm{local},(0)}$ in Phase~3. Then, Phase~4 
applies the local TMLE targeting of Algorithm~\ref{alg:tmle_hazard_update} 
to obtain targeted hazard estimates 
$\hat {\bm{\eta}}_{k,p}^{\mathrm{local},(1)}$, which are broadcast across 
sites. Phase~5 assembles the full nuisance vector 
$\hat {\bm{\eta}}_{k,p}$ for each pair $(k,p)$ computes 
$\hat {\zeta}_{k,p}^{\mathrm{TMLE}}(a)$.

\begin{proposition}[Multiple robustness]
Suppose
that the proposed nuisance estimators $(\hat \lambda_k^a,  \hat b_{k,p}^{a,a^*}, \hat W_p^{a^*}, \hat g_k^a, \hat \Omega_{k,p}^{a,a^*}, \hat S_{k}^{a})$ admit the probability limits $(\tilde\lambda_k^a, \tilde b_{k,p}^{a,a^*}, \tilde W_p^{a^*}, \tilde g_k^a, \tilde\Omega_{k,p}^{a,a^*}, \tilde S_{k}^{a})$.
The One-Step and TMLE estimators,
$\hat {\theta}_{k,p}^{a,a^*,\mathrm{OS}}$ and
$\hat {\theta}_{k,p}^{a,a^*,\mathrm{TMLE}}$, are consistent for
$\theta_{k,p}^{a,a^*}$ if either of the following conditions holds:
\begin{enumerate}
    \item[(i)] $\tilde{\lambda}_k^{a} = \lambda_k^a
    \,\text{and either}\,
    \widetilde{b}_{k,p}^{a,a^*}
    =
    b_{k,p}^{a,a^*}\,\text{or}\,
    \widetilde{W}_p^{a^*} = W_p^{a^*}.$

    \item[(ii)]$\tilde{g}_k^a = g_k^a,
    \,
    \widetilde{\Omega}_{k,p}^{a,a^*}
    =
    \Omega_{k,p}^{a,a^*}
    \,\text{and either }\,
    \widetilde{W}_p^{a^*} = W_p^{a^*}
    \,\text{or }\,
    \widetilde{b}_{k,p}^{a,a^*}
    =
    \mathbb{E}\left[
        \widetilde{S}_{k}^{a}(\nu \mid L,M)
        \mid L,a^*,p,R=0
    \right].$
\end{enumerate}
\end{proposition}

\begin{proposition}[Asymptotic properties]
    Let $\hat\theta_{k,p}^{a,a^*}$ denote either the OS or TMLE estimator of $\theta_{k,p}^{a,a^*}$. Suppose that (i) the positivity condition in Assumption \ref{ass:nie} holds; (ii) the second-order remainder term satisfies $R(\hat{\bm\eta},\bm\eta)=o_P(n^{-1/2})$; and (iii) the class of functions
$
\Bigl\{
\varphi_{k,p}^{a,a^*}(\bm\eta,\theta'):
|\theta'-\theta|<\delta,\;
\|\bm\eta_{k,p}-\tilde{\bm\eta}_{k,p}\|<\delta
\Bigr\}
$
is Donsker for some $\delta>0$, with
$
\mathbb{E} \left [\left(
\varphi_{k,p}^{a,a^*}(\bm\eta_{k,p},\theta')
-
\varphi_{k,p}^{a,a^*}(\tilde{\bm\eta}_{k,p},\theta)
\right)^2\right]
\to 0
$
as $(\bm\eta_{k,p},\theta')\to(\tilde{\bm\eta}_{k,p},\theta)$. Then,
$
\hat\theta_{k,p}^{a,a^*}
=
\theta_{k,p}^{a,a^*}
+
\frac{1}{n}\sum_{i=1}^{n}
\varphi_{k,p}^{a,a^*}(O_i,\bm\eta_{k,p})
+
o_P(n^{-1/2}),
$
which implies $
\sqrt{n}
\bigl(
\hat\theta_{k,p}^{a,a^*}
-
\theta_{k,p}^{a,a^*}
\bigr)
\xrightarrow{D}
N\!\left(
0,\,
\sigma_{k,p}^{a,a^{*}\,2}
\right),
$
where
$
\sigma_{k,p}^{a,a^{*}\,2}
=
V\!\left\{
\varphi_{k,p}^{a,a^*}(O,\bm\eta)
\right\}
$
is the semiparametric efficiency bound.
\end{proposition}

%Condition (ii) holds when the relevant products of nuisance-estimation errors are $o_{\mathbb{P}}(n^{-1/2})$. A sufficient condition is that the nuisance estimators entering each product converge at rates $o_{\mathbb{P}}(n^{-1/4})$.Such rates can be achieved by a variety of data-adaptive procedures, including the LASSO and highly adaptive LASSO under suitable regularity conditions \citep{ benkeser2016highly}.
Condition (ii) can be made explicit using the factorization of the
second-order remainder derived in the Online Supplementary Materials.
Due to the Cauchy--Schwarz
inequality, this condition is satisfied under the following sufficient conditions: \begin{align*}
\left\|
\hat  W_p^{a^*}-W_{0,p}^{a^*}
\right\|_2
\left\{
\left\|
\hat  b_{k,p}^{a,a^*}-b_{0,k,p}^{a,a^*}
\right\|_2
+
\left\|
\hat  S_k^a(\nu)-S_{0,k}^a(\nu)
\right\|_2
\right\}
&=
o_{\mathbb P}(n^{-1/2})\\
\left\|
\hat \Omega_{k,p}^{a,a^*}
-\Omega_{0,k,p}^{a,a^*}
\right\|_2
\left\|
\hat  S_k^a(\nu)-S_{0,k}^a(\nu)
\right\|_2
&=
o_{\mathbb P}(n^{-1/2}),
\\
\sum_{t_d\leq \nu}
\left\|\frac{\prod_{t_d<t_j\leq\nu}
  \{1- \hat \lambda_k^a(t_j\mid L,M)\}}{\widehat G_{0,a}^{k}(t_{d-1})} \right\|_{2} 
\left\|
\widehat G_{a}^{k}(t_{d-1})-G_{0,a}^{k}(t_{d-1})
\right\|_{2}
\left\|
\widehat\lambda_{k}^{a}(t_d)-\lambda_{0,k}^{a}(t_d)
\right\|_{2} &=
o_{\mathbb P}(n^{-1/2}).
\end{align*}
%If the nuisance estimators
%$\widehat{\lambda}_k^a$, $\widehat{g}_k^a$, $b_{k,p}^{a,a^*}, \widehat{W}_k^a$, and
%$\widehat{\Omega}_{k,p}^{a,a^*}$ 
%converge in $L_2(P_0)$ at rate %$o_{\mathbb P}(n^{-1/4})$ and   %$\frac{\prod_{t_d<t_j\leq\nu}
%  \{1- \hat \lambda_k^a(t_j\mid L,M)\}}{\widehat G(t_{d-1})}$ is uniformly bounded, then the condition (ii) is satisfied.

Since $\nu$ is fixed and the number of discrete time points up to $\nu$
is finite, the convergence rates of
$\hat  S_k^a-S_{0,k}^a$ and
$\hat  G_k^a-G_{0,k}^a$ are controlled by those of
$\hat \lambda_k^a-\lambda_{0,k}^a$ and
$\hat  g_k^a-g_{0,k}^a$, respectively. 
Thus, a simple sufficient condition to guarantee $R(\hat{\bm\eta},\bm\eta)=o_P(n^{-1/2})$ is that the nuisance estimators
$\widehat{\lambda}_k^a$, $\widehat{g}_k^a$, $b_{k,p}^{a,a^*}, \widehat{W}_k^a$, and
$\widehat{\Omega}_{k,p}^{a,a^*}$
converge in $L_2(P_0)$ at rate $o_{\mathbb P}(n^{-1/4})$, and $\prod_{t_d<t_j\leq\nu}
  \{1- \hat \lambda_k^a(t_j\mid L,M)\}/\widehat G(t_{d-1})$ is uniformly bounded.

One challenge, however, is that the two secondary quantities $\widehat W_k^a$ and $\widehat\Omega_{k,p}^{a,a^*}$ are ratios of primary nuisance functions. Consequently, their convergence at the $o_{\mathbb P}(n^{-1/4})$ rate does not follow directly from $n^{-1/4}$-convergence of the constituent nuisance estimators, as estimation errors may be amplified when the denominators approach zero. To address this issue, Corollary~1 imposes a stronger positivity assumption which ensures that $R(\hat{\bm \eta},\bm \eta)=o_{\mathbb P}(n^{-1/2})$ when the primary nuisance estimators $\{\lambda_{k}^{a}, g_{k}^{a}, \tau_{k}^{a}, \pi_{k}^{a}, \rho_{k}^{a}, U_k, V_k, b_{k,p}^{a,a^*}\}$ converge to their true counterparts at sufficiently fast rates. A formal proof is provided in the Online Supplementary Materials.

%, the relevant denominator terms are
%uniformly bounded away from zero. This controls the corresponding inverse
%weights and allows the componentwise $o_{\mathbb P}(n^{-1/4})$ rates to be
%transferred to $\widehat W_k^a,$ 
%$\widehat\Omega_{k,p}^{a,a^*}$ and %$\frac{\prod_{t_d<t_j\leq\nu}
%  \{1- \hat \lambda_k^a(t_j\mid L,M)\}}{\widehat G(t_{d-1})}.$ A formal proof is provided in the Online
%Supplementary Materials.

\begin{corollary}
Assume that there exists a constant $c>0$ such that, for all relevant $(s,m,l,a)$ and all $t\le \nu$,
\begin{align*}
\mathbb P(A=a\mid L,S),\;
\mathbb P(A=a\mid M,L,S),\;
\mathbb P(S=s\mid L),\;
\mathbb P(S=s\mid L,M) &\ge c,\\
1-c \ge \mathbb P(C>t\mid A=a,L=l,S=s) \ge c, \quad
\mathbb P(R=0\mid A=a,L=l,S=s) &\ge c,
\end{align*}
Let $0<\varepsilon<c$. For each nuisance estimator
$\widehat\eta^{\,\mathrm{raw}}$ entering $\widehat\Omega_{k,p}^{a,a^*}$ or $W_p^a$, and $g_k^{a}$ 
let $\widehat\eta$ denote its truncated version, with probability estimators
truncated to $[\varepsilon,1]$ (or $[\varepsilon,1-\varepsilon]$ when needed)
and hazard estimators to $[0,1]$. Then, if $\|\widehat\eta-\eta_0\|_2=o_{\mathbb P}(n^{-1/4})$ for every
$\eta\in\boldsymbol\eta$, it follows that
$R(\widehat{\boldsymbol\eta},\boldsymbol\eta_0)=o_{\mathbb P}(n^{-1/2})$.
\end{corollary}

Under standard regularity conditions for federated optimization, including smoothness and strong convexity of the local objective functions and bounded stochastic-gradient variability, and with a fixed number of centers, the learning rate and number of communication rounds can be chosen such that the optimization error is asymptotically negligible relative to the statistical estimation error \citep{li2019convergence}. 

While the $o_{\mathbb{P}}(n^{-1/4})$ convergence-rate requirement imposed on the nuisance estimators can be achieved by many data-adaptive procedures, including LASSO and highly adaptive LASSO under suitable regularity conditions \citep{benkeser2016highly}, the Donsker condition (iii) in Proposition 2 may be restrictive when highly flexible machine-learning algorithms are employed. A common strategy to circumvent this condition is cross-fitting \citep{10.1111/ectj.12097}:  the sample is randomly partitioned into $Q$ approximately equal folds, denoted $D_1,\ldots,D_Q$;  for each fold $q$, nuisance functions are estimated using the training sample $T_q=\{1,\ldots,n\}\setminus D_q$ and subsequently evaluated on the validation sample $D_q$, yielding predictions $\hat{\bm\eta}^{(-q)}$. Cross-fitted OS and TMLE estimators are then obtained by replacing $\hat{\bm\eta}(O_i)$ with the corresponding out-of-sample prediction $\hat{\bm\eta}^{(-q)}(O_i)$ for each observation $i\in D_q$. In a federated setting, the central server first creates the \(Q\) fold partitions and sends the partition assignments to each participating site server. Each site server then uses the received assignments to split its local dataset into \(Q\) folds. The same fold labels are used across all sites, while patient-level identifiers and local data indices remain stored locally and are never transmitted to the central server or to other sites. For each fold $q$, all sites simultaneously use their corresponding local training folds to perform Phases 1, 2, and 3 of Algorithm~\ref{alg:federated_tmle}. The resulting nuisance functions are then evaluated on each site's local held-out fold. Repeating this procedure across all Q folds yields cross-fitted nuisance predictions that are subsequently used to construct the One-Step or TMLE estimator.

As a direct consequence of the above asymptotic expansion~5, the asymptotic variance of $\hat\theta_{k,p}^{a,a^*}$ can be consistently estimated by the empirical variance of the estimated EIF,
$
\hat\sigma_{k,p}^{a,a^{*}\,2}
=
\hat{\mathrm{Var}}
\!\left\{
\varphi_{k,p}^{a,a^*}(O,\hat{\bm\eta}_{k,p})
\right\}.$
An estimator of
$
\zeta_{k,p}$ can then be obtained as
$
\hat\zeta_{k,p}
=
\hat\theta_{k,p}^{a,1}
-
\hat\theta_{k,p}^{a,0},$
with variance estimated by the sample variance of the empirical EIF of $\zeta_{k,p}$, i.e., $  \varphi_{k,p}^{a,1}(O,\hat{\bm\eta}_{k,p}) - \varphi_{k,p}^{a,0}(O,\hat{\bm\eta}_{k,p}).$ In practice, variance estimation is naturally implementable in a federated manner. Specifically, for each observation $i$, define the estimated EIF vector as
$\hat {\boldsymbol{\phi}}_{i}
=
\left(
\hat {\phi}_{k,p}^{a}
\left(
O_{i},
\hat {\bm{\eta}}_{k,p}
\right)
\right)_{(k,p)}
\in \mathbb{R}^{K^2}.$ After model fitting, each site evaluates the corresponding EIF contributions for its own individuals, i.e. $\Sigma_s=\sum_{i : S_i=s}
\hat {\boldsymbol{\phi}}_{i}\hat {\boldsymbol{\phi}}_{i}^{\top} $. 
These site-level summaries can then be combined to obtain the covariance
matrix $\hat \Sigma$ of $(\zeta_{k,p})$, i.e., $\hat \Sigma = n^{-2}\sum_s \Sigma_s$. 
No individual-level observations or individual influence-function values need
to be shared. 
\begin{algorithm}[h]
\small
\caption{Federated TMLE for the transported NIE vector}
\label{alg:federated_tmle}
\begin{algorithmic}[1]
\Require Local datasets at source sites $1,\ldots,K$ and baseline covariate data from the target population $S=0.$

\State \textbf{Phase 1: Federated estimation of shared nuisance parameters: $U$ and $V$ models}

\State \textbf{Phase 2: Local nuisance estimation}

\State \textbf{Phase 3: Cross-site estimation of the $b_{k,p}^{a,a^*}$ models for all $(k,p)$}

    \Statex \textbf{Phase 4: Target parameter estimation}
\For{each site pair $(k,p)$, and treatment $a^*$}
    \State Apply Algorithm~\ref{alg:tmle_hazard_update}
    with inputs $\hat {\bm{\eta}}^{\mathrm{fed}}$
    and $\hat {\bm{\eta}}_{k,p}^{\mathrm{local},(0)}$
    to obtain targeted estimates
    $\bigl\{\hat\lambda_k^{a,(1)}(t_j)\bigr\}_{j=1}^h$.
    \State Broadcast $\hat {\bm{\eta}}_{k,p}^{\mathrm{local},(1)}$
    to all other sites.
    \State Apply TMLE step to correct $b^{a,a^*}_{k,p}$
    with inputs $\hat {\bm{\eta}}^{\mathrm{fed}}$
    and $\hat {\bm{\eta}}_{k,p}^{\mathrm{local},(1)}$
    to obtain targeted estimates $b^{a,a^*, (2)}_{k,p}$
    \State Broadcast $\hat {\bm{\eta}}_{k,p}^{\mathrm{local},(2)}$
    to all other sites.
    \State Assemble
    $\hat {\bm{\eta}}_{k,p}
    =\bigl(\hat {\bm{\eta}}^{\mathrm{fed}},\,
    \hat {\bm{\eta}}_{k,p}^{\mathrm{local},(2)}\bigr)$.
    \State Compute the TMLE estimate
    $\hat {\zeta}_{k,p}^{\mathrm{TMLE}}$
    using $\hat {\bm{\eta}}_{k,p}$.
\EndFor
\Statex \textbf{Phase 5: Variance Estimation}
\For{each site $k \in \{0, \ldots, K\}$}
    
    \State Estimate $ \sum_{S_i=k} \hat \phi_{i}\phi_{i}^{\top}$

\EndFor

\Ensure $\hat{\boldsymbol{\zeta}}^a$, $\hat {\Sigma}$
\end{algorithmic}
\end{algorithm}

\section{Simulations}

We conduct a simulation study that involves six hypothetical populations: one target population and five trial populations. The total sample size is 60,000 individuals, and each simulation experiment is replicated 2,000 times with two follow-up time points:
\noindent
\begin{align*}
&L_1 \sim \mathcal{U}(0,1), 
\qquad 
L_2 \sim \mathcal{B}(0.5), \\
&\mathbb{P}(S=0\mid L_1,L_2)
=
\left[
1+\sum_{s=1}^5 
\exp\left\{(1,L_1,L_2)\bm{\gamma}_s\right\}
\right]^{-1}, \\
&\mathbb{P}(S=s\mid L_1, L_2)
=
\mathbb{P}(S=0\mid L_1, L_2)
\exp\left\{(1, L_1, L_2)\bm{\gamma}_s\right\},
\qquad s=1,\ldots,5\\
&\mathbb{P}(A=1\mid L_1,L_2,S=s)
= 0.5\\
&\mathbb{P}(M=1\mid A,L_1,L_2,S=s)
=
\operatorname{expit}\left(
-1
-2L_1\mathbb{I}(0.2<L_1\leq 0.6)
+L_2
+\omega_s A
\right)\\
&\lambda(t\mid L_1,L_2,M,A,s)
=
\operatorname{expit}\{
[-0.5-0.5I(s>3)]A
-M
+0.5L_2
\\
&\qquad\qquad\qquad\qquad\quad-L_1L_2+1.2\mathbb{I}(0.2<L_1\leq 0.6)
\}\\
&g(t\mid L_1,L_2,M,A,S=s)
=
\operatorname{expit}\left(
0.5t
-t \times M
+0.2L_2
+L_1L_2
\right).
\end{align*}
where:
\begin{align*}
\bm{\gamma}_1
=
\bm{\gamma}_2
=
\bm{\gamma}_3  
& =
(0.15,0.10,0.10)^\top,  \quad
\bm{\gamma}_4
=
\bm{\gamma}_5
=  (0.15,-0.10,-0.10)^\top,
\\
\bm{\omega}_0
=
\bm{\omega}_5
&=1.2, \quad  
\bm{\omega}_1=
\bm{\omega}_3=1.4, \quad
\bm{\omega}_2
=
\bm{\omega}_4=1.8.
\end{align*}

We compared the parametric estimators (G-formula, IPCW (Equation~\eqref{IPW_1_main}) and IPW  (Equation~\eqref{IPW_2_main})) with the semiparametric estimators OS and TMLE. For all parametric estimators, the nuisance functions were estimated using logistic regression or multinomial logistic regression models, with all baseline covariates included as linear main effects. No interaction terms or nonlinear transformations of the covariates were included. The function $b_{k,p}^{a,a^*}$ was estimated using linear regression, using the same set of baseline covariates as linear main effects. For OS and TMLE, we considered two nuisance estimation strategies: parametric models (OS GLM and TMLE GLM) and random forests (OS RF and TMLE RF). Random forests were implemented in Python using scikit-learn, except for $b_{k,p}^{a,a^*}$, which was estimated using honest forest regression implemented in \texttt{EconML}. Transported estimates were then combined by the meta-analytical approach described in Section~\ref{sec:2}, with uniform weighting scheme. 

Relative bias was computed for each $\zeta_{k,p}$ and averaged across the 25 site-pair estimands. Estimated variance was computed for each $\hat\zeta_{k,p}$, and the median across the 25 site pairs was reported. Empirical variance was calculated as the Monte Carlo variance across the 2,000 replications. Coverage was defined as the proportion of confidence intervals containing the true value, and the median coverage across the 25 site-pair estimands was reported. To assess the proposed heterogeneity decomposition, we also reported the median estimate of each heterogeneity component and compared it with its true value.

Results of this simulation study are presented in Table~1. 
Overall, the parametric estimators exhibit substantial bias and poor confidence-interval coverage under model misspecification. In contrast, the OS RF and TMLE RF approaches exhibit satisfactory performance when the nuisance functions are estimated using data-adaptive methods. For the aggregate effect estimator, the data-adaptive OS and TMLE approaches also showed satisfactory finite-sample performance, with 95\% confidence-interval coverage exceeding 90\%. 

\begin{table}[ht] \label{table:sim}
\centering
\caption{\normalsize{Performance of estimators under the considered simulation scenario. 
RF denotes Random Forest-based nuisance estimation and GLM denotes generalized linear model-based nuisance estimation. 
Bias, estimated variance $\hat {\operatorname{Var}}$, empirical variance $\operatorname{Var}$ are averaged over 2,000 Monte Carlo replications per 25 site pairs. Biases are averaged across the 25 site-pair estimands and medians are reported for variances and coverage.}}
\label{table:simus_combined}

\resizebox{\textwidth}{!}{%
\begin{tabular}{llcccc}
\toprule
Target & Estimator 
& Bias (\%) 
& Estimated $\hat {\operatorname{Var}}$
& Empirical $\operatorname{Var}$
& Coverage \\
\midrule

$\hat {\zeta}_{k,p}$ 
& OS RF 
& $4.46\times10^{-1}$ 
& $8.40\times10^{-5}$ 
& $3.50\times10^{-5}$ 
& $92.5\%$ \\

& OS GLM
& $2.71\times10^{1}$ 
& $8.20\times10^{-5}$ 
& $2.80\times10^{-5}$ 
& $12.4\%$ \\

& TMLE GLM 
& $2.90\times10^{1}$ 
& $8.30\times10^{-5}$ 
& $2.90\times10^{-5}$ 
& $12.2\%$ \\

& TMLE RF 
& $2.92\times10^{0}$ 
& $1.14\times10^{-4}$ 
& $3.80\times10^{-5}$ 
& $94.5\%$ \\

& IPCW GLM
& $2.91\times10^{1}$ 
& $3.60\times10^{-5}$ 
& $3.60\times10^{-5}$ 
& $16.7\%$ \\

& IPW GLM
& $2.45\times10^{1}$ 
& $2.70\times10^{-5}$ 
& $2.80\times10^{-5}$ 
& $16.8\%$ \\

& G-formula GLM
& $2.60\times10^{1}$ 
& $2.70\times10^{-5}$ 
& $2.80\times10^{-5}$ 
& $14.4\%$ \\

\midrule

$\hat {\zeta}$ 
& OS RF 
& $4.25\times10^{-1}$ 
& $5.00\times10^{-6}$ 
& -- 
& $91.8\%$ \\

& OS GLM
& $2.71\times10^{1}$ 
& $6.00\times10^{-6}$ 
& -- 
& $0.0\%$ \\

& TMLE RF 
& $2.81\times10^{0}$ 
& $7.00\times10^{-6}$ 
& -- 
& $90.2\%$ \\

& TMLE GLM 
& $2.91\times10^{1}$ 
& $6.00\times10^{-6}$ 
& -- 
& $0.0\%$ \\

& IPCW GLM
& $2.91\times10^{1}$ 
& $7.00\times10^{-6}$ 
& -- 
& $0.0\%$ \\

& IPW GLM
& $2.45\times10^{1}$ 
& $6.00\times10^{-6}$ 
& -- 
& $0.0\%$ \\

& G-formula GLM
& $2.62\times10^{1}$ 
& $6.00\times10^{-6}$ 
& -- 
& $0.0\%$ \\

\bottomrule
\end{tabular}%
}
\end{table}

\section{Application to real-world data}
Using data from the French National Health Data System, which covers 99\% of the French population \citep{MAILLARD2024659}, we conducted a landmark analysis to investigate the role of methotrexate coprescription in explaining the causal effect of TNFi versus IL-12/23 inhibitor therapy on subsequent treatment persistence in psoriasis. Eligible patients were adults with psoriasis ($\geq$ 18 years of age) who initiated one of the two biologic therapies between January 1, 2013, and December 31, 2022. Treatment was coded as $A=1$ for TNFi and $A=0$ for IL-12/23 inhibitor initiation. Mediator was coded as $M=1$ if methotrexate was co-prescribed during the first month after treatment initiation, and $M=0$ otherwise. The event of interest was treatment switching to any alternative therapy or not (measured every month) during a 6-month follow-up. %Patients who experienced treatment switching or were lost to follow-up before the one-month mediator assessment were excluded. Consequently, the analysis pertains to patients who remained event-free and under follow-up through the first month after treatment initiation. 
Baseline confounders included age, sex, Charlson Comorbidity Index (CCI) \citep{charlson1987new}, diagnosis of other inflammatory diseases, diagnosis of chronic diseases not reported in the CCI, and historical drug use. Details on these covariates are available in the Online Supplementary Materials.
The analysis included 19,443 patients across 12 administrative regions in France, with each region treated as a separate local data source. Île-de-France was designated as the target population.

The proposed OS and TMLE approaches were used to transport mediator and outcome information from each site to the target population. For the nuisance parameters, we used random forests to estimate $\lambda_{k}^{a}, g_{k}^{a}$, as well as  $b_{k,p}^{a,a^*}$. For $\rho_{p}^{a^*}, \phi_{k}^{a}, \tau_{k}^{a}$, we used L2-regularized logistic regression. For $U_p$ and $V_k$, L2-regularized multinomial regressions were fitted by FedAvg (Algorithm~1), with softmax link function and no constraints on the number of communication rounds. As the analysis was mostly for illustrative purposes, hyperparameters in all machine learning algorithms were prespecified rather than tuned by cross-validation. More computational details about nuisance parameter estimations are available in the Online Supplementary Materials.
For the federated gradient-descent optimization, we used a fixed learning rate of $10^{-3}$. Finally, standardized indirect effect estimates were combined by the meta-analytic approach proposed in section 2, with uniform weighting scheme. 

Results of this analysis are reported in Table~\ref{tab:data_results}. Both OS and TMLE yielded very small natural indirect effect estimates, suggesting limited mediation through methotrexate use. The estimated heterogeneity decomposition is also highly consistent across the two estimators. The total heterogeneity is approximately 30 times larger than the corresponding sample variance of $\hat \zeta$, indicating substantial heterogeneity. The total heterogeneity is $9.6\times 10^{-6}$ for OS estimators and $9.6\times 10^{-6}$ for TMLE estimators, while the corresponding sample variances are $3.22\times10^{-7}$ and $3.09\times10^{-7}$ respectively. Most of the total heterogeneity is attributable to differences in the outcome-generating mechanism across sites (41.5\% for OS and 45.9\% for TMLE) and to the interaction component (51.88\% for OS and 45.9\% for TMLE). The interaction component accounts for a pronounced proportion of the estimated heterogeneity.
%For population model $U$ and $V$ fitted in decentralized manner, each population is treated as one federated client. Local optimization is performed using full-batch gradient descent, with one gradient steps per communication round. After the local updates, model parameters are aggregated using sample-size-weighted FedAvg. Training stops when the infinity norm of the global gradient falls below `tol = 1e-4`, or when `max_rounds` is reached.

\begin{table}[ht]
\centering
\caption{\normalsize{Aggregate estimates and heterogeneity decomposition for the OS and
TMLE estimators. The 95\% confidence intervals for $\hat{\zeta}$ are
calculated as $\hat{\zeta}\pm1.96 \sqrt{\hat{V}}$. Contributions are calculated
relative to
$\hat{\tau}^{2}
=\hat{\eta}^{2}+\hat{\xi}^{2}
+\hat {\mathrm{Interaction}}$.}}
\resizebox{\textwidth}{!}{%
\begin{tabular}{l|cc|cccc|ccc}
\toprule
 & \multicolumn{2}{c|}{Aggregate}
 & \multicolumn{3}{c|}{Heterogeneity}
 & \multicolumn{3}{c}{Contribution of $\tau^2$ (\%)} \\
\cmidrule(lr){2-3}
\cmidrule(lr){4-6}
\cmidrule(lr){7-9}
 & $\hat{\zeta}$ (95\% CI)
 & $\hat{V}$
 & $\hat{\eta}^2$
 & $\hat{\xi}^2$
 & Interaction
 & $\hat{\tau}^2$
 & Mediator
 & Outcome
 & Interaction \\
\midrule

OS
& $1.061 \times 10^{-3}
   \left[-6.10 \times10^{-6},\,2.22\times 10^{-3}\right]$& $3.22 \times 10^{-7}$
& $6.38\times 10^{-7}$
& $4.00\times 10^{-6}$
& $5.00\times 10^{-6}$
& $9.63 \times 10^{-6}$
& $6.61$
& $41.50$
& $51.88$ \\

TMLE
& $7.00 \times 10^{-4}
   \left[-3.90 \times 10^{-4},\,1.79\times 10^{-3}\right]$
& $3.09 \times 10^{-7}$
& $7.14 \times 10^{-7}$
& $4.00 \times 10^{-6}$
& $4.00 \times 10^{-6}$
& $8.71 \times 10^{-6}$
& $8.19$
& $45.90$
& $45.90$ \\

\bottomrule
\end{tabular}%
}
\label{tab:data_results}
\end{table}

\section{Conclusion}
In this paper, we developed a federated framework for causal meta-mediation analysis with right-censored time-to-event outcomes. Natural indirect effects in a well-defined target population are identified and estimated using a site-by-site strategy, allowing potential heterogeneity in the outcome and mediator mechanisms across sites to be characterized. We further developed semiparametric One-Step and TMLE estimators based on the efficient influence function, along with federated procedures for nuisance estimation, bias correction, and variance estimation.

Several directions warrant further investigation. First, it would be valuable to study the properties of the proposed estimators under more restrictive communication settings, such as when only a limited number of communication rounds are permitted or when participating sites can communicate only once by sharing aggregated summary statistics. Second, our framework primarily focuses on time-independent mediators and discrete time-to-event outcomes. Extending it to continuous survival outcomes and settings with repeatedly measured mediators and time-varying confounders would further broaden the applicability of the proposed methods in practice.

\section{Acknowledgments}
T.T.V. is supported by the French National Research Agency (Agence Nationale de la Recherche) through funding from the Chaires de Professeur Junior program (23R09551S-MEDIATION).

\section{Data Availability Statement}
The data underlying this article are from the French National Health Data System (Système National des Données de Santé, SNDS). Due to legal and data protection requirements, these individual-level data cannot be made publicly available or shared by the authors. Access to SNDS data is subject to authorization under the applicable French regulatory framework. Researchers wishing to access these data may submit an application through the official SNDS access procedure.
\bibliographystyle{abbrvnat}

\bibliography{references}

@inproceedings{liuprivacy,
  title={Privacy-protected causal survival analysis under distribution shift},
  author={Liu, Yi and Levis, Alexander W and Zhu, Ke and Yang, Shu and Gilbert, Peter B and Han, Larry},
  booktitle={The Fourteenth International Conference on Learning Representations (ICLR)},
  year={2026}
}

@article{han2025federated,
  title={Federated adaptive causal estimation (face) of target treatment effects},
  author={Han, Larry and Hou, Jue and Cho, Kelly and Duan, Rui and Cai, Tianxi},
  journal={Journal of the American Statistical Association},
  volume={120},
  number={551},
  pages={1503--1516},
  year={2025},
  publisher={Taylor \& Francis}
}

@article{vo2019novel,
  title={A novel approach for identifying and addressing case-mix heterogeneity in individual participant data meta-analysis},
  author={Vo, Tat-Thang and Porcher, Raphael and Chaimani, Anna and Vansteelandt, Stijn},
  journal={Research synthesis methods},
  volume={10},
  number={4},
  pages={582--596},
  year={2019},
  publisher={Wiley Online Library}
}

@article{van1998functional,
  title={Functional delta method},
  author={Van der Vaart, AW},
  journal={Asymptotic Statistics},
  pages={291--303},
  year={1998},
  publisher={Cambridge University Press Cambridge, United Kingdom}
}

@article{khellaf2025federated,
  title={Federated Causal Inference from Multi-Site Observational Data via Propensity Score Aggregation},
  author={Khellaf, R{\'e}mi and Bellet, Aur{\'e}lien and Josse, Julie},
  year={2025}
}

@article{rochon2014mediation,
  title={Mediation analysis of the relationship between institutional research activity and patient survival},
  author={Rochon, Justine and du Bois, Andreas and Lange, Theis},
  journal={BMC medical research methodology},
  volume={14},
  number={1},
  pages={9},
  year={2014},
  publisher={Springer}
}

@article{vanderweele2016mediation,
  title={Mediation analysis: a practitioner's guide},
  author={VanderWeele, Tyler J},
  journal={Annual review of public health},
  volume={37},
  pages={17--32},
  year={2016},
  publisher={Annual Reviews}
}

@article{MAILLARD2024659,
title = {Use of the French National Health Data System (SNDS) in pharmacoepidemiology: A systematic review in its maturation phase},
journal = {Therapies},
volume = {79},
number = {6},
pages = {659-669},
year = {2024},
issn = {0040-5957},
doi = {https://doi.org/10.1016/j.therap.2024.05.003},
url = {https://www.sciencedirect.com/science/article/pii/S0040595724000659},
author = {Olivier Maillard and René Bun and Moussa Laanani and Amandine Verga-Gérard and Taylor Leroy and Nathalie Gault and Candice Estellat and Pernelle Noize and Florentia Kaguelidou and Agnès Sommet and Maryse Lapeyre-Mestre and Annie Fourrier-Réglat and Alain Weill and Catherine Quantin and Florence Tubach}}

@article{charlson1987new,
  title={A new method of classifying prognostic comorbidity in longitudinal studies: development and validation},
  author={Charlson, Mary E and Pompei, Peter and Ales, Kathy L and MacKenzie, C Ronald},
  journal={Journal of chronic diseases},
  volume={40},
  number={5},
  pages={373--383},
  year={1987},
  publisher={Elsevier}
}

@article{landsiedel2025hazard,
  title={Hazard-Based Targeted Maximum Likelihood Estimation for Survival in Resampling Designs},
  author={Landsiedel, Kirsten E and Phillips, Rachael V and Petersen, Maya L and Van der Laan, Mark J},
  journal={arXiv preprint arXiv:2511.15045},
  year={2025}
}

@inproceedings{mcmahan2017communication,
  title={Communication-efficient learning of deep networks from decentralized data},
  author={McMahan, Brendan and Moore, Eider and Ramage, Daniel and Hampson, Seth and y Arcas, Blaise Aguera},
  booktitle={Artificial intelligence and statistics},
  pages={1273--1282},
  year={2017},
  organization={PMLR}
}

@article{kairouz2021advances,
  title={Advances and Open Problems in Federated Learning},
  author={Kairouz, Peter and McMahan, H. Brendan and Avent, Brendan and others},
  journal={Foundations and Trends in Machine Learning},
  volume={14},
  number={1--2},
  pages={1--210},
  year={2021}
}

@inproceedings{benkeser2016highly,
  title={The highly adaptive lasso estimator},
  author={Benkeser, David and Van der Laan, Mark},
  booktitle={2016 IEEE international conference on data science and advanced analytics (DSAA)},
  pages={689--696},
  year={2016},
  organization={IEEE}
}

@misc{gdpr,
  author = {{European Union}},
  title  = {Regulation (EU) 2016/679: General Data Protection Regulation (GDPR)},
  year   = {2016},
  url    = {https://eur-lex.europa.eu/legal-content/EN/TXT/?uri=CELEX:32016R0679}
}

@misc{hipaa,
  author = {{U.S. HHS}},
  title  = {HIPAA Privacy Rule to Support Reproductive Health Care Privacy},
  year   = {2024},
  url    = {https://www.hhs.gov/hipaa/for-professionals/special-topics/reproductive-health/index.html}
}

@article{li2020federated,
  title={Federated learning: Challenges, methods, and future directions},
  author={Li, Tian and Sahu, Anit Kumar and Talwalkar, Ameet and Smith, Virginia},
  journal={IEEE signal processing magazine},
  volume={37},
  number={3},
  pages={50--60},
  year={2020},
  publisher={IEEE}
}

@article{deng2026survey,
  title={A Survey on Causality with Federated Learning: Challenges, Techniques, and Applications},
  author={Deng, Jing and Chen, Handi and Jiang, Zhihan and Wong, Raymond Chi-Wing and Ngai, Edith CH},
  journal={ACM Transactions on Knowledge Discovery from Data},
  volume={20},
  number={5},
  pages={1--34},
  year={2026},
  publisher={ACM New York, NY}
}

@article{xiong2023federated,
  title={Federated causal inference in heterogeneous observational data},
  author={Xiong, Ruoxuan and Koenecke, Allison and Powell, Michael and Shen, Zhu and Vogelstein, Joshua T and Athey, Susan},
  journal={Statistics in Medicine},
  volume={42},
  number={24},
  pages={4418--4439},
  year={2023},
  publisher={Wiley Online Library}
}

@article{van2006targeted,
  title={Targeted maximum likelihood learning},
  author={Van der Laan, Mark J and Rubin, Daniel},
  year={2006},
  publisher={bepress}
}

@inproceedings{vo2022bayesian,
  title={Bayesian federated estimation of causal effects from observational data},
  author={Vo, Thanh Vinh and Lee, Young and Hoang, Trong Nghia and Leong, Tze-Yun},
  booktitle={Uncertainty in artificial intelligence},
  pages={2024--2034},
  year={2022},
  organization={PMLR}
}

@article{li2019convergence,
  title={On the convergence of fedavg on non-iid data},
  author={Li, Xiang and Huang, Kaixuan and Yang, Wenhao and Wang, Shusen and Zhang, Zhihua},
  journal={arXiv preprint arXiv:1907.02189},
  year={2019}
}

@article{10.1111/ectj.12097,
    author = {Chernozhukov, Victor and Chetverikov, Denis and Demirer, Mert and Duflo, Esther and Hansen, Christian and Newey, Whitney and Robins, James},
    title = {Double/debiased machine learning for treatment and structural parameters},
    journal = {The Econometrics Journal},
    volume = {21},
    number = {1},
    pages = {C1-C68},
    year = {2018},
    month = {02},
    issn = {1368-4221},
    doi = {10.1111/ectj.12097},
    url = {https://doi.org/10.1111/ectj.12097},
    eprint = {https://academic.oup.com/ectj/article-pdf/21/1/C1/27684918/ectj00c1.pdf},
}

@article{jang2026privacy,
  title={Privacy-preserving causal mediation analysis using distributed electronic health record networks},
  author={Jang, Hyojung and Radwan, Rotana and Risk, Malcolm and Lee, Yao and Bian, Jiang and Shi, Xu and Guo, Serena and Zhao, Lili},
  journal={arXiv preprint arXiv:2607.17958},
  year={2026}
}

\appendix
\section{Identification}

\begin{align*}
&\mathbb{E}\left[
    I\{T(a,M(a^{*}))>\nu\}
    \mid S=0
\right]
\\
&\quad=
\mathbb{E}_{L}\left[
    \mathbb{E}\left[
        I\{T(a,M(a^{*}))>\nu\}
        \mid L
    \right]
    \,\middle|\, S=0
\right]
\\
\intertext{\textit{Cross-world independence:}}
&\quad=
\mathbb{E}_{L}\Biggl[
    \sum_{m\in\mathcal{M}}
    \mathbb{E}\left[
        I\{T(a,M(a^{*}))>\nu\}
        \,\middle|\,
        L,M(a^{*})=m
    \right]
\nonumber\\[-0.2em]
&\hspace{12em}\times
    \mathbb{P}\left(
        M(a^{*})=m
        \mid L
    \right)
    \,\Biggm|\, S=0
\Biggr]
\\
\intertext{\textit{Outcome transportability:}}
&\quad=
\mathbb{E}_{L}\Biggl[
    \sum_{m\in\mathcal{M}}
    \mathbb{E}\left[
        I\{T(a,M(a^{*}))>\nu\}
        \,\middle|\,
        \substack{L,M(a^{*})=m,\\ S=k}
    \right]
\nonumber\\[-0.2em]
&\hspace{12em}\times
    \mathbb{P}\left(
        M(a^{*})=m
        \mid L
    \right)
    \,\Biggm|\, S=0
\Biggr]
\\
\intertext{\textit{Conditional outcome exchangeability:}}
&\quad=
\mathbb{E}_{L}\Biggl[
    \sum_{m\in\mathcal{M}}
    \mathbb{E}\left[
        I\{T(a,M(a^{*}))>\nu\}
        \,\middle|\,
        \substack{L,M(a^{*})=m,\\ A=a,S=k}
    \right]
\nonumber\\[-0.2em]
&\hspace{12em}\times
    \mathbb{P}\left(
        M(a^{*})=m
        \mid L
    \right)
    \,\Biggm|\, S=0
\Biggr]
\\
\intertext{\textit{Outcome consistency:}}
&\quad=
\mathbb{E}_{L}\Biggl[
    \sum_{m\in\mathcal{M}}
    \mathbb{E}\left[
        I\{T>\nu\}
        \,\middle|\,
        L,M=m,A=a,S=k
    \right]
\nonumber\\[-0.2em]
&\hspace{12em}\times
    \mathbb{P}\left(
        M(a^{*})=m
        \mid L
    \right)
    \,\Biggm|\, S=0
\Biggr]
\\
\intertext{\textit{Mediator transportability:}}
&\quad=
\mathbb{E}_{L}\Biggl[
    \sum_{m\in\mathcal{M}}
    \mathbb{E}\left[
        I\{T>\nu\}
        \,\middle|\,
        L,M=m,A=a,S=k
    \right]
\nonumber\\[-0.2em]
&\hspace{12em}\times
    \mathbb{P}\left(
        M(a^{*})=m
        \mid L,S=p
    \right)
    \,\Biggm|\, S=0
\Biggr]
\\
\intertext{\textit{Conditional mediator exchangeability:}}
&\quad=
\mathbb{E}_{L}\Biggl[
    \sum_{m\in\mathcal{M}}
    \mathbb{E}\left[
        I\{T>\nu\}
        \,\middle|\,
        L,M=m,A=a,S=k
    \right]
\nonumber\\[-0.2em]
&\hspace{12em}\times
    \mathbb{P}\left(
        M(a^{*})=m
        \mid L,A=a^{*},S=p
    \right)
    \,\Biggm|\, S=0
\Biggr]
\\
\intertext{\textit{Mediator consistency:}}
&\quad=
\mathbb{E}_{L}\Biggl[
    \sum_{m\in\mathcal{M}}
    \mathbb{E}\left[
        I\{T>\nu\}
        \,\middle|\,
        L,M=m,A=a,S=k
    \right]
\nonumber\\[-0.2em]
&\hspace{12em}\times
    \mathbb{P}\left(
        M=m
        \mid L,A=a^{*},S=p
    \right)
    \,\Biggm|\, S=0
\Biggr] \\
&\quad=
    \sum_{(m,l)\in\mathcal{M} \times \mathcal{L}}
    \mathbb{E}\left[
        I\{T>\nu\}
        \,\middle|\,
        L=l,M=m,A=a,S=k
    \right]
\nonumber\\[-0.2em]
&\hspace{12em}\times
    \mathbb{P}\left(
        M=m
        \mid L=l,A=a^{*},S=p
    \right)
    \mathbb{P}\left( L=l\,|\, S=0  \right).
\end{align*}

\section{G-formula}

Let \(C\) denote the censoring time and define
\[
\widetilde{T}=\min(T,C),
\qquad
\delta=I(T\leq C).
\]
Under conditional independent censoring,
\[
T\perp C
\mid A,L,M,S,R=0,
\]
and the corresponding positivity condition, define the censoring hazard
and survival functions as
\begin{align*}
g_k^a(t\mid L,M,R=0)
&=
\mathbb{P}\left(
    C=t
    \mid
    C\geq t,A=a,L,M,S=k,R=0
\right),
\\
G_k^a(t\mid L,M,R=0)
&=
\mathbb{P}\left(
    C>t
    \mid
    A=a,L,M,S=k,R=0
\right)
\nonumber\\
&=
\prod_{u\leq t}
\left\{
    1-g_k^a(u\mid L,M,R=0)
\right\}.
\end{align*}

Similarly, define the event hazard and survival functions as
\begin{align*}
\lambda_k^a(t\mid L,M,R=0)
&=
\mathbb{P}\left(
    T=t
    \mid
    T\geq t,A=a,L,M,S=k,R=0
\right),
\\
S_k^a(t\mid L,M,R=0)
&=
\mathbb{P}\left(
    T>t
    \mid
    A=a,L,M,S=k,R=0
\right)
\nonumber\\
&=
\prod_{u\leq t}
\left\{
    1-\lambda_k^a(u\mid L,M,R=0)
\right\}.
\end{align*}

Because \(\delta=I(T\leq C)\), the following equivalences hold:
\[
\{T=t,C\geq t\}
=
\{\widetilde{T}=t,\delta=1\},
\qquad
\{T\geq t,C\geq t\}
=
\{\widetilde{T}\geq t\},
\]
whereas
\[
\{C=t,T>t\}
=
\{\widetilde{T}=t,\delta=0\},
\]
and
\[
\{C\geq t,T>t\}
=
\{\widetilde{T}>t\}
\cup
\{\widetilde{T}=t,\delta=0\}.
\]

Under Assumption~4, the event hazard can therefore be
expressed in terms of the observed data as follows:
\begin{align}
\lambda_k^a(t\mid L,M,R=0)
&=
\mathbb{P}\left(
    T=t
    \mid
    T\geq t,A=a,L,M,S=k,R=0
\right)
\nonumber\\
&=
\frac{
    \mathbb{P}\left(
        T=t
        \mid
        A=a,L,M,S=k,R=0
    \right)
}{
    \mathbb{P}\left(
        T\geq t
        \mid
        A=a,L,M,S=k,R=0
    \right)
}
\nonumber\\
&=
\frac{
    \mathbb{P}\left(
        T=t
        \mid
        A=a,L,M,S=k,R=0
    \right)
    \mathbb{P}\left(
        C\geq t
        \mid
        A=a,L,M,S=k,R=0
    \right)
}{
    \mathbb{P}\left(
        T\geq t
        \mid
        A=a,L,M,S=k,R=0
    \right)
    \mathbb{P}\left(
        C\geq t
        \mid
        A=a,L,M,S=k,R=0
    \right)
}
\nonumber\\
&=
\frac{
    \mathbb{P}\left(
        T=t,C\geq t
        \mid
        A=a,L,M,S=k,R=0
    \right)
}{
    \mathbb{P}\left(
        T\geq t,C\geq t
        \mid
        A=a,L,M,S=k,R=0
    \right)
}
\nonumber\\
&=
\mathbb{P}\left(
    T=t,C\geq t
    \mid
    T\geq t,C\geq t,
    A=a,L,M,S=k,R=0
\right)
\nonumber\\
&=
\mathbb{P}\left(
    \widetilde{T}=t,\delta=1
    \mid
    \widetilde{T}\geq t,
    A=a,L,M,S=k,R=0
\right).
\label{eq:event-hazard-observed}
\end{align}

The third equality follows from conditional independent censoring, which
allows the joint probabilities involving \(T\) and \(C\) to be
factorized conditionally on \(A,L,M,S,R=0\).

The censoring hazard can be expressed in terms of the observed data in
the same way. Since censoring at time \(t\) is observed only when
\(C=t\) and \(T>t\), we obtain
\begin{align}
g_k^a(t\mid L,M,R=0)
&=
\mathbb{P}\left(
    C=t
    \mid
    C\geq t,A=a,L,M,S=k,R=0
\right)
\nonumber\\
&=
\frac{
    \mathbb{P}\left(
        C=t
        \mid
        A=a,L,M,S=k,R=0
    \right)
}{
    \mathbb{P}\left(
        C\geq t
        \mid
        A=a,L,M,S=k,R=0
    \right)
}
\nonumber\\
&=
\frac{
    \mathbb{P}\left(
        C=t
        \mid
        A=a,L,M,S=k,R=0
    \right)
    \mathbb{P}\left(
        T>t
        \mid
        A=a,L,M,S=k,R=0
    \right)
}{
    \mathbb{P}\left(
        C\geq t
        \mid
        A=a,L,M,S=k,R=0
    \right)
    \mathbb{P}\left(
        T>t
        \mid
        A=a,L,M,S=k,R=0
    \right)
}
\nonumber\\
&=
\frac{
    \mathbb{P}\left(
        C=t,T>t
        \mid
        A=a,L,M,S=k,R=0
    \right)
}{
    \mathbb{P}\left(
        C\geq t,T>t
        \mid
        A=a,L,M,S=k,R=0
    \right)
}
\nonumber\\
&=
\mathbb{P}\left(
    C=t,T>t
    \mid
    C\geq t,T>t,
    A=a,L,M,S=k,R=0
\right)
\nonumber\\
&=
\mathbb{P}\left(
    \widetilde{T}=t,\delta=0
    \,\middle|\,
    \{\widetilde{T}>t\}
    \cup
    \{\widetilde{T}=t,\delta=0\},
    A=a,L,M,S=k,R=0
\right).
\label{eq:censoring-hazard-observed}
\end{align}

Indeed, conditional independent censoring implies
\begin{align*}
&\frac{
    \mathbb{P}\left(
        C=t,T>t
        \mid
        A=a,L,M,S=k,R=0
    \right)
}{
    \mathbb{P}\left(
        C\geq t,T>t
        \mid
        A=a,L,M,S=k,R=0
    \right)
}
\nonumber\\
&\quad =
\frac{
    \mathbb{P}\left(
        C=t
        \mid
        A=a,L,M,S=k,R=0
    \right)
    \mathbb{P}\left(
        T>t
        \mid
        A=a,L,M,S=k,R=0
    \right)
}{
    \mathbb{P}\left(
        C\geq t
        \mid
        A=a,L,M,S=k,R=0
    \right)
    \mathbb{P}\left(
        T>t
        \mid
        A=a,L,M,S=k,R=0
    \right)
}
\nonumber\\
&\quad =
\frac{
    \mathbb{P}\left(
        C=t
        \mid
        A=a,L,M,S=k,R=0
    \right)
}{
    \mathbb{P}\left(
        C\geq t
        \mid
        A=a,L,M,S=k,R=0
    \right)
}.
\end{align*}

Consequently, the event and censoring hazards are identified from the
observed variables \((\widetilde{T},\delta)\) as
\begin{align*}
\lambda_k^a(t\mid L,M,R=0)
&=
\mathbb{P}\left(
    \widetilde{T}=t,\delta=1
    \mid
    \widetilde{T}\geq t,
    A=a,L,M,S=k,R=0
\right),
\\
g_k^a(t\mid L,M,R=0)
&=
\mathbb{P}\left(
    \widetilde{T}=t,\delta=0
    \,\middle|\,
    \{\widetilde{T}>t\}
    \cup
    \{\widetilde{T}=t,\delta=0\},
    A=a,L,M,S=k,R=0
\right).
\end{align*}

Under the missing-at-random assumption, conditional on
\((A,S,L,M)\), the distribution of the event and censoring times does
not depend on \(R\). Therefore, these quantities are identified using
the observed data from individuals with \(R=0\).

The G-formula is then given by
\begin{align}
\theta_{k,p}^{a,a^{*}}
&=
\sum_{(m,l)\in\mathcal{M}\times\mathcal{L}}
S_k^a(\nu\mid l,m,R=0)
\,
\mathbb{P}\left(
    M=m
    \mid
    L=l,R=0,A=a^{*},S=p
\right)
\mathbb{P}\left(
    L=l
    \mid
    S=0
\right).
\end{align}

\section{IPCW and IPW formulations}

\subsection{IPCW}

Let \(S\) denote the population or data-source indicator, with \(j\), \(k\), and \(p\) denoting three possible values. Under the identification assumptions stated above,
\begin{align}
\theta_{k,p}^{a,a^{*}}(P)
&=
\mathbb{E}\left[
    I\{T(a,M(a^{*}))>\nu\}
    \mid S=0
\right]
\nonumber\\
&=
\sum_{(\ell,m)\in\mathcal{L}\times\mathcal{M}}
\mathbb{E}\left[
    I\{T>\nu\}
    \mid
    L=\ell,M=m,A=a,S=k,R=0
\right]
\nonumber\\[-0.2em]
&\qquad\times
\mathbb{P}(M=m\mid L=\ell,A=a^{*},S=p,R=0)
\mathbb{P}(L=\ell\mid S=0).
\label{eq:gformula-ipcw}
\end{align}

For every \((\ell,m)\),
\begin{align}
&\mathbb{E}\left[
    I\{T>\nu\}
    \mid
    L=\ell,M=m,A=a,S=k,R=0
\right]
\nonumber\\
&\quad=
\frac{
    \mathbb{E}\left[
        I\{T>\nu\}
        I\{M=m,A=a,S=k,R=0\}
        \mid L=\ell
    \right]
}{
    \mathbb{P}(M=m\mid L=\ell,A=a,S=k,R=0)
}
\nonumber\\[-0.2em]
&\qquad\times
\frac{1}{
    \mathbb{P}(R=0\mid L=\ell,A=a,S=k)
    \mathbb{P}(A=a\mid L=\ell,S=k)
    \mathbb{P}(S=k\mid L=\ell)
}.
\label{eq:conditional-mean-ipcw}
\end{align}

Substituting \eqref{eq:conditional-mean-ipcw} into
\eqref{eq:gformula-ipcw} and using
\[
\mathbb{P}(L=\ell\mid S=0)
=
\frac{
    \mathbb{P}(S=0\mid L=\ell)\mathbb{P}(L=\ell)
}{
    \mathbb{P}(S=0)
},
\]
gives
\begin{align}
\theta_{k,p}^{a,a^{*}}(P)
&=
\frac{1}{\mathbb{P}(S=0)}
\mathbb{E}\Biggl[
    I\{T>\nu\}I\{A=a,S=k,R=0\}
\nonumber\\[-0.2em]
&\qquad\times
\frac{
    \mathbb{P}(M\mid L,A=a^{*},S=p,R=0)
}{
    \mathbb{P}(M\mid L,A=a,S=k,R=0)
}
\nonumber\\[-0.2em]
&\qquad\times
\frac{
    \mathbb{P}(S=0\mid L)
}{
    \mathbb{P}(R=0\mid L,A=a,S=k)
    \mathbb{P}(A=a\mid L,S=k)
    \mathbb{P}(S=k\mid L)
}
\Biggr].
\label{eq:initial-ipcw}
\end{align}

The mediator distribution can be rewritten using Bayes' rule. For
\(s\in\{k,p\}\) and any treatment value \(a'\),
\begin{align}
\mathbb{P}(M\mid L,A=a',S=s,R=0)
&=
\frac{
    \mathbb{P}(A=a'\mid L,M,S=s,R=0)
    \mathbb{P}(S=s\mid L,M,R=0)
}{
    \mathbb{P}(A=a'\mid L,S=s,R=0)
    \mathbb{P}(S=s\mid L,R=0)
}
\nonumber\\[-0.2em]
&\qquad\times
\mathbb{P}(M\mid L,R=0).
\label{eq:mediator-bayes}
\end{align}

Therefore,
\begin{align}
&\frac{
    \mathbb{P}(M\mid L,A=a^{*},S=p,R=0)
}{
    \mathbb{P}(M\mid L,A=a,S=k,R=0)
}
\nonumber\\
&\quad=
\frac{
    \mathbb{P}(A=a^{*}\mid L,M,S=p,R=0)
    \mathbb{P}(S=p\mid L,M,R=0)
}{
    \mathbb{P}(A=a\mid L,M,S=k,R=0)
    \mathbb{P}(S=k\mid L,M,R=0)
}
\nonumber\\[-0.2em]
&\qquad\times
\frac{
    \mathbb{P}(A=a\mid L,S=k,R=0)
    \mathbb{P}(S=k\mid L,R=0)
}{
    \mathbb{P}(A=a^{*}\mid L,S=p,R=0)
    \mathbb{P}(S=p\mid L,R=0)
}.
\label{eq:mediator-ratio-ipcw}
\end{align}

Moreover,
\begin{align}
&\mathbb{P}(A=a\mid L,S=k,R=0)
\mathbb{P}(S=k\mid L,R=0)
\nonumber\\
&\qquad=
\frac{
    \mathbb{P}(R=0\mid L,A=a,S=k)
    \mathbb{P}(A=a\mid L,S=k)
    \mathbb{P}(S=k\mid L)
}{
    \mathbb{P}(R=0\mid L)
},
\label{eq:bayes-k}
\\
&\mathbb{P}(A=a^{*}\mid L,S=p,R=0)
\mathbb{P}(S=p\mid L,R=0)
\nonumber\\
&\qquad=
\frac{
    \mathbb{P}(R=0\mid L,A=a^{*},S=p)
    \mathbb{P}(A=a^{*}\mid L,S=p)
    \mathbb{P}(S=p\mid L)
}{
    \mathbb{P}(R=0\mid L)
}.
\label{eq:bayes-p}
\end{align}

Hence,
\begin{align}
&\frac{
    \mathbb{P}(A=a\mid L,S=k,R=0)
    \mathbb{P}(S=k\mid L,R=0)
}{
    \mathbb{P}(A=a^{*}\mid L,S=p,R=0)
    \mathbb{P}(S=p\mid L,R=0)
}
\nonumber\\
&\quad=
\frac{
    \mathbb{P}(R=0\mid L,A=a,S=k)
    \mathbb{P}(A=a\mid L,S=k)
    \mathbb{P}(S=k\mid L)
}{
    \mathbb{P}(R=0\mid L,A=a^{*},S=p)
    \mathbb{P}(A=a^{*}\mid L,S=p)
    \mathbb{P}(S=p\mid L)
}.
\label{eq:bayes-ratio-ipcw}
\end{align}

Substituting \eqref{eq:mediator-ratio-ipcw} and
\eqref{eq:bayes-ratio-ipcw} into \eqref{eq:initial-ipcw}, the terms
\[
\mathbb{P}(R=0\mid L,A=a,S=k)
\mathbb{P}(A=a\mid L,S=k)
\mathbb{P}(S=k\mid L)
\]
cancel. It follows that
\begin{align}
\theta_{k,p}^{a,a^{*}}(P)
&=
\frac{1}{\mathbb{P}(S=0)}
\mathbb{E}\Biggl[
    I\{T>\nu\}I\{A=a,S=k,R=0\}
\nonumber\\[-0.2em]
&\qquad\times
\frac{
    \mathbb{P}(A=a^{*}\mid L,M,S=p,R=0)
    \mathbb{P}(S=p\mid L,M,R=0)
    \mathbb{P}(S=0\mid L)
}{
    \mathbb{P}(A=a\mid L,M,S=k,R=0)
    \mathbb{P}(S=k\mid L,M,R=0)
}
\nonumber\\[-0.2em]
&\qquad\times
\frac{1}{
    \mathbb{P}(A=a^{*}\mid L,S=p)
    \mathbb{P}(S=p\mid L)
    \mathbb{P}(R=0\mid L,A=a^{*},S=p)
}
\Biggr].
\label{eq:final-ipcw}
\end{align}
Define
\begin{align}
\Omega_{k,p}^{a,a^{*}}(L,M)
&=
\frac{
    \mathbb{P}(A=a^{*}\mid L,M,S=p,R=0)
    \mathbb{P}(S=p\mid L,M,R=0)
    \mathbb{P}(S=0\mid L)
}{
    \mathbb{P}(A=a\mid L,M,S=k,R=0)
    \mathbb{P}(S=k\mid L,M,R=0)
}
\nonumber\\[-0.2em]
&\quad\times
\frac{1}{
    \mathbb{P}(A=a^{*}\mid L,S=p)
    \mathbb{P}(S=p\mid L)
    \mathbb{P}(R=0\mid L,A=a^{*},S=p)
}.
\label{eq:omega-ipcw}
\end{align}

Since
\[
I\{\widetilde{T}>\nu\}
=
I\{T>\nu\}I\{C>\nu\},
\]
under the conditional independent censoring assumption
\[
T\perp C
\mid L,M,A=a,S=k,R=0,
\]
we have
\begin{align*}
&\mathbb{E}\left[
    \left.
    \frac{I\{\widetilde{T}>\nu\}}
    {G_k^a(\nu\mid L,M,R=0)}
    \right\|_2
    L,M,A=a,S=k,R=0
\right]
\nonumber\\
&\quad=
\frac{
    \mathbb{P}(T>\nu,C>\nu
    \mid L,M,A=a,S=k,R=0)
}{
    G_k^a(\nu\mid L,M,R=0)
}
\nonumber\\
&\quad=
\frac{
    \mathbb{P}(T>\nu\mid L,M,A=a,S=k,R=0)
    \mathbb{P}(C>\nu\mid L,M,A=a,S=k,R=0)
}{
    G_k^a(\nu\mid L,M,R=0)
}
\nonumber\\
&\quad=
\mathbb{P}(T>\nu\mid L,M,A=a,S=k,R=0)
\nonumber\\
&\quad=
\mathbb{E}\left[
    I\{T>\nu\}
    \mid L,M,A=a,S=k,R=0
\right].\nonumber
\end{align*}

Consequently, the IPCW representation is
\begin{align*}
\theta_{k,p}^{a,a^{*}}(P)
&=
\frac{1}{\mathbb{P}(S=0)}
\mathbb{E}\Biggl[
    \frac{
        \Omega_{k,p}^{a,a^{*}}(L,M)
    }{
        G_k^a(\nu\mid L,M,R=0)
    }
\nonumber\\[-0.2em]
&\hspace{7em}\times
    I\{\widetilde{T}>\nu\}
    I\{A=a,S=k,R=0\}
\Biggr].\nonumber
\end{align*}

Then
\begin{align}
\theta_{k,p}^{a,a^{*}}(P)
&=
\frac{1}{\mathbb{P}(S=0)}
\mathbb{E}\Biggl[
    \frac{
        \Omega_{k,p}^{a,a^{*}}(L,M)
    }{
        G_k^a(\nu\mid L,M,R=0)
    }
    I\{\widetilde{T}>\nu\}
    I\{A=a,S=k,R=0\}
\Biggr].
\label{IPW_1}
\end{align}

\subsection{IPW}

Under the identification assumptions stated above,
\begin{align*}
\theta_{k,p}^{a,a^{*}}(P)
&=
\mathbb{E}\left[
    I\{T(a,M(a^{*}))>\nu\}
    \mid S=0
\right]
\nonumber\\
&=
\mathbb{E}_{L}\Biggl[
    \mathbb{E}_{M}\left[
         S_k^a( \nu \mid L,M,R=0)
        \mid L,A=a^{*},S=p,R=0
    \right]
    \,\bigg|\, S=0
\Biggr].
\end{align*}

Define
\[
b_{k,p}^{a,a^{*}}(\nu\mid L)
=
\mathbb{E}\left[
     S_k^a( \nu \mid L,M,R=0)
    \mid L,A=a^{*},S=p,R=0
\right].
\]
Then
\begin{align}
\theta_{k,p}^{a,a^{*}}(P)
&=
\frac{1}{\mathbb{P}(S=0)}
\mathbb{E}\left[
    b_{k,p}^{a,a^{*}}(\nu\mid L)
    \mathbb{P}(S=0\mid L)
\right]
\nonumber\\
&=
\frac{1}{\mathbb{P}(S=0)}
\mathbb{E}\Biggl[
    \frac{
        I\{A=a^{*},S=p,R=0\}
        \mathbb{P}(S=0\mid L)
    }{
        \mathbb{P}(A=a^{*}\mid L,S=p,R=0)\,
        \mathbb{P}(R=0\mid L,S=p)\,
        \mathbb{P}(S=p\mid L)
    }
\nonumber\\[-0.2em]
&\hspace{11em}\times
     S_k^a( \nu \mid L,M,R=0)
\Biggr] \nonumber\\
&= \frac{1}{\mathbb{P}(S=0)}
\mathbb{E}\Biggl[I\{A=a^{*},S=p,R=0\} W_{k}^{a}(L) \,
     S_k^a( \nu \mid L,M,R=0)
\Biggr].\label{IPW_2}
\end{align}

with $W_{k}^{a}(L) = \displaystyle \frac{
        \mathbb{P}(S=0\mid L)
    }{\mathbb{P}(A=a^{*}\mid L,S=p)\,\mathbb{P}(S=p\mid L) \mathbb{P}(R=0\mid L,p,a^*) }$

\section{Parametric Estimation}

$\psi_{\theta,i}^{a,a^*}$ denotes the estimating function for the target parameter $\theta_{k,p}^{a,a^*}$, and the remaining components correspond to the local, federated, and mixed nuisance parameters. The indicator functions ensure that the local and mixed nuisance estimating equations are evaluated only in studies $S=k$ and $S=p$, respectively. The estimating equation for $\theta_{k,p}^{a,a^*}$ is evaluated according to the estimator representation: IPW representations (\ref{IPW_1}) and (\ref{IPW_2}) or the G-formula representation (\ref{G-formula1}). The federated nuisance estimating equations combine study-specific contributions across all participating studies.

The first component, $\psi_{\theta,i}^{a,a^*}$, is the estimating equation defining the target parameter. For the weighting representation (\ref{IPW_1}), it takes the form
\[
\psi_{\theta,i}^{a,a^*}(\bm{\vartheta})
= \frac{n_0}{n} \biggl \{
\frac{\Omega_{k,p}^{a,a^*}(L_i,M_i)}
{G_k^a(\nu\mid L_i,M_i,R_i=0)}
I(\tilde T_i>\nu,a,k,R_i=0)
-\theta_{k,p}^{a,a^*}\biggr\},
\]
whereas for the alternative representation (\ref{IPW_2}) it can be written as
\[
\psi_{\theta,i}^{a,a^*}(\bm{\vartheta})
=\frac{n_0}{n} \biggl \{
W_p^{a^*}(L_i)
I(a^*,p,R_i=0)
S_k^a(\nu\mid L_i,M_i,R_i=0)
-\theta_{k,p}^{a,a^*} \biggr \}.
\]

For the G-formula representation,
\[
\psi_{\theta,i}^{a,a^*}(\bm{\vartheta})
= \frac{n_0}{n} \biggl \{
I(S_i=0)
b_{k,p}^{a,a^*}(L_i)
-\theta_{k,p}^{a,a^*} \biggr \}.
\]

The second component,
$\bm{\psi}^{\mathrm{local}}_{k,p,i}$,
contains the score equations associated with the nuisance models estimated entirely within study $k$:
\[
\bm{\psi}^{\mathrm{local}}_{k,p,i}
=
\left(
\psi_{\pi,i},
\psi_{\tau,i},
\psi_{\rho,i},
\psi_{\lambda,i},
\psi_{g,i}
\right)^{\top},
\]
where each element corresponds to the score of the parametric model used to estimate
$\pi_k^a$,
$\tau_k^a$,
$\rho_k^a$,
$\lambda_k^a$,
and
$g_k^a$.
These estimating equations involve only observations from study $k$.

The third component,
$\bm{\psi}^{\mathrm{fed}}_{k,p,i}$,
contains the estimating equations for the federated nuisance parameters:
\[
\bm{\psi}^{\mathrm{fed}}_{k,p,i}
=
\left(
\psi_{U,i},
\psi_{V,i}
\right)^{\top},
\]
where $\psi_{U,i}$ and $\psi_{V,i}$ are the score contributions for the multinomial model defining $U_p(L)$ and the model defining $V_k(L,M)$, respectively. Parameters are optimized using a federated learning algorithm.

The final component,
$\bm{\psi}^{\mathrm{mix}}_{k,p,i}$,
corresponds to the estimating equations for the sequentially estimated parameter
$b_{k,p}^{a,a^*}$.
If
\[
b_{k,p}^{a,a^*}(L)=m(L;\beta_b),
\]
then
\[
\bm{\psi}^{\mathrm{mix}}_{k,p,i}
=
\frac{\partial m(L_i;\beta_b)}
{\partial\beta_b}
\left\{
\widehat S_k^a(\nu\mid L_i,M_i)
-
m(L_i;\beta_b)
\right\},
\]
which is evaluated only within study $p$.

\begin{algorithm}[H]
\small
\caption{Federated parametric estimation of $\theta_{k,p}^{a,a^*}$}
\label{alg:federated_ipw}
\begin{algorithmic}[1]
\Require Local datasets at sites $k=1,\ldots,K$ and target data at site $0$.

\Statex
\textbf{Phase 1: Federated estimation of shared nuisance parameter}

\If{representation (\ref{IPW_1}) or (\ref{IPW_2})}

    \State Obtain
    $
    \widehat{\bm{\eta}}^{\mathrm{fed}}
    =
    (\hat U_0,\hat U_1,\ldots,\hat U_K,\hat V_1,\ldots,\hat V_K)
    $
    by fitting the parametric models for
    $U_s(L)$
    and
    $V_s(L,M)$
    using a federated learning algorithm.

\EndIf

\Statex
\textbf{Phase 2: Local nuisance estimation}

\For{each site $k$}

    \State Estimate the local nuisance parameters
   $
    \widehat{\bm{\eta}}_k^{\mathrm{local}}
    =
    \bigl(
    \hat\pi_k^a,\,
    \hat\tau_k^a,\,
    \hat\rho_k^a,\,
    \{\hat\lambda_k^a(t_m)\}_{m=1}^h,\,
    \{\hat g_k^a(t_m)\}_{m=1}^h
    \bigr)
   $ using local data.

    \State Broadcast parameters of
    $\widehat{\bm{\eta}}_k^{\mathrm{local}}$  to all other sites
\EndFor

\Statex
\textbf{Phase 3: Cross-site nuisance assembly}

\For{each pair $(k,p)$}

    \If{G-formula representation} Fit the parametric model $b_{k,p}^{a,a^*}(L;\beta_b)$ and obtain $\hat b_{k,p}^{a,a^*}$.
    \EndIf

\EndFor

\Statex
\textbf{Phase 4: Estimation of $\theta_{k,p}^{a,a^*}$}
\For{each pair $(k,p)$} 
    \If{IPW (\ref{IPW_1})} site $k$ computes $\widehat{\theta}_{k,p}^{a,a^*}$ 
    \ElsIf{IPW (\ref{IPW_2})} then site $p$ computes $\widehat{\theta}_{k,p}^{a,a^*}$; 
    \ElsIf{G-formula} then target site computes $\widehat{\theta}_{k,p}^{a,a^*}$. 
    \EndIf
\EndFor

\Statex
\textbf{Phase 5: Variance estimation}

\For{each site $s$}

    \State Compute the local estimating-function contributions $\hat A_s,
    \hat B_s.$ and transmit to the target server.

\EndFor

\State Aggregate
$\hat A
=
\sum_s \hat A_s,
\hat B
=
\sum_s \hat B_s,$
and compute $\widehat{\mathrm{Var}}(\hat{\bm{\vartheta}})
=
\hat A^{-1}
\hat B
\hat A^{-\top}.$

\Ensure
$\widehat{\theta}_{k,p}^{a,a^*}$,
$\widehat{\mathrm{Var}}(\widehat{\theta}_{k,p}^{a,a^*})$,
and associated confidence intervals.

\end{algorithmic}
\end{algorithm}

Variance can alternatively be estimated using a nonparametric bootstrap, which is valid for M-estimators under standard regularity conditions. In the federated setting, bootstrap resampling can be implemented without sharing individual-level data. Specifically, the target population site generates bootstrap samples by resampling subject indices and communicates the corresponding bootstrap indices to the participating sites. Each site then reconstructs the corresponding bootstrap sample locally, re-estimates the required nuisance parameters, and recomputes the federated estimator $\hat{\theta}_{k,p}^{a,a^*}$. Repeating this procedure over $B$ bootstrap replicates yields the empirical bootstrap distribution of the estimator, from which the variance and confidence intervals can be obtained.

\section{Influence function}

By the product rule, the pathwise derivative decomposes as
\[
\dot{\theta}(P)=A+B+C,
\]
where
\begin{align}
A
&=
\sum_{(\ell,m)\in\mathcal{L}\times\mathcal{M}}
\dot S_{a,k}(\nu\mid \ell,m,R=0)\,
q_{p}^{a^{*}}(m\mid \ell)\,
r_0(\ell),
\label{eq:if-component-a}
\\
B
&=
\sum_{(\ell,m)\in\mathcal{L}\times\mathcal{M}}
S_{a,k}(\nu\mid \ell,m,R=0)\,
\dot q_{p}^{a^{*}}(m\mid \ell)\,
r_0(\ell),
\label{eq:if-component-b}
\\
C
&=
\sum_{(\ell,m)\in\mathcal{L}\times\mathcal{M}}
S_{a,k}(\nu\mid \ell,m,R=0)\,
q_{p}^{a^{*}}(m\mid \ell)\,
\dot r_0(\ell).
\label{eq:if-component-c}
\end{align}

We derive each term separately.

%%%%%%%%%%%%%%%%%%%%%%%%%%%%%%%%%%%%%%%%%%%%%%%%%%%%%%%%%%%%
\paragraph*{Contribution of the mediator distribution: $B. $}

For fixed $(\ell,m)$, the conditional probability
$q_{p}^{a^{*}}(m\mid \ell)$ has influence-function contribution
\[
\frac{
    I(L=\ell,A=a^{*},S=p,R=0)
}{
    \mathbb{P}(L=\ell,A=a^{*},S=p,R=0)
}
\left\{
    I(M=m)-q_{p}^{a^{*}}(m\mid \ell)
\right\}.
\]

Hence,
\begin{align*}
B
&=
\sum_{(\ell,m)\in\mathcal{L}\times\mathcal{M}}
S_{a,k}(\nu\mid \ell,m,R=0)
\frac{
    I(L=\ell,A=a^{*},S=p,R=0)
}{
    \mathbb{P}(L=\ell,A=a^{*},S=p,R=0)
}
\nonumber\\
&\qquad\qquad\qquad\times
\left\{
    I(M=m)-q_{p}^{a^{*}}(m\mid \ell)
\right\}
r_0(\ell).
\end{align*}

Summing over $(\ell,m)$ yields
\begin{align*}
B
&=
\frac{1}{\mathbb{P}(S=0)}
\frac{
    I(A=a^{*},S=p,R=0)
}{
    \mathbb{P}(A=a^{*},S=p,R=0\mid L)
}
\mathbb{P}(S=0\mid L)
\nonumber\\
&\qquad\times
\Bigg[
S_{a,k}(\nu\mid L,M,R=0)
-
\mathbb{E}_{M}
\left\{
S_{a,k}(\nu\mid L,M,R=0)
\mid L,A=a^{*},S=p,R=0
\right\}
\Bigg].
\end{align*}

Equivalently, using
\[
\mathbb{P}(A=a^{*},S=p,R=0\mid L)
=
\mathbb{P}(A=a^{*}\mid L,S=p,R=0)
\mathbb{P}(S=p,R=0\mid L),
\]
we may write
\begin{align*}
B
&=
\frac{1}{\mathbb{P}(S=0)}
\frac{
    I(A=a^{*},S=p,R=0)
}{
    \mathbb{P}(A=a^{*}\mid L,S=p,R=0)
}
\frac{
    \mathbb{P}(S=0\mid L)
}{
    \mathbb{P}(S=p,R=0\mid L)
}
\nonumber\\
&\qquad\times
\Bigg[
S_{a,k}(\nu\mid L,M,R=0)
-
\mathbb{E}_{M}
\left\{
S_{a,k}(\nu\mid L,M,R=0)
\mid L,A=a^{*},S=p,R=0
\right\}
\Bigg].
\end{align*}

%%%%%%%%%%%%%%%%%%%%%%%%%%%%%%%%%%%%%%%%%%%%%%%%%%%%%%%%%%%%
\paragraph*{Contribution of the target covariate distribution: $C. $}

For $r_0(\ell)=\mathbb{P}(L=\ell\mid S=0)$, the corresponding
influence-function contribution is
\[
\frac{I(S=0)}{\mathbb{P}(S=0)}
\left\{
I(L=\ell)-r_0(\ell)
\right\}.
\]

Therefore,
\begin{align*}
C
&=
\sum_{(\ell,m)\in\mathcal{L}\times\mathcal{M}}
S_{a,k}(\nu\mid \ell,m,R=0)
q_{p}^{a^{*}}(m\mid \ell)
\frac{I(S=0)}{\mathbb{P}(S=0)}
\left\{
I(L=\ell)-r_0(\ell)
\right\}.
\end{align*}

Thus,
\begin{align*}
C
&=
\frac{I(S=0)}{\mathbb{P}(S=0)}
\Bigg[
\sum_{m\in\mathcal{M}}
S_{a,k}(\nu\mid L,m,R=0)
q_{p}^{a^{*}}(m\mid L)
\nonumber\\
&\qquad\qquad
-
\sum_{(\ell,m)\in\mathcal{L}\times\mathcal{M}}
S_{a,k}(\nu\mid \ell,m,R=0)
q_{p}^{a^{*}}(m\mid \ell)
r_0(\ell)
\Bigg].
\end{align*}

Recalling the definition of $\theta(P)$, this becomes
\begin{align*}
C
&=
\frac{I(S=0)}{\mathbb{P}(S=0)}
\Bigg[
\mathbb{E}_{M}
\left\{
S_{a,k}(\nu\mid L,M,R=0)
\mid L,A=a^{*},S=p,R=0
\right\}
-
\theta(P)
\Bigg].
\end{align*}

%%%%%%%%%%%%%%%%%%%%%%%%%%%%%%%%%%%%%%%%%%%%%%%%%%%%%%%%%%%%
\paragraph*{Contribution of the conditional survival distribution: $A. $}

We now consider the first term. For discrete event times
$t_1<t_2<\cdots$, define
\begin{align*}
\lambda_{a,k}(t_i\mid \ell,m,R=0)
&=
\mathbb{P}
\left(
\widetilde T=t_i, \delta=1
\mid
\widetilde T\geq t_i,
A=a,S=k,L=\ell,M=m,R=0
\right).
\end{align*}

Then
\begin{align*}
S_{a,k}(\nu\mid \ell,m,R=0)
&=
\prod_{t_i\leq\nu}
\left\{
1-\lambda_{a,k}(t_i\mid \ell,m,R=0)
\right\}.
\end{align*}

For a fixed $t_l$,
\begin{align*}
\frac{d}{d\varepsilon}
\left[
1-\lambda_{a,k,\varepsilon}
(t_l\mid \ell,m,R=0)
\right]_{\varepsilon=0}
&=
-
\dot\lambda_{a,k}(t_l\mid \ell,m,R=0).
\end{align*}

The influence-function contribution associated with
$\lambda_{a,k}(t_l\mid\ell,m,R=0)$ is
\begin{align*}
&\frac{
I(A=a,S=k,L=\ell,M=m,R=0,\widetilde T\geq t_l)
}{
\mathbb{P}
(A=a,S=k,L=\ell,M=m,R=0,\widetilde T\geq t_l)
}
\nonumber\\
&\qquad\times
\Big[
I(\widetilde T=t_l, \delta=1)
-
\lambda_{a,k}(t_l\mid\ell,m,R=0)
\Big].
\end{align*}

By the product rule,
\begin{align*}
\dot S_{a,k}(\nu\mid\ell,m,R=0)
&=
-
\sum_{t_l\leq\nu}
\prod_{\substack{t_i\leq\nu\\t_i\neq t_l}}
\left\{
1-\lambda_{a,k}(t_i\mid\ell,m,R=0)
\right\}
\nonumber\\
&\qquad\times
\frac{
I(A=a,S=k,L=\ell,M=m,R=0,\widetilde T\geq t_l)
}{
\mathbb{P}
(A=a,S=k,L=\ell,M=m,R=0,\widetilde T\geq t_l)
}
\nonumber\\
&\qquad\times
\Big[
I(\widetilde T=t_l, \delta=1)
-
\lambda_{a,k}(t_l\mid\ell,m,R=0)
\Big].
\end{align*}

Substituting this expression into $A$ gives
\begin{align*}
A
&=
-
\sum_{(\ell,m)\in\mathcal{L}\times\mathcal{M}}
\sum_{t_l\leq\nu}
\prod_{\substack{t_i\leq\nu\\t_i\neq t_l}}
\left\{
1-\lambda_{a,k}(t_i\mid\ell,m,R=0)
\right\}
\nonumber\\
&\quad\times
\frac{
I(A=a,S=k,L=\ell,M=m,R=0,\widetilde T\geq t_l)
}{
\mathbb{P}
(A=a,S=k,L=\ell,M=m,R=0,\widetilde T\geq t_l)
}
\nonumber\\
&\quad\times
\Big[
I(\widetilde T=t_l, \delta=1)
-
\lambda_{a,k}(t_l\mid\ell,m,R=0)
\Big]
\nonumber\\
&\quad\times
q_{p}^{a^{*}}(m\mid\ell)
r_0(\ell).
\end{align*}

After summing over $(\ell,m)$,
\begin{align*}
A
&=
-
\frac{1}{\mathbb{P}(S=0)}
\frac{
I(A=a,S=k,R=0)\,
\mathbb{P}(A=a^{*},S=p,R=0\mid L,M)\,
\mathbb{P}(S=0\mid L)
}{
\mathbb{P}(A=a,S=k,R=0\mid L,M)\,
\mathbb{P}(A=a^{*},S=p,R=0\mid L)
}
\nonumber\\
&\quad\times
\sum_{t_l\leq\nu}
\prod_{\substack{t_i\leq\nu\\t_i\neq t_l}}
\left\{
1-\lambda_{a,k}(t_i\mid L,M,R=0)
\right\}
\nonumber\\
&\quad\times
\frac{
I(\widetilde T\geq t_l)
}{
\mathbb{P}
(\widetilde T\geq t_l
\mid A=a,S=k,L,M,R=0)
}
\nonumber\\
&\quad\times
\Big[
I(\widetilde T=t_l, \delta=1)
-
\lambda_{a,k}(t_l\mid L,M,R=0)
\Big].
\end{align*}

Under conditional independent censoring,
\begin{align*}
&\mathbb{P}
(\widetilde T\geq t_l
\mid A=a,S=k,L,M,R=0)
\nonumber\\
&\qquad=
\mathbb{P}
(T\geq t_l, C\geq t_l
\mid A=a,S=k,L,M,R=0)
\\
&\qquad=
\mathbb{P}
(T\geq t_l
\mid A=a,S=k,L,M,R=0)
\,
\mathbb{P}
( C\geq t_l
\mid A=a,S=k,L,M,R=0)
\\
&\qquad=
S_{a,k}(t_{l-1}\mid L,M,R=0)
G_{a,k}(t_{l-1}\mid L,M,R=0)
\\
&\qquad=
\prod_{\alpha=1}^{l-1}
\left\{
1-\lambda_{a,k}(t_\alpha\mid L,M,R=0)
\right\}
\left\{
1-g_{a,k}(t_\alpha\mid L,M,R=0)
\right\}.
\end{align*}

Therefore,
\begin{align*}
&\frac{
\displaystyle
\prod_{\substack{t_i\leq\nu\\t_i\neq t_l}}
\left\{
1-\lambda_{a,k}(t_i\mid L,M,R=0)
\right\}
}{
\displaystyle
\mathbb{P}
(\widetilde T\geq t_l
\mid A=a,S=k,L,M,R=0)
}
\nonumber\\
&\qquad=
\frac{
\displaystyle
\prod_{t_{l+1}\leq t_i\leq\nu}
\left\{
1-\lambda_{a,k}(t_i\mid L,M,R=0)
\right\}
}{
\displaystyle
\prod_{t_i\leq t_{l-1}}
\left\{
1-g_{a,k}(t_i\mid L,M,R=0)
\right\}
}.
\end{align*}

We define
\begin{align*}
\mathcal{G}_{a,k}^{(t_l)}(\nu\mid L,M)
&=
\frac{
\displaystyle
\prod_{t_{l+1}\leq t_i\leq\nu}
\left\{
1-\lambda_{a,k}(t_i\mid L,M,R=0)
\right\}
}{
\displaystyle
\prod_{t_i\leq t_{l-1}}
\left\{
1-g_{a,k}(t_i\mid L,M,R=0)
\right\}
}.
\end{align*}

Hence,
\begin{align*}
A
&=
-
\frac{1}{\mathbb{P}(S=0)}
\frac{
I(A=a,S=k,R=0)\,
\mathbb{P}(A=a^{*},S=p,R=0\mid L,M)\,
\mathbb{P}(S=0\mid L)
}{
\mathbb{P}(A=a,S=k,R=0\mid L,M)\,
\mathbb{P}(A=a^{*},S=p,R=0\mid L)
}
\nonumber\\
&\quad\times
\sum_{t_l\leq\nu}
I(\widetilde T\geq t_l)
\mathcal{G}_{a,k}^{(t_l)}(\nu\mid L,M)
\nonumber\\
&\quad\times
\Big[
I(\widetilde T=t_l, \delta=1)
-
\lambda_{a,k}(t_l\mid L,M,R=0)
\Big].
\end{align*}

%%%%%%%%%%%%%%%%%%%%%%%%%%%%%%%%%%%%%%%%%%%%%%%%%%%%%%%%%%%%
\subsubsection*{Influence function}

Combining the three contributions, we obtain
\begin{align*}
\varphi(O)
&=
-
\frac{1}{\mathbb{P}(S=0)}
\frac{
I(A=a,S=k,R=0)\,
\mathbb{P}(A=a^{*},S=p,R=0\mid L,M)\,
\mathbb{P}(S=0\mid L)
}{
\mathbb{P}(A=a,S=k,R=0\mid L,M)\,
\mathbb{P}(A=a^{*},S=p,R=0\mid L)
}
\nonumber\\
&\quad\times
\sum_{t_l\leq\nu}
I(\widetilde T\geq t_l)
\mathcal{G}_{a,k}^{(t_l)}(\nu\mid L,M)
\Big[
I(\widetilde T=t_l, \delta=1)
-
\lambda_{a,k}(t_l\mid L,M,R=0)
\Big]
\nonumber\\[0.2cm]
&\quad+
\frac{1}{\mathbb{P}(S=0)}
\frac{
I(A=a^{*},S=p,R=0)
}{
\mathbb{P}(A=a^{*},S=p,R=0\mid L)
}
\mathbb{P}(S=0\mid L)
\nonumber\\
&\qquad\times
\Bigg[
S_{a,k}(\nu\mid L,M,R=0)
-
\mathbb{E}_{M}
\left\{
S_{a,k}(\nu\mid L,M,R=0)
\mid L,A=a^{*},S=p,R=0
\right\}
\Bigg]
\nonumber\\[0.2cm]
&\quad+
\frac{I(S=0)}{\mathbb{P}(S=0)}
\Bigg[
\mathbb{E}_{M}
\left\{
S_{a,k}(\nu\mid L,M,R=0)
\mid L,A=a^{*},S=p,R=0
\right\}
-
\theta(P)
\Bigg].
\end{align*}

\section{Multiply robustness}

We consider two main cases. First, we consider the case where 
$\hat{\lambda}_{k}^{a}$ is correctly specified. Second, we consider the case where 
$\hat{g}_{k}^{a}$ and $\widehat{\Omega}_{k,p}^{a,a^*}$ are correctly specified.

Throughout the proof, we use the shorthand notation
\[
I(a,k,R=0)=I(A=a,S=k,R=0),
\qquad
I(a^*,p,R=0)=I(A=a^*,S=p,R=0).
\]

Let
\[
S_0(\nu\mid L,M)
=
\mathbb{P}(T>\nu\mid a,L,M,k,R=0),
\]
and
\[
\lambda_0(t_l\mid L,M)
=
\mathbb{P}
\left(
T=t_l
\mid
a,L,M,k,R=0,T\geq t_l
\right).
\]

We also denote
$
b_0(L)
=
\mathbb{E}\left[
S_0(\nu\mid L,M)
\mid L,a^*,p,R=0
\right].
$

Moreover, $g_k^a$ denotes the censoring hazard among individuals for whom
the mediator is observed, such that
$
\prod_{t_j\leq t_{l-1}}
\left\{
1-g_{0,k}^{a}(t_j\mid L,M)
\right\}
=
\mathbb{P}
\left(
C>t_{l-1}
\mid a,L,M,k,R=0
\right).
$

The nuisance parameter associated with mediator missingness, denoted by
$\rho$, is absorbed into the definitions of
$\Omega_{k,p}^{a,a^*}$ and $W_p^a$. In particular, under correct
specification,
\[
\Omega_{0,k,p}^{a,a^*}(L,M)
=
\frac{
\mathbb{P}(A=a^*,S=p,R=0\mid L,M)
\mathbb{P}(S=0\mid L)
}{
\mathbb{P}(A=a,S=k,R=0\mid L,M)
\mathbb{P}(A=a^*,S=p,R=0\mid L)
},
\]
and
$
W_{0,p}^{a}(L)
=
\frac{
\mathbb{P}(S=0\mid L)
}{
\mathbb{P}(A=a^*,S=p,R=0\mid L)
}.
$

We also have
$
\frac{n_0}{n}
\xrightarrow{\mathbb{P}}
\mathbb{P}(S=0).
$

% ============================================================
% CASE 1 : LAMBDA CORRECTLY SPECIFIED
% ============================================================

\subsection*{Case 1: $\hat{\lambda}_{k}^{a}$ correctly specified}

Suppose first that $\hat{\lambda}_{k}^{a}$ is correctly specified,
while allowing the remaining nuisance parameters to be misspecified. Then
\[
\widehat{\lambda}_{k}^{a}(t_l\mid L,M,R=0)
=
\lambda_0(t_l\mid L,M)
+
o_{\mathbb{P}}(1),
\]
and consequently
\[
\widehat{S}_{k}^{a}(\nu\mid L,M,R=0)
=
S_0(\nu\mid L,M)
+
o_{\mathbb{P}}(1).
\]

Let the remaining nuisance parameters converge to arbitrary probability
limits. In particular,
\[
\widehat{W}_{p}^{a}(L)
=
\tilde{W}(L)
+
o_{\mathbb{P}}(1),
\]
and
\[
\widehat{b}_{k,p}^{a,a^*}(L)
=
\tilde{b}(L)
+
o_{\mathbb{P}}(1).
\]

Similarly, denote the probability limits of the nuisance parameters
appearing in the first augmentation term by
$\tilde{\omega}_1(L,M)$ and
$\tilde{\mathcal G}^{\,t_l}(L,M)$.

By the law of large numbers,
$
\widehat{\theta}_{k,p}^{a,a^*}
\xrightarrow{\mathbb{P}}
X+Y+
\frac{1}{\mathbb{P}(S=0)}
\mathbb{E}\left[
I(S=0)\tilde{b}(L)
\right],
$
where
\begin{align*}
X
&=
-\frac{1}{\mathbb{P}(S=0)}
\mathbb{E}\Biggl[
I(a,k,R=0)\tilde{\omega}_1(L,M)
\sum_{t_l\leq\nu}
\tilde{\mathcal G}^{\,t_l}(L,M)
I(\tilde{T}\geq t_l)
\\
&\hspace{4cm}\times
\Bigl(
I(\tilde{T}=t_l,\delta=1)
-
\lambda_0(t_l\mid L,M)
\Bigr)
\Biggr],
\end{align*}
and
\begin{align*}
Y
&=
\frac{1}{\mathbb{P}(S=0)}
\mathbb{E}\Biggl[
I(a^*,p,R=0)\tilde{W}(L)
\Bigl(
S_0(\nu\mid L,M)
-
\tilde{b}(L)
\Bigr)
\Biggr].
\end{align*}

We first consider $X$. By iterated expectations,
\begin{align*}
X
&=
-\frac{1}{\mathbb{P}(S=0)}
\mathbb{E}\Biggl[
\mathbb{E}\Biggl[
I(a,k,R=0)\tilde{\omega}_1(L,M)
\sum_{t_l\leq\nu}
\tilde{\mathcal G}^{\,t_l}(L,M)
I(\tilde{T}\geq t_l)
\\
&\hspace{3cm}\times
\Bigl(
I(\tilde{T}=t_l,\delta=1)
-
\lambda_0(t_l\mid L,M)
\Bigr)
\Bigm|L,M
\Biggr]
\Biggr]
\\
&=
-\frac{1}{\mathbb{P}(S=0)}
\sum_{t_l\leq\nu}
\mathbb{E}\Biggl[
\tilde{\omega}_1(L,M)
\mathbb{P}
\left(
a,k,R=0,\tilde{T}\geq t_l
\mid L,M
\right)
\tilde{\mathcal G}^{\,t_l}(L,M)
\\
&\hspace{2cm}\times
\mathbb{E}\Biggl[
I(\tilde{T}=t_l,\delta=1)
-
\lambda_0(t_l\mid L,M)
\Bigm|
a,k,R=0,L,M,\tilde{T}\geq t_l
\Biggr]
\Biggr].
\end{align*}

Since $\lambda_0$ is correctly specified,
\[
\mathbb{E}\left[
I(\tilde{T}=t_l,\delta=1)
\mid
a,k,R=0,L,M,\tilde{T}\geq t_l
\right]
=
\lambda_0(t_l\mid L,M),
\]
and hence
\[
X=0.
\]

We now consider $Y$. By iterated expectations,
\begin{align*}
Y
&=
\frac{1}{\mathbb{P}(S=0)}
\mathbb{E}\Biggl[
\mathbb{E}\Biggl[
I(a^*,p,R=0)\tilde{W}(L)
\Bigl(
S_0(\nu\mid L,M)
-
\tilde{b}(L)
\Bigr)
\Bigm|L
\Biggr]
\Biggr]
\\
&=
\frac{1}{\mathbb{P}(S=0)}
\mathbb{E}\Biggl[
\mathbb{P}(a^*,p,R=0\mid L)
\tilde{W}(L)
\\
&\hspace{2cm}\times
\mathbb{E}\Bigl[
S_0(\nu\mid L,M)
-
\tilde{b}(L)
\mid L,a^*,p,R=0
\Bigr]
\Biggr]
\\
&=
\frac{1}{\mathbb{P}(S=0)}
\mathbb{E}\Biggl[
\mathbb{P}(a^*,p,R=0\mid L)
\tilde{W}(L)
\Bigl(
b_0(L)-\tilde{b}(L)
\Bigr)
\Biggr].
\end{align*}

Moreover,
\[
\theta_{k,p}^{a,a^*}
=
\frac{1}{\mathbb{P}(S=0)}
\mathbb{E}\left[
\mathbb{P}(S=0\mid L)b_0(L)
\right].
\]

Therefore,
\begin{align*}
&
\widehat{\theta}_{k,p}^{a,a^*}
-
\theta_{k,p}^{a,a^*}
\\
&=
\frac{1}{\mathbb{P}(S=0)}
\mathbb{E}\Biggl[
\mathbb{P}(a^*,p,R=0\mid L)
\tilde{W}(L)
\Bigl(
b_0(L)-\tilde{b}(L)
\Bigr) + o_{\mathbb{P}}(1)
\\
&\hspace{2cm}
+
\mathbb{P}(S=0\mid L)\tilde{b}(L)
-
\mathbb{P}(S=0\mid L)b_0(L)
\Biggr] + o_{\mathbb{P}}(1)
\\
&=
\frac{1}{\mathbb{P}(S=0)}
\mathbb{E}\Biggl[
\left\{
\mathbb{P}(a^*,p,R=0\mid L)\tilde{W}(L)
-
\mathbb{P}(S=0\mid L)
\right\}
\\
&\hspace{3cm}\times
\left\{
b_0(L)-\tilde{b}(L)
\right\}
\Biggr] + o_{\mathbb{P}}(1).
\end{align*}

Thus, when $\hat{\lambda}_{k}^{a}$ is correctly specified, the bias is
zero if either $\tilde{b}(L)=b_0(L),$
that is, if $\widehat{b}_{k,p}^{a,a^*}$ is correctly specified, or if $\mathbb{P}(a^*,p,R=0\mid L)\tilde{W}(L)
=
\mathbb{P}(S=0\mid L).$

In particular, when $\widehat{W}_{p}^{a}$ is correctly specified,$
\tilde{W}(L)
=
\frac{
\mathbb{P}(S=0\mid L)
}{
\mathbb{P}(a^*,p,R=0\mid L)
},$
and therefore $\mathbb{P}(a^*,p,R=0\mid L)\tilde{W}(L)
=
\mathbb{P}(S=0\mid L).$

Hence, under correct specification of $\lambda_k^a$, $\widehat{\theta}_{k,p}^{a,a^*}
\xrightarrow{\mathbb{P}}
\theta_{k,p}^{a,a^*}$
if either $\widehat{b}_{k,p}^{a,a^*}$ or $\widehat{W}_{p}^{a}$ is
correctly specified. This gives the first two robustness conditions: $\{\lambda_k^a,b_{k,p}^{a,a^*}\}
\qquad\text{or}
\{\lambda_k^a,W_p^a\}.$

% ============================================================
% CASE 2 : g AND OMEGA CORRECTLY SPECIFIED
% ============================================================

\subsection*{Case 2: $\hat{g}_{k}^{a}$ and
$\widehat{\Omega}_{k,p}^{a,a^*}$ correctly specified}

Suppose now that $\hat{g}_{k}^{a}$ and
$\widehat{\Omega}_{k,p}^{a,a^*}$ are correctly specified, while allowing
the remaining nuisance parameters to be misspecified.

Then
\begin{align*}
\widehat{\Omega}_{k,p}^{a,a^*}(L,M)
&=
\frac{
\mathbb{P}(A=a^*,S=p,R=0\mid L,M)
\mathbb{P}(S=0\mid L)
}{
\mathbb{P}(A=a,S=k,R=0\mid L,M)
\mathbb{P}(A=a^*,S=p,R=0\mid L)
}
+
o_{\mathbb{P}}(1), \\
\widehat{g}_{k}^{a}(t_l\mid L,M)
&=
g_{0,k}^{a}(t_l\mid L,M)
+
o_{\mathbb{P}}(1).
\end{align*}

Let\begin{align*}
\widehat{\lambda}_{k}^{a}(t_l\mid L,M,R=0)
&=
\tilde{\lambda}(t_l\mid L,M)
+
o_{\mathbb{P}}(1),
\\
\widehat{S}_{k}^{a}(\nu\mid L,M,R=0)
&=
\tilde{S}(\nu\mid L,M)
+
o_{\mathbb{P}}(1),\\
\widehat{b}_{k,p}^{a,a^*}(L)
&=
\tilde{b}(L)
+
o_{\mathbb{P}}(1),
\\
\widehat{W}_{p}^{a}(L)
&=
\tilde{W}(L)
+
o_{\mathbb{P}}(1).
\end{align*}

We also have
$
\widehat{\mathcal{G}}_{k}^{a,(t_l)}
(\nu\mid L,M)
=
\frac{
\prod_{t_l<t_j\leq\nu}
\left\{
1-\tilde{\lambda}(t_j\mid L,M)
\right\}
}{
\mathbb{P}
\left(
C>t_{l-1}
\mid a,L,M,k,R=0
\right)
}
+
o_{\mathbb{P}}(1).
$

Let
$
\tilde{p}^{\,t_l}(L,M)
=
\prod_{t_l<t_j\leq\nu}
\left\{
1-\tilde{\lambda}(t_j\mid L,M)
\right\}.
$

By the law of large numbers,
$
\widehat{\theta}_{k,p}^{a,a^*}
\xrightarrow{\mathbb{P}}
\bar{X}+\bar{Y}
+
\frac{1}{\mathbb{P}(S=0)}
\mathbb{E}\left[
I(S=0)\tilde{b}(L)
\right],
$
where
\begin{align*}
\bar{X}
&=
-\frac{1}{\mathbb{P}(S=0)}
\mathbb{E}\Biggl[
I(a,k,R=0)
\frac{
\mathbb{P}(a^*,p,R=0\mid L,M)
\mathbb{P}(S=0\mid L)
}{
\mathbb{P}(a,k,R=0\mid L,M)
\mathbb{P}(a^*,p,R=0\mid L)
}
\\
&\hspace{1.5cm}\times
\sum_{t_l\leq\nu}
\frac{
I(\tilde{T}\geq t_l)
\tilde{p}^{\,t_l}(L,M)
}{
\mathbb{P}
\left(
C>t_{l-1}
\mid M,L,a,k,R=0
\right)
}
\\
&\hspace{2.5cm}\times
\Bigl(
I(\tilde{T}=t_l,\delta=1)
-
\tilde{\lambda}(t_l\mid L,M)
\Bigr)
\Biggr],
\end{align*}
and
\begin{align*}
\bar{Y}
&=
\frac{1}{\mathbb{P}(S=0)}
\mathbb{E}\Biggl[
I(a^*,p,R=0)\tilde{W}(L)
\Bigl(
\tilde{S}(\nu\mid L,M)
-
\tilde{b}(L)
\Bigr)
\Biggr].
\end{align*}

We first consider $\bar{X}$. We have
\begin{align*}
\bar{X}
&=
-\frac{1}{\mathbb{P}(S=0)}
\mathbb{E}\Biggl[
I(a,k,R=0)
\frac{
\mathbb{P}(a^*,p,R=0\mid L,M)
\mathbb{P}(S=0\mid L)
}{
\mathbb{P}(a,k,R=0\mid L,M)
\mathbb{P}(a^*,p,R=0\mid L)
}
\\
&\hspace{1.5cm}\times
\sum_{t_l\leq\nu}
\frac{
I(\tilde{T}\geq t_l)
\tilde{p}^{\,t_l}(L,M)
}{
\mathbb{P}
\left(
C>t_{l-1}
\mid M,L,a,k,R=0
\right)
}
\\
&\hspace{2.5cm}\times
\Bigl(
I(\tilde{T}=t_l,\delta=1)
-
\tilde{\lambda}(t_l\mid L,M)
\Bigr)
\Biggr]
\\[0.3cm]
&=
-\frac{1}{\mathbb{P}(S=0)}
\mathbb{E}\Biggl[
\frac{
\mathbb{P}(a^*,p,R=0\mid L,M)
\mathbb{P}(S=0\mid L)
}{
\mathbb{P}(a^*,p,R=0\mid L)
}
\\
&\hspace{1cm}\times
\mathbb{E}\Biggl[
\sum_{t_l\leq\nu}
\frac{
I(\tilde{T}\geq t_l)
\tilde{p}^{\,t_l}(L,M)
}{
\mathbb{P}
\left(
C>t_{l-1}
\mid M,L,a,k,R=0
\right)
}
\\
&\hspace{2cm}\times
\Bigl(
I(\tilde{T}=t_l,\delta=1)
-
\tilde{\lambda}(t_l\mid L,M)
\Bigr)
\Bigm|
a,k,R=0,L,M
\Biggr]
\Biggr]
\\[0.3cm]
&=
-\frac{1}{\mathbb{P}(S=0)}
\mathbb{E}\Biggl[
\mathbb{P}(S=0\mid L)
\mathbb{E}\Biggl[
\mathbb{E}\Biggl[
\sum_{t_l\leq\nu}
\frac{
I(\tilde{T}\geq t_l)
\tilde{p}^{\,t_l}(L,M)
}{
\mathbb{P}
\left(
C>t_{l-1}
\mid M,L,a,k,R=0
\right)
}
\\
&\hspace{2cm}\times
\Bigl(
I(\tilde{T}=t_l,\delta=1)
-
\tilde{\lambda}(t_l\mid L,M)
\Bigr)
\Bigm|
a,k,R=0,L,M
\Biggr]
\Bigm|
L,a^*,p,R=0
\Biggr]
\Biggr].
\end{align*}

For fixed $(L,M)$, consider
\begin{align*}
&\mathbb{E}\Biggl[
\sum_{t_l\leq\nu}
\frac{
I(\tilde{T}\geq t_l)
\tilde{p}^{\,t_l}(L,M)
}{
\mathbb{P}
\left(
C>t_{l-1}
\mid M,L,a,k,R=0
\right)
}
\\
&\hspace{2cm}\times
\Bigl(
I(\tilde{T}=t_l,\delta=1)
-
\tilde{\lambda}(t_l\mid L,M)
\Bigr)
\Bigm|
a,k,R=0,L,M
\Biggr]
\\
&=
\sum_{t_l\leq\nu}
\frac{
\mathbb{P}
\left(
\tilde{T}\geq t_l
\mid a,k,R=0,L,M
\right)
}{
\mathbb{P}
\left(
C>t_{l-1}
\mid a,k,R=0,L,M
\right)
}
\tilde{p}^{\,t_l}(L,M)
\\
&\hspace{1cm}\times
\mathbb{E}\Biggl[
I(\tilde{T}=t_l,\delta=1)
-
\tilde{\lambda}(t_l\mid L,M)
\Bigm|
a,k,R=0,L,M,\tilde{T}\geq t_l
\Biggr]
\\
&=
\sum_{t_l\leq\nu}
\frac{
\mathbb{P}
\left(
\tilde{T}\geq t_l
\mid a,k,R=0,L,M
\right)
}{
\mathbb{P}
\left(
C>t_{l-1}
\mid a,k,R=0,L,M
\right)
}
\tilde{p}^{\,t_l}(L,M)
\\
&\hspace{1cm}\times
\Bigl[
\lambda_0(t_l\mid L,M)
-
\tilde{\lambda}(t_l\mid L,M)
\Bigr].
\end{align*}

Under conditional independent censoring,
\[
\frac{
\mathbb{P}
\left(
\tilde{T}\geq t_l
\mid a,k,R=0,L,M
\right)
}{
\mathbb{P}
\left(
C>t_{l-1}
\mid a,k,R=0,L,M
\right)
}
=
\mathbb{P}
\left(
T>t_{l-1}
\mid a,k,R=0,L,M
\right).
\]

Therefore,
\begin{align*}
&\mathbb{E}\Biggl[
\sum_{t_l\leq\nu}
\frac{
I(\tilde{T}\geq t_l)
\tilde{p}^{\,t_l}(L,M)
}{
\mathbb{P}
\left(
C>t_{l-1}
\mid M,L,a,k,R=0
\right)
}
\\
&\hspace{2cm}\times
\Bigl(
I(\tilde{T}=t_l,\delta=1)
-
\tilde{\lambda}(t_l\mid L,M)
\Bigr)
\Bigm|
a,k,R=0,L,M
\Biggr]
\\
&=
\sum_{t_l\leq\nu}
\mathbb{P}
\left(
T>t_{l-1}
\mid a,k,R=0,L,M
\right)
\tilde{p}^{\,t_l}(L,M)
\\
&\hspace{1cm}\times
\Bigl[
\lambda_0(t_l\mid L,M)
-
\tilde{\lambda}(t_l\mid L,M)
\Bigr].
\end{align*}

Now, noting that
\[
\mathbb{P}
\left(
T>t_{l-1}
\mid a,k,R=0,L,M
\right)
=
\prod_{t_j<t_l}
\left\{
1-\lambda_0(t_j\mid L,M)
\right\},
\]
and
\[
\tilde{p}^{\,t_l}(L,M)
=
\prod_{t_l<t_j\leq\nu}
\left\{
1-\tilde{\lambda}(t_j\mid L,M)
\right\},
\]
we obtain
\begin{align*}
&\sum_{t_l\leq\nu}
\mathbb{P}
\left(
T>t_{l-1}
\mid a,k,R=0,L,M
\right)
\tilde{p}^{\,t_l}(L,M)
\Bigl[
\lambda_0(t_l\mid L,M)
-
\tilde{\lambda}(t_l\mid L,M)
\Bigr]
\\
&=
\sum_{t_l\leq\nu}
\left[
\prod_{t_j<t_l}
\left\{
1-\lambda_0(t_j\mid L,M)
\right\}
\right]
\Bigl[
\lambda_0(t_l\mid L,M)
-
\tilde{\lambda}(t_l\mid L,M)
\Bigr]
\\
&\hspace{2cm}\times
\left[
\prod_{t_l<t_j\leq\nu}
\left\{
1-\tilde{\lambda}(t_j\mid L,M)
\right\}
\right]
\\
&=
\prod_{t_j\leq\nu}
\left\{
1-\tilde{\lambda}(t_j\mid L,M)
\right\}
-
\prod_{t_j\leq\nu}
\left\{
1-\lambda_0(t_j\mid L,M)
\right\}
\\
&=
\tilde{S}(\nu\mid L,M)
-
S_0(\nu\mid L,M),
\end{align*}
where the last equality follows from the telescoping product identity.

Recalling the minus sign in front of the first augmentation term, we
obtain
\begin{align*}
\bar{X}
&=
\frac{1}{\mathbb{P}(S=0)}
\mathbb{E}\Biggl[
\mathbb{P}(S=0\mid L)
\mathbb{E}\Bigl[
S_0(\nu\mid L,M)
-
\tilde{S}(\nu\mid L,M)
\mid L,a^*,p,R=0
\Bigr]
\Biggr].
\end{align*}

We now consider $\bar{Y}$:
\begin{align*}
\bar{Y}
&=
\frac{1}{\mathbb{P}(S=0)}
\mathbb{E}\Biggl[
I(a^*,p,R=0)\tilde{W}(L)
\Bigl(
\tilde{S}(\nu\mid L,M)
-
\tilde{b}(L)
\Bigr)
\Biggr]
\\
&=
\frac{1}{\mathbb{P}(S=0)}
\mathbb{E}\Biggl[
\mathbb{P}(a^*,p,R=0\mid L)
\tilde{W}(L)
\\
&\hspace{2cm}\times
\mathbb{E}\Bigl[
\tilde{S}(\nu\mid L,M)
-
\tilde{b}(L)
\mid L,a^*,p,R=0
\Bigr]
\Biggr]
\\
&=
\frac{1}{\mathbb{P}(S=0)}
\mathbb{E}\Biggl[
\mathbb{P}(a^*,p,R=0\mid L)
\tilde{W}(L)
\\
&\hspace{2cm}\times
\left\{
\mathbb{E}\Bigl[
\tilde{S}(\nu\mid L,M)
\mid L,a^*,p,R=0
\Bigr]
-
\tilde{b}(L)
\right\}
\Biggr].
\end{align*}

Moreover,
$
\theta_{k,p}^{a,a^*}
=
\frac{1}{\mathbb{P}(S=0)}
\mathbb{E}\Biggl[
\mathbb{P}(S=0\mid L)
\mathbb{E}\Bigl[
S_0(\nu\mid L,M)
\mid L,a^*,p,R=0
\Bigr]
\Biggr].$

Therefore,
\begin{align*}
&
\widehat{\theta}_{k,p}^{a,a^*}
-
\theta_{k,p}^{a,a^*}
\\
&=
\frac{1}{\mathbb{P}(S=0)}
\mathbb{E}\Biggl[
\mathbb{P}(S=0\mid L)
\mathbb{E}\Bigl[
S_0(\nu\mid L,M)
-
\tilde{S}(\nu\mid L,M)
\mid L,a^*,p,R=0
\Bigr] + o_{\mathbb{P}}(1)
\\
&\qquad+
\mathbb{P}(a^*,p,R=0\mid L)
\tilde{W}(L)
\left\{
\mathbb{E}\Bigl[
\tilde{S}(\nu\mid L,M)
\mid L,a^*,p,R=0
\Bigr]
-
\tilde{b}(L)
\right\} + o_{\mathbb{P}}(1)
\\
&\qquad+
\mathbb{P}(S=0\mid L)\tilde{b}(L)
-
\mathbb{P}(S=0\mid L)
\mathbb{E}\Bigl[
S_0(\nu\mid L,M)
\mid L,a^*,p,R=0
\Bigr]
\Biggr] + o_{\mathbb{P}}(1).
\end{align*}

The terms involving $S_0$ cancel, and thus
\begin{align*}
&
\widehat{\theta}_{k,p}^{a,a^*}
-
\theta_{k,p}^{a,a^*}
\\
&=
\frac{1}{\mathbb{P}(S=0)}
\mathbb{E}\Biggl[
-\mathbb{P}(S=0\mid L)
\mathbb{E}\Bigl[
\tilde{S}(\nu\mid L,M)
\mid L,a^*,p,R=0
\Bigr] + o_{\mathbb{P}}(1)
\\
&\qquad+
\mathbb{P}(a^*,p,R=0\mid L)
\tilde{W}(L)
\mathbb{E}\Bigl[
\tilde{S}(\nu\mid L,M)
\mid L,a^*,p,R=0
\Bigr]
\\
&\qquad-
\mathbb{P}(a^*,p,R=0\mid L)
\tilde{W}(L)\tilde{b}(L)
+
\mathbb{P}(S=0\mid L)\tilde{b}(L)
\Biggr]
\\
&=
\frac{1}{\mathbb{P}(S=0)}
\mathbb{E}\Biggl[
\left\{
\mathbb{P}(a^*,p,R=0\mid L)\tilde{W}(L)
-
\mathbb{P}(S=0\mid L)
\right\}
\\
&\hspace{2cm}\times
\left\{
\mathbb{E}\Bigl[
\tilde{S}(\nu\mid L,M)
\mid L,a^*,p,R=0
\Bigr]
-
\tilde{b}(L)
\right\}
\Biggr].
\end{align*}

Therefore, when $\hat{g}_{k}^{a}$ and
$\widehat{\Omega}_{k,p}^{a,a^*}$ are correctly specified, the bias is
zero if either $\mathbb{P}(a^*,p,R=0\mid L)\tilde{W}(L)
=
\mathbb{P}(S=0\mid L),$
which holds when $\widehat{W}_{p}^{a}$ is correctly specified, or if $\tilde{b}(L)
=
\mathbb{E}\left[
\tilde{S}(\nu\mid L,M)
\mid L,a^*,p,R=0
\right].$

Hence, $\widehat{\theta}_{k,p}^{a,a^*}
\xrightarrow{\mathbb{P}}
\theta_{k,p}^{a,a^*}$
if either $\{g_k^a,\Omega_{k,p}^{a,a^*},W_p^a\}
$
are correctly specified, or if $
\{g_k^a,\Omega_{k,p}^{a,a^*}\}$
are correctly specified and
$\tilde{b}(L)
=
\mathbb{E}\left[
\tilde{S}(\nu\mid L,M)
\mid L,a^*,p,R=0
\right].$

% ============================================================
% MULTIPLY ROBUSTNESS RESULT
% ============================================================

\subsection*{Multiply robustness result}

Combining the two previous cases, the estimator
$\widehat{\theta}_{k,p}^{a,a^*}$ is consistent for
$\theta_{k,p}^{a,a^*}$ under any of the following conditions:
\[
\begin{array}{ll}
\textnormal{(i)}
&
\hat{\lambda}_{k}^{a}
\textnormal{ and }
\widehat{b}_{k,p}^{a,a^*}
\textnormal{ are correctly specified},
\\[0.2cm]
\textnormal{(ii)}
&
\hat{\lambda}_{k}^{a}
\textnormal{ and }
\widehat{W}_{p}^{a}
\textnormal{ are correctly specified},
\\[0.2cm]
\textnormal{(iii)}
&
\hat{g}_{k}^{a},
\widehat{\Omega}_{k,p}^{a,a^*}
\textnormal{ and }
\widehat{W}_{p}^{a}
\textnormal{ are correctly specified},
\\[0.2cm]
\textnormal{(iv)}
&
\hat{g}_{k}^{a}
\textnormal{ and }
\widehat{\Omega}_{k,p}^{a,a^*}
\textnormal{ are correctly specified, and}
\\
&
\displaystyle
\tilde{b}(L)
=
\mathbb{E}\left[
\tilde{S}(\nu\mid L,M)
\mid L,a^*,p,R=0
\right].
\end{array}
\]

Since the mediator-missingness nuisance parameter $\rho$ is absorbed
into $\Omega_{k,p}^{a,a^*}$ and $W_p^a$, no additional robustness
condition involving $\rho$ needs to be stated separately. Correct
specification of $\Omega_{k,p}^{a,a^*}$ or $W_p^a$ includes correct
specification of the corresponding mediator-observation mechanism. Thus, under any of the four sets of conditions above $\widehat{\theta}_{k,p}^{a,a^*}
\xrightarrow{\mathbb{P}}
\theta_{k,p}^{a,a^*}.$

\paragraph*{Reminder term}

\begin{align*}
\Omega_{k,p}^{a,a^*}(L,M)
&=
\frac{
V_p(L,M)\tau_p^{a^*}(L,M)U_0(L)
}{
V_k(L,M)\tau_k^a(L,M)
U_p(L)\pi_p^{a^*}(L)\rho_p^{a^*}(L)
},
&
W_p^{a^*}(L)
&=
\frac{
U_0(L)
}{
U_p(L)\pi_p^{a^*}(L)\rho_p^{a^*}(L)
}.
\end{align*}

Let
\begin{align*}
S_{k}^{a}(\nu\mid L,M)
&=
\prod_{t_j\leq\nu}
\left\{1-\lambda_k^a(t_j\mid L,M)\right\},
&
G_k^a(t_{l-1}\mid L,M)
&=
\prod_{t_j<t_l}
\left\{1-g_k^a(t_j\mid L,M)\right\},
\\
p_{\lambda,k}^{a,t_l}(L,M)
&=
\prod_{t_l<t_j\leq\nu}
\left\{1-\lambda_k^a(t_j\mid L,M)\right\},
&
\bar b_{k,p}^{a,a^*}(L)
&=
\mathbb E_0
\left[
S_k^a(\nu\mid L,M)
\mid a^*,p,L
\right].
\end{align*}

The population bias of the estimating equation can be written as
\begin{align*}
R_2(\eta,\eta_0)
={}&
-\frac{1}{\mathbb P_0(S=0)}
\mathbb E_0\Bigg[
I(a,k,R=0)\Omega_{k,p}^{a,a^*}(L,M)
\sum_{t_l\leq\nu}
\frac{
I(\widetilde T\geq t_l)
p_{\lambda,k}^{a,t_l}(L,M)
}{
G_k^a(t_{l-1}\mid L,M)
}
\\[-0.1cm]
&\hspace{4.5cm}\times
\left\{
I(\widetilde T=t_l, \delta=1)
-\lambda_k^a(t_l\mid L,M)
\right\}
\Bigg]
\\
&+
\frac{1}{\mathbb P_0(S=0)}
\mathbb E_0\left[
I(a^*,p,R=0)W_p^{a^*}(L)
\left\{
S_k^a(\nu\mid L,M)
-b_{k,p}^{a,a^*}(L)
\right\}
\right]
\\
&+
\frac{1}{\mathbb P_0(S=0)}
\mathbb E_0\left[
U_0(L)b_{k,p}^{a,a^*}(L)
\right]
-\theta_0 .
\end{align*}

For the first augmentation term, iterated expectations and conditional
independent censoring give
\begin{align*}
X(\eta)
={}&
\frac{1}{\mathbb P_0(S=0)}
\mathbb E_0\Bigg[
\mathbb P_0(a,k,R=0\mid L,M)
\Omega_{k,p}^{a,a^*}(L,M)
\\[-0.1cm]
&\qquad\times
\sum_{t_l\leq\nu}
S_{0,k}^{a}(t_{l-1}\mid L,M)
p_{\lambda,k}^{a,t_l}(L,M)
\frac{
G_{0,k}^{a}(t_{l-1}\mid L,M)
}{
G_k^{a}(t_{l-1}\mid L,M)
}
\\[-0.1cm]
&\qquad\times
\left\{
\lambda_k^a(t_l\mid L,M)
-\lambda_{0,k}^a(t_l\mid L,M)
\right\}
\Bigg].
\end{align*}

Using
\begin{align*}
\frac{G_{0,k}^{a}(t_{l-1})}{G_k^{a}(t_{l-1})}
&=
1+
\left\{
\frac{G_{0,k}^{a}(t_{l-1})}{G_k^{a}(t_{l-1})}-1
\right\},
\\
\sum_{t_l\leq\nu}
S_{0,k}^{a}(t_{l-1})
p_{\lambda,k}^{a,t_l}
\left\{
\lambda_k^a(t_l)-\lambda_{0,k}^a(t_l)
\right\}
&=
S_{0,k}^{a}(\nu)-S_k^{a}(\nu),
\end{align*}
we obtain
\begin{align*}
X(\eta)
={}&
\frac{1}{\mathbb P_0(S=0)}
\mathbb E_0\left[
\mathbb P_0(a,k,R=0\mid L,M)
\Omega_{k,p}^{a,a^*}
\left\{
S_{0,k}^{a}(\nu)-S_k^{a}(\nu)
\right\}
\right]
\\
&+
R_{g,\lambda},
\end{align*}
where
\begin{align*}
R_{g,\lambda}
={}&
\frac{1}{\mathbb P_0(S=0)}
\mathbb E_0\Bigg[
\mathbb P_0(a,k,R=0\mid L,M)
\Omega_{k,p}^{a,a^*}
\\
&\qquad\times
\sum_{t_l\leq\nu}
S_{0,k}^{a}(t_{l-1})
p_{\lambda,k}^{a,t_l}
\left\{
\frac{G_{0,k}^{a}(t_{l-1})}
     {G_k^{a}(t_{l-1})}
-1
\right\}
\left\{
\lambda_k^a(t_l)-\lambda_{0,k}^a(t_l)
\right\}
\Bigg].
\end{align*}

Adding and subtracting $\Omega_{0,k,p}^{a,a^*}$ gives
\begin{align*}
X(\eta)
={}&
\frac{1}{\mathbb P_0(S=0)}
\mathbb E_0\left[
U_0(L)
\left\{
b_{0,k,p}^{a,a^*}(L)
-\bar b_{k,p}^{a,a^*}(L)
\right\}
\right]
\\
&-
\frac{1}{\mathbb P_0(S=0)}
\mathbb E_0\left[
\mathbb P_0(a,k,R=0\mid L,M)
\left\{
\Omega_{k,p}^{a,a^*}
-\Omega_{0,k,p}^{a,a^*}
\right\}
\right.
\\[-0.1cm]
&\hspace{5cm}\left.
\times
\left\{
S_k^{a}(\nu)-S_{0,k}^{a}(\nu)
\right\}
\right]
+
R_{g,\lambda}.
\end{align*}

Moreover,
\begin{align*}
Y(\eta)
&=
\frac{1}{\mathbb P_0(S=0)}
\mathbb E_0\left[
U_{0,p}(L)\pi_{0,p}^{a^*}(L)\rho_{0,p}^{a^*}(L)
W_p^{a^*}(L)
\left\{
\bar b_{k,p}^{a,a^*}(L)
-b_{k,p}^{a,a^*}(L)
\right\}
\right],
\\
Z(\eta)-\theta_0
&=
\frac{1}{\mathbb P_0(S=0)}
\mathbb E_0\left[
U_0(L)
\left\{
b_{k,p}^{a,a^*}(L)
-b_{0,k,p}^{a,a^*}(L)
\right\}
\right].
\end{align*}

Since
\begin{align*}
U_0(L)
&=
U_{0,p}(L)\pi_{0,p}^{a^*}(L)\rho_{0,p}^{a^*}(L)
W_{0,p}^{a^*}(L),
\end{align*}
the terms involving $b_{0,k,p}^{a,a^*}$ cancel and
\begin{align*}
&
U_0\{b_0-\bar b\}
+
U_{0,p}\pi_{0,p}^{a^*}\rho_{0,p}^{a^*}
W_p^{a^*}\{\bar b-b\}
+
U_0\{b-b_0\}
\\
&\qquad=
U_{0,p}\pi_{0,p}^{a^*}\rho_{0,p}^{a^*}
\left\{
W_p^{a^*}-W_{0,p}^{a^*}
\right\}
\left\{
\bar b_{k,p}^{a,a^*}-b_{k,p}^{a,a^*}
\right\}.
\end{align*}

Hence,
\begin{align*}
R_2(\eta,\eta_0)
={}&
\frac{1}{\mathbb P_0(S=0)}
\mathbb E_0\Big[
U_{0,p}\pi_{0,p}^{a^*}\rho_{0,p}^{a^*}
\left(W_p^{a^*}-W_{0,p}^{a^*}\right)
\left(\bar b_{k,p}^{a,a^*}-b_{k,p}^{a,a^*}\right)
\Big]
\\
&-
\frac{1}{\mathbb P_0(S=0)}
\mathbb E_0\Big[
\mathbb P_0(a,k,R=0\mid L,M)
\left(\Omega_{k,p}^{a,a^*}
-\Omega_{0,k,p}^{a,a^*}\right)
\left(S_k^a(\nu)-S_{0,k}^a(\nu)\right)
\Big]
\\
&+
R_{g,\lambda}.
\end{align*}

Finally,
\begin{align*}
S_k^a(\nu)-S_{0,k}^a(\nu)
&=
-\sum_{t_l\leq\nu}
S_{0,k}^a(t_{l-1})
p_{\lambda,k}^{a,t_l}
\left\{
\lambda_k^a(t_l)-\lambda_{0,k}^a(t_l)
\right\},
\\
\frac{G_{0,k}^a(t_{l-1})}{G_k^a(t_{l-1})}-1
&=
\frac{
G_{0,k}^a(t_{l-1})-G_k^a(t_{l-1})
}{
G_k^a(t_{l-1})
},
\\
G_{0,k}^a(t_{l-1})-G_k^a(t_{l-1})
&=
\sum_{t_j<t_l}
\left[
\prod_{t_r<t_j}\{1-g_{0,k}^a(t_r)\}
\right]
\left\{
g_k^a(t_j)-g_{0,k}^a(t_j)
\right\}
\left[
\prod_{t_j<t_r<t_l}\{1-g_k^a(t_r)\}
\right].
\end{align*}

Therefore, up to bounded weights,

\begin{align*}
R_2(\eta,\eta_0)
={}&
\frac{1}{\mathbb P_0(S=0)}
\mathbb E_0\Big[
U_{0,p}(L)\pi_{0,p}^{a^*}(L)\rho_{0,p}^{a^*}(L)
\\
&\hspace{1.7cm}\times
\{W_p^a(L)-W_{0,p}^a(L)\}
\{b_{\lambda,k,p}^{a,a^*}(L)-b_{k,p}^{a,a^*}(L)\}
\Big]
\\
&+
\frac{1}{\mathbb P_0(S=0)}
\mathbb E_0\Big[
\mathbb P_0(a,k,R=0\mid L,M)
\\
&\hspace{1.7cm}\times
\{\Omega_{k,p}^{a,a^*}(L,M)
-\Omega_{0,k,p}^{a,a^*}(L,M)\}
\\
&\hspace{1.7cm}\times
\{S_{0,k}^a(\nu\mid L,M)
-S_k^a(\nu\mid L,M)\}
\Big]
+
R_{g,\lambda}.
\end{align*}

\begin{align*}
R_{g,\lambda}
={}&
\frac{1}{\mathbb P_0(S=0)}
\mathbb E_0\Bigg[
\mathbb P_0(a,k,R=0\mid L,M)
\Omega_{k,p}^{a,a^*}(L,M)
\\
&\qquad\times
\sum_{t_l\leq\nu}
S_{0,k}^a(t_{l-1}\mid L,M)
\prod_{t_l<t_j\leq\nu}
\{1-\lambda_k^a(t_j\mid L,M)\}
\\
&\qquad\times
\left\{
\frac{G_{0,k}^a(t_{l-1}\mid L,M)}
     {G_k^a(t_{l-1}\mid L,M)}
-1
\right\}
\{\lambda_k^a(t_l\mid L,M)
-\lambda_{0,k}^a(t_l\mid L,M)\}
\Bigg].
\end{align*}

\paragraph*{Double robustness of the TMLE.}
The same robustness property holds for the TMLE. By construction, the targeting
step solves the empirical efficient influence function equation,
\[
P_n \varphi\!\left(\widehat\eta^{\,*},
\widehat\theta_{\mathrm{TMLE}}\right)
=o_{\mathbb P}(1).
\]
Under the usual regularity conditions,
\[
(P_n-P_0)
\varphi\!\left(\widehat\eta^{\,*},
\widehat\theta_{\mathrm{TMLE}}\right)
=o_{\mathbb P}(1),
\]
so that any probability limit
$(\eta^{*,\dagger},\theta^\dagger)$ of the targeted estimator satisfies
\[
P_0\varphi(\eta^{*,\dagger},\theta^\dagger)=0.
\]
The population-bias decomposition derived above shows that, under any of the
robustness conditions stated above,
\[
P_0\varphi(\eta^{*,\dagger},\theta_0)=0.
\]
Since the estimating equation is linear in $\theta$ through the term
$-\theta$, $\theta_0$ is the unique solution of the limiting estimating
equation. Therefore,
\[
\widehat\theta_{\mathrm{TMLE}}
\xrightarrow{\mathbb P}
\theta_0.
\]
Hence, the TMLE inherits the same double (or multiple) robustness property as
the corresponding efficient-influence-function estimating equation.

By the Cauchy--Schwarz inequality and the positivity assumptions,
\begin{align*}
\left\|R_2(\widehat\eta,\eta_0)\right\|_{2}
\leq{}&
\eta_1
\left\|
\widehat W-W_0
\right\|_{2}
\left\|
\widehat{\bar b}-\widehat b
\right\|_{2}
\\
&+
\eta_2
\left\|
\widehat\Omega-\Omega_0
\right\|_{2}
\left\|
\widehat S-S_0
\right\|_{2}
\\
&+
\sum_{t_l\leq \nu}
\eta_{3,l} \left\|\frac{\prod_{t_l<t_j\leq\nu}
  \{1- \hat \lambda_k^a(t_j\mid L,M)\}}{\widehat G(t_{l-1})} \times
\{  \widehat G(t_{l-1})-G_0(t_{l-1})\}
\right\|_{2}
\left\|
\widehat\lambda(t_l)-\lambda_0(t_l)
\right\|_{2},
\end{align*}
for some finite constants $\eta_1,\eta_2,$ and $\eta_{3,l}$.
 If  $\frac{\prod_{t_l<t_j\leq\nu}
  \{1- \hat \lambda_k^a(t_j\mid L,M)\}}{\widehat G(t_{l-1})} $ is bounded such that $$\left \| \frac{\prod_{t_l<t_j\leq\nu}
  \{1- \hat \lambda_k^a(t_j\mid L,M)\}}{\widehat G(t_{l-1})} \right \|_{\infty} \leq \eta_{4, l} < \infty$$ with $\left \| X \right\|_{\infty} = \inf \left\{ M \geq 0, \mathbb{P}(|X| <C) =1\right\}$

Then

\begin{align*}
\left\|R_2(\widehat\eta,\eta_0)\right\|_{2}
\leq{}&
\eta_1
\left\|
\widehat W-W_0
\right\|_{2}
\left\|
\widehat{\bar b}-\widehat b
\right\|_{2}
\\
&+
\eta_2
\left\|
\widehat\Omega-\Omega_0
\right\|_{2}
\left\|
\widehat S-S_0
\right\|_{2}
\\
&+
\sum_{t_l\leq \nu}
\tilde{\eta}_{3,l} \left\|
 \widehat G(t_{l-1})-G_0(t_{l-1})
\right\|_{2}
\left\|
\widehat\lambda(t_l)-\lambda_0(t_l)
\right\|_{2},
\end{align*}

Moreover, the conditions involving \(S\) and \(G\) may be expressed directly in terms of the corresponding hazards. Indeed, by the product-difference identities,

$$
\|\widehat S(\nu)-S_0(\nu)\|_2
\le
\sum_{t_l\le\nu}
\left\|
S_0(t_{l-1})
\prod_{t_l<t_j\le\nu}(1-\widehat\lambda(t_j))
\{\widehat\lambda(t_l)-\lambda_0(t_l)\}
\right\|_2,
$$

and

$$
\|\widehat G(t_{l-1})-G_0(t_{l-1})\|_2
\le
\sum_{t_j<t_l}
\left\|
\prod_{t_r<t_j}(1-g_0(t_r))
\prod_{t_j<t_r<t_l}(1-\widehat g(t_r))
\{\widehat g(t_j)-g_0(t_j)\}
\right\|_2.
$$

If $\widehat g$ and $\widehat\lambda$ take values in $[0,1]$, then all
product weights appearing above are bounded by one. Consequently,
\begin{align*}
\left\|\widehat S(\nu)-S_0(\nu)\right\|_2
&\le
\sum_{t_l\le\nu}
\left\|
\widehat\lambda(t_l)-\lambda_0(t_l)
\right\|_2,
\\
\left\|\widehat G(t_{l-1})-G_0(t_{l-1})\right\|_2
&\le
\sum_{t_j<t_l}
\left\|
\widehat g(t_j)-g_0(t_j)
\right\|_2.
\end{align*}
Hence, the sufficient conditions above may be stated directly in terms
of the estimation errors of $g$ and $\lambda$. 

So under uniformly bounded  $\frac{\prod_{t_l<t_j\leq\nu}
  \{1- \hat \lambda_k^a(t_j\mid L,M)\}}{\widehat G(t_{l-1})} $, $n^{-1/4}$ convergence rate for $ \lambda_{k}^{a}, g_{k}^{a}, \Omega_{k,p}^{a, a^*},  W_{k}^{a}, b_{k,p}^{a,a^*}$ is sufficient for the convenient $n^{-1/2}$ convergence rate for the remainder term.

\begin{suppassumption}[Strong positivity]
There exists a constant $c>0$ such that 
\begin{align*}
\mathbb P(A=a\mid L,S)\ge c,
&\qquad
\mathbb P(A=a\mid M,L,S)\ge c, \\
\mathbb P(S=s\mid L)\ge c,
&\qquad
\mathbb P(S=s\mid L,M)\ge c,\\
\mathbb P(R=0\mid A=a,L=l,S=s)\ge c,
&\qquad  1-c \ge \mathbb P(C > t \mid A=a,L=l,S=s)\ge c,
\end{align*}
for all appropriate $(s,m,l,a),$and $t \le\nu.$
\end{suppassumption}

\begin{proof}[Corollary~1]
Under Assumption~A, let $0<\epsilon<c$. We define the truncated
versions of the nuisance estimators by
$\widehat g_k^a
=
\max\left\{
\epsilon,
\min\left(1-\epsilon,\widehat g_k^{a,\mathrm{raw}}\right)
\right\},$
and $\widehat\lambda_k^a
=
\max\left\{0,
\min\left(1,\widehat\lambda_k^{a,\mathrm{raw}}\right)
\right\}.$

Since $\widehat g_k^a\le 1-\epsilon$, we have $1-\widehat g_k^a(t_j\mid L,M)\ge \epsilon.$
Hence, if
$\widehat G_k^a(t_{l-1}\mid L,M)
=
\prod_{t_j\le t_{l-1}}
\left\{
1-\widehat g_k^a(t_j\mid L,M)
\right\},$
then
$\widehat G_k^a(t_{l-1}\mid L,M)
\ge
\epsilon^{N_l},$
where $N_l$ denotes the number of time points entering the product.
Moreover, since $\widehat\lambda_k^a\in[0,1]$,
$
0\le
\prod_{t_l<t_j\le\nu}
\left\{
1-\widehat\lambda_k^a(t_j\mid L,M)
\right\}
\le 1.
$
Therefore,
$
\left\|
\frac{
\prod_{t_l<t_j\le\nu}
\{1-\widehat\lambda_k^a(t_j\mid L,M)\}
}{
\widehat G_k^a(t_{l-1}\mid L,M)
}
\right\|_\infty
\le
\epsilon^{-N_l},
$
and consequently
$
\left\|
\frac{  
\prod_{t_l<t_j\le\nu}
\{1-\widehat\lambda_k^a(t_j\mid L,M)\}
}{
\widehat G_k^a(t_{l-1}\mid L,M)
}
\right\|_2
=
O(1).
$
Moreover,
\begin{align*}
\left\|
\widehat{\bar b}-\widehat b
\right\|_2
&\leq
\left\|
\widehat{\bar b}-b_0
\right\|_2
+
\left\|
\widehat b-b_0
\right\|_2
\\
&\lesssim
\left\|
\widehat S-S_0
\right\|_2
+
\left\|
\widehat b-b_0
\right\|_2,
\end{align*}
where the last inequality follows from the contraction property of
conditional expectation in $L_2(P_0)$. Hence,
\begin{align*}
|R_2(\widehat\eta,\eta_0)|
\lesssim{}&
\left\|
\widehat W-W_0
\right\|_{2}
\left\{
\left\|
\widehat b-b_0
\right\|_{2}
+
\left\|
\widehat S-S_0
\right\|_{2}
\right\}
\\
&+
\left\|
\widehat\Omega-\Omega_0
\right\|_{2}
\left\|
\widehat S-S_0
\right\|_{2}
\\
&+
\sum_{t_l\leq \nu}
\left\|
\widehat G(t_{l-1})-G_0(t_{l-1})
\right\|_{2}
\left\|
\widehat\lambda(t_l)-\lambda_0(t_l)
\right\|_{2}.
\end{align*}

Therefore, sufficient conditions for
$R_2(\widehat\eta,\eta_0)=o_{\mathbb P}(n^{-1/2})$ are
\begin{align*}
\left\|
\widehat W-W_0
\right\|_{2}
\left\{
\left\|
\widehat b-b_0
\right\|_{2}
+
\left\|
\widehat S-S_0
\right\|_{2}
\right\}
&=
o_{\mathbb P}(n^{-1/2}),
\\
\left\|
\widehat\Omega-\Omega_0
\right\|_{2}
\left\|
\widehat S-S_0
\right\|_{2}
&=
o_{\mathbb P}(n^{-1/2}),
\\
\sum_{t_l\leq\nu}
\left\|
\widehat G(t_{l-1})-G_0(t_{l-1})
\right\|_{2}
\left\|
\widehat\lambda(t_l)-\lambda_0(t_l)
\right\|_{2}
&=
o_{\mathbb P}(n^{-1/2}).
\end{align*}

\begin{align*}
&\Omega^{a,a^{*}}_{k,p}(L,M)
-
\widehat{\Omega}^{a,a^{*}}_{k,p}(L,M)
\\[0.5em]
&\quad=
\frac{
\tau^{a^*}_{p}(L,M)
V_p(L,M)
U_0(L)
\,
\widehat{\tau}^{a}_{k}(L,M)
\widehat V_k(L,M)
\widehat\pi^{a^*}_{p}(L)
\widehat U_p(L)
\widehat\rho^{a^*}_{p}(L)
}{
\tau^{a}_{k}(L,M)
V_k(L,M)
\pi^{a^*}_{p}(L)
U_p(L)
\rho^{a^*}_{p}(L)
\,
\widehat{\tau}^{a}_{k}(L,M)
\widehat V_k(L,M)
\widehat\pi^{a^*}_{p}(L)
\widehat U_p(L)
\widehat\rho^{a^*}_{p}(L)
}
\\[0.5em]
&\qquad-
\frac{
\widehat{\tau}^{a^*}_{p}(L,M)
\widehat V_p(L,M)
\widehat U_0(L)
\,
\tau^{a}_{k}(L,M)
V_k(L,M)
\pi^{a^*}_{p}(L)
U_p(L)
\rho^{a^*}_{p}(L)
}{
\tau^{a}_{k}(L,M)
V_k(L,M)
\pi^{a^*}_{p}(L)
U_p(L)
\rho^{a^*}_{p}(L)
\,
\widehat{\tau}^{a}_{k}(L,M)
\widehat V_k(L,M)
\widehat\pi^{a^*}_{p}(L)
\widehat U_p(L)
\widehat\rho^{a^*}_{p}(L)
}
\\[0.5em]
&\quad=
\frac{
\begin{aligned}
&
\tau^{a^*}_{p}(L,M)
V_p(L,M)
U_0(L)
\,
\widehat{\tau}^{a}_{k}(L,M)
\widehat V_k(L,M)
\widehat\pi^{a^*}_{p}(L)
\widehat U_p(L)
\widehat\rho^{a^*}_{p}(L)
\\
&\quad-
\widehat{\tau}^{a^*}_{p}(L,M)
\widehat V_p(L,M)
\widehat U_0(L)
\,
\tau^{a}_{k}(L,M)
V_k(L,M)
\pi^{a^*}_{p}(L)
U_p(L)
\rho^{a^*}_{p}(L)
\end{aligned}
}{
\begin{aligned}
&
\tau^{a}_{k}(L,M)
V_k(L,M)
\pi^{a^*}_{p}(L)
U_p(L)
\rho^{a^*}_{p}(L)
\\
&\qquad\times
\widehat{\tau}^{a}_{k}(L,M)
\widehat V_k(L,M)
\widehat\pi^{a^*}_{p}(L)
\widehat U_p(L)
\widehat\rho^{a^*}_{p}(L)
\end{aligned}
} \\
&\Omega^{a,a^*}_{k,p}(L,M)
-\widehat{\Omega}^{a,a^*}_{k,p}(L,M)
\\
&=
\frac{1}{
\tau_k^a V_k \pi_p^{a^*} U_p \rho_p^{a^*}
\,
\widehat\tau_k^a \widehat V_k
\widehat\pi_p^{a^*}\widehat U_p\widehat\rho_p^{a^*}}
\\
&\quad\times
\Bigg[
\widehat\tau_k^a \widehat V_k
\widehat\pi_p^{a^*}\widehat U_p\widehat\rho_p^{a^*}
\Big\{
(\tau_p^{a^*}-\widehat\tau_p^{a^*})V_pU_0
\\
&\qquad\qquad\qquad\qquad
+\widehat\tau_p^{a^*}(V_p-\widehat V_p)U_0
+\widehat\tau_p^{a^*}\widehat V_p(U_0-\widehat U_0)
\Big\}
\\
&\qquad+
\widehat\tau_p^{a^*}\widehat V_p\widehat U_0
\Big\{
(\widehat\tau_k^a-\tau_k^a)V_k\pi_p^{a^*}U_p\rho_p^{a^*}
\\
&\qquad\qquad
+\widehat\tau_k^a(\widehat V_k-V_k)\pi_p^{a^*}U_p\rho_p^{a^*}
\\
&\qquad\qquad
+\widehat\tau_k^a\widehat V_k
(\widehat\pi_p^{a^*}-\pi_p^{a^*})U_p\rho_p^{a^*}
\\
&\qquad\qquad
+\widehat\tau_k^a\widehat V_k\widehat\pi_p^{a^*}
(\widehat U_p-U_p)\rho_p^{a^*}
\\
&\qquad\qquad
+\widehat\tau_k^a\widehat V_k\widehat\pi_p^{a^*}
\widehat U_p
(\widehat\rho_p^{a^*}-\rho_p^{a^*})
\Big\}
\Bigg].
\end{align*}

By the triangle inequality, and since all nuisance functions take
values in $[0,1]$, the numerator of the preceding expression is
bounded in absolute value by
\begin{align*}
&
\left\|
\tau_p^{a^*}-\widehat\tau_p^{a^*}
\right\|_2
+
\left\|
V_p-\widehat V_p
\right\|_2
+
\left\|
U_0-\widehat U_0
\right\|_2
\\
&\qquad+
\left\|
\tau_k^a-\widehat\tau_k^a
\right\|_2
+
\left\|
V_k-\widehat V_k
\right\|_2
+
\left\|
\pi_p^{a^*}-\widehat\pi_p^{a^*}
\right\|_2
\\
&\qquad+
\left\|
U_p-\widehat U_p
\right\|_2
+
\left\|
\rho_p^{a^*}-\widehat\rho_p^{a^*}
\right\|_2.
\end{align*}

Under Assumption A, let $0<\varepsilon<c$ be fixed. For every nuisance estimator
$\widehat\eta^{\,\mathrm{raw}}$ entering
$\widehat\Omega_{k,p}^{a,a^*}$, define its truncated version as
$
\widehat\eta
=
\max\left\{
\varepsilon,
\min\left(1,\widehat\eta^{\,\mathrm{raw}}\right)
\right\}.$
For simplicity of notation, all nuisance estimators appearing below
refer to these truncated estimators.

Since the corresponding true nuisance functions take values in
$[c,1]\subset[\varepsilon,1]$, truncation does not increase the
estimation error. More precisely,
$
\left|
\widehat\eta-\eta_0
\right|
\le
\left|
\widehat\eta^{\,\mathrm{raw}}-\eta_0
\right|$
pointwise, and therefore
$
\left\|
\widehat\eta-\eta_0
\right\|_2
\le
\left\|
\widehat\eta^{\,\mathrm{raw}}-\eta_0
\right\|_2.$
Consequently, any $L_2(P_0)$ convergence rate satisfied by the raw
estimator is preserved after truncation.

Consider now
$
\Omega^{a,a^*}_{k,p}(L,M)
-
\widehat\Omega^{a,a^*}_{k,p}(L,M).$
Using a common denominator, we obtain
\begin{align*}
&\Omega^{a,a^*}_{k,p}(L,M)
-
\widehat\Omega^{a,a^*}_{k,p}(L,M)
\\
&=
\frac{1}{
\tau_k^a V_k \pi_p^{a^*} U_p \rho_p^{a^*}
\,
\widehat\tau_k^a \widehat V_k
\widehat\pi_p^{a^*}\widehat U_p
\widehat\rho_p^{a^*}}
\\
&\quad\times
\Bigg[
\widehat\tau_k^a \widehat V_k
\widehat\pi_p^{a^*}\widehat U_p
\widehat\rho_p^{a^*}
\Big\{
(\tau_p^{a^*}-\widehat\tau_p^{a^*})V_pU_0
\\
&\qquad\qquad\qquad
+
\widehat\tau_p^{a^*}(V_p-\widehat V_p)U_0
+
\widehat\tau_p^{a^*}\widehat V_p
(U_0-\widehat U_0)
\Big\}
\\
&\qquad+
\widehat\tau_p^{a^*}\widehat V_p\widehat U_0
\Big\{
(\widehat\tau_k^a-\tau_k^a)
V_k\pi_p^{a^*}U_p\rho_p^{a^*}
\\
&\qquad\qquad
+
\widehat\tau_k^a(\widehat V_k-V_k)
\pi_p^{a^*}U_p\rho_p^{a^*}
\\
&\qquad\qquad
+
\widehat\tau_k^a\widehat V_k
(\widehat\pi_p^{a^*}-\pi_p^{a^*})
U_p\rho_p^{a^*}
\\
&\qquad\qquad
+
\widehat\tau_k^a\widehat V_k
\widehat\pi_p^{a^*}
(\widehat U_p-U_p)\rho_p^{a^*}
\\
&\qquad\qquad
+
\widehat\tau_k^a\widehat V_k
\widehat\pi_p^{a^*}\widehat U_p
(\widehat\rho_p^{a^*}-\rho_p^{a^*})
\Big\}
\Bigg].
\end{align*}

Define $D
=
\tau_k^a V_k \pi_p^{a^*} U_p \rho_p^{a^*}$
and $\widehat D
=
\widehat\tau_k^a
\widehat V_k
\widehat\pi_p^{a^*}
\widehat U_p
\widehat\rho_p^{a^*}.$
Under the strong positivity assumption, $D\ge c^5.$
Moreover, by construction of the truncated nuisance estimators,$\widehat D\ge\varepsilon^5$
for every realization of the training sample and every value of the
covariates. Hence,
$
\left\|
(D\widehat D)^{-1}
\right\|_\infty
\le
c^{-5}\varepsilon^{-5}.$ Let $N$ denote the numerator in the preceding decomposition. Since all
true and truncated nuisance functions take values in $[0,1]$, the
triangle inequality yields
\begin{align*}
\|N\|_2
\le{}&
\left\|
\tau_p^{a^*}-\widehat\tau_p^{a^*}
\right\|_2
+
\left\|
V_p-\widehat V_p
\right\|_2
+
\left\|
U_0-\widehat U_0
\right\|_2
\\
&+
\left\|
\tau_k^a-\widehat\tau_k^a
\right\|_2
+
\left\|
V_k-\widehat V_k
\right\|_2
+
\left\|
\pi_p^{a^*}-\widehat\pi_p^{a^*}
\right\|_2
\\
&+
\left\|
U_p-\widehat U_p
\right\|_2
+
\left\|
\rho_p^{a^*}-\widehat\rho_p^{a^*}
\right\|_2.
\end{align*}

Therefore,
\begin{align*}
&
\left\|
\widehat\Omega_{k,p}^{a,a^*}
-
\Omega_{k,p}^{a,a^*}
\right\|_2
\\
&\qquad\le
\left\|
(D\widehat D)^{-1}
\right\|_\infty
\|N\|_2
\\
&\qquad\le
c^{-5}\varepsilon^{-5}
\Bigg[
\left\|
\tau_p^{a^*}-\widehat\tau_p^{a^*}
\right\|_2
+
\left\|
V_p-\widehat V_p
\right\|_2
+
\left\|
U_0-\widehat U_0
\right\|_2
\\
&\qquad\qquad+
\left\|
\tau_k^a-\widehat\tau_k^a
\right\|_2
+
\left\|
V_k-\widehat V_k
\right\|_2
+
\left\|
\pi_p^{a^*}-\widehat\pi_p^{a^*}
\right\|_2
\\
&\qquad\qquad+
\left\|
U_p-\widehat U_p
\right\|_2
+
\left\|
\rho_p^{a^*}-\widehat\rho_p^{a^*}
\right\|_2
\Bigg].
\end{align*}

Suppose that each raw nuisance estimator entering
$\widehat\Omega_{k,p}^{a,a^*}$ satisfies
$\left\|
\widehat\eta^{\,\mathrm{raw}}-\eta_0
\right\|_2
=
o_{\mathbb P}(n^{-1/4}).$
Since truncation preserves the convergence rate,$\left\|
\widehat\eta-\eta_0
\right\|_2
=
o_{\mathbb P}(n^{-1/4})$
for each corresponding truncated nuisance estimator. Since only
finitely many nuisance functions enter the preceding bound, it follows
that
\begin{align*}
&
\left\|
\tau_p^{a^*}-\widehat\tau_p^{a^*}
\right\|_2
+
\left\|
V_p-\widehat V_p
\right\|_2
+
\left\|
U_0-\widehat U_0
\right\|_2
\\
&\quad+
\left\|
\tau_k^a-\widehat\tau_k^a
\right\|_2
+
\left\|
V_k-\widehat V_k
\right\|_2
+
\left\|
\pi_p^{a^*}-\widehat\pi_p^{a^*}
\right\|_2
\\
&\quad+
\left\|
U_p-\widehat U_p
\right\|_2
+
\left\|
\rho_p^{a^*}-\widehat\rho_p^{a^*}
\right\|_2
=
o_{\mathbb P}(n^{-1/4}).
\end{align*}
Consequently,
$\left\|
\widehat\Omega_{k,p}^{a,a^*}
-
\Omega_{k,p}^{a,a^*}
\right\|_2
=
o_{\mathbb P}(n^{-1/4}).$

An analogous argument applies to $\widehat W_p^{a^*}$. Provided that
every nuisance function appearing in its denominator is uniformly
bounded below by $c$, and that the corresponding nuisance estimators
are truncated below at $\varepsilon<c$, the estimated denominator is
uniformly bounded away from zero. Hence, under the same
$L_2(P_0)$ convergence-rate conditions,
$\left\|
\widehat W_p^{a^*}
-W_{0,p}^{a^*}
\right\|_2
=o_{\mathbb P}(n^{-1/4}).$

Then, under Assumption~A and appropriate clipping, convergence rate $o_{\mathbb{P}}(n^{-1/4})$ is sufficient to have a remainder term that satisfies a convergence rate $o_{\mathbb{P}}(n^{-1/2})$

\end{proof}
\section{Multiple Mediators}

\subsection{Identification}
The path-specific indirect effects are defined as follows.  

\begin{align*}
    \operatorname{NIE}_{A->M_1->M_2->T}^{a}(\nu) &=&
\Pr\!\left\{T\bigl(a,M_1(a^*), M_2(a^{*} ,M_1(a^{*}))\bigr)>\nu\mid S=0\right\} \\
&-&
\Pr\!\left\{T\bigl(a,M_1(a), M_2(a ,M_1(a)\bigr)>\nu\mid S=0\right\} \\
    \operatorname{NIE}_{A->M_2->T} ^{a}(\nu) &=&
\Pr\!\left\{T\bigl(a, M_1(a^*), M_2(a, M_1(a^{*}))\bigr)>\nu\mid S=0\right\} \\
&-&
\Pr\!\left\{T\bigl(a, M_1(a^*), M_2(a^* ,M_1(a^*))\bigr)>\nu\mid S=0\right\} \\
    \operatorname{NIE}_{A->M_1T} ^{a}(\nu) &=&
\Pr\!\left\{T\bigl(a, M_1(a), M_2(a, M_1(a))\bigr)>\nu\mid S=0\right\} \\
&-&
\Pr\!\left\{T\bigl(a, M_1(a^*), M_2(a,M_1(a^*))\bigr)>\nu\mid S=0\right\}.
\end{align*}

We need to identify  two counterfactual quantities   $$\theta^{a, a^{*}}(\nu):=\mathbb{P}[T\bigl(a,M_1(a^*), M_2(a^{*} ,M_1(a^{*})) \bigr)>\nu\mid S=0] \text {  and  } $$$$\beta^{a, a^{*}}(\nu):=\mathbb{P}[T\bigl(a,M_1(a^*), M_2(a ,M_1(a^{*})) \bigr)>\nu\mid S=0]$$

\renewcommand{\theassumption}{\arabic{assumption}bis}
\setcounter{assumption}{0}
\begin{assumption}\label{ass:nie2}(Identification of $\nie^{a}(\nu)$) \begin{itemize} 
\item [(i)] (Ignorability) $T(a), M_1(a), M_2(a,m) \perp A \mid L, S$ and $M_2(a,m) \perp M_1(a)   \mid L, S, A $ \item [(ii)] (Consistency) $T = T(a)$ and $M_1 = M_1(a),$  if $A=a$  and $M_2 = M_2(a,m)$ if $A=a$ and $M_1=m$  \item [(iii)] (Cross-world independence) $T(a)\ \perp\ ( M_1(a^*), M_2(a^*))\mid L, S.$ \item [(iv)] (Positivity) $0<\mathbb{P}(A=a\mid L,S), \mathbb{P}(M_1=m\mid A,L,S),\mathbb{P}(M_2=m\mid A,L,S),\mathbb{P}(S=s\mid L), \mathbb{P}(C \ge t_i \mid A, M_1,M_2, L,S=k) < 1,$ for all relevant $(a,m_1,m_2,s)$, $k = 1,\ldots K$ and $i=1,\ldots,\tau$. 
\end{itemize} 
\end{assumption}

\begin{assumption}[Transportability]
\label{ass:transportability}
$T(a,m_1, m_2)\perp S\mid L, \quad M_1(a)\perp S\mid L,\quad M_2(a, m_1)\perp S\mid L.$
\end{assumption}

\begin{assumption}\label{ass:mar}(Missing at Random)
\begin{itemize}
\item [(i)] $R\perp(M_1,M_2,T,C)\mid A,L,S=k,
\qquad
k=1,\ldots,K,$
\item [(ii)]$\tilde{R} \perp(M_2,T,C)\mid M_1,R,A,L,S=k,
\qquad
k=1,\ldots,K,$ 
\end{itemize}
where $R$ (respectively, $\tilde{R}$) denotes the missingness indicator for the mediator $M_1$ (respectively, $M_2$), with $R=1$ ($\tilde{R}=1$) indicating that $M_1$ ($M_2$) is missing and $R=0$ ($\tilde{R}=0$) indicating that it is observed. The missingness mechanism is assumed to be monotone, in a sense that if the mediator $M_1$ is missing then the mediator $M_2$ is missing.
\end{assumption}

Under assumptions~\ref{ass:nie2}--\ref{ass:mar}, the counterfactual survival probability $\theta^{a, a^{*}}(\nu)$ and  $\beta^{a, a^{*}}(\nu)$ are identified by the functional $\beta^{a,a^*}_{k,p, q}(\nu)$ and $\theta^{a,a^*}_{k,p, q}(\nu)$, where:
\begin{eqnarray*}\label{G-formula1}
    \theta^{a, a^{*}}_{k,p,s} &=& \mathbb{E}\biggl[\mathbb{E} \bigl[ \mathbb{P}(T>\nu\mid a,M_1, M_2, \tilde{R}=0, L,k)\mid a^*, M_1, L,p,  \tilde{R}=0 \bigr]  \mid a^*, L, s, R=0\bigr] \mid S=0\biggr] \\
    \beta^{a, a^{*}}_{k,p,s} &=& \mathbb{E}\biggl[\mathbb{E} \bigl[ \mathbb{P}(T>\nu\mid a,M_1, M_2, \tilde{R}=0, L,k)\mid a, M_1, L,p,  \tilde{R}=0 \bigr]  \mid a^*, L, s, R=0\bigr] \mid S=0\biggr]\\
\end{eqnarray*} 
\subsection{ANOVA decomposition of heterogeneity}

Specifically, let $\zeta = \sum_{k,p} w_kq_p \gamma_s\zeta_{k,p,s}$
and $\tau^2 = \sum_{k,p,s}w_kq_p\gamma_s(\zeta_{k,p,s}-\zeta)^2,$ 
where $(w_k)$ and $(q_p)$ are prespecified weight functions for outcome and mediator sources, with $\sum_k w_k=1 ,\sum_p q_p=1$ and $\sum_s \gamma_s=1$ . The total variance $\tau^2$ can then be decomposed as:
\begin{eqnarray*} \label{eq:anova_2}
    \tau^2 &=&
    \underbrace{\sum_k w_k(\zeta_{k,\cdot, \cdot}-\zeta)^2}_{ \xi^2 \text{Outcome-related variation}}
    +
    \underbrace{\sum_p q_p(\zeta_{\cdot,p,\cdot,}-\zeta)^2}_{ \eta^2 \text{Mediator 1-related variation}}+ \underbrace{\sum_s \gamma_s(\zeta_{\cdot,\cdot,s}-\zeta)^2}_{ \alpha^2 \text{Mediator 2-related variation}}\\
   & + & \underbrace{
\sum_{k,p} w_k q_p
\left(
\zeta_{k,p,\cdot}
-\zeta_{k,\cdot,\cdot}
-\zeta_{\cdot,p,\cdot}
+\zeta
\right)^2
}_{\text{Outcome} \times \text{Mediator 1 interaction}}
\nonumber + \underbrace{
\sum_{k,s} w_k \gamma_s
\left(
\zeta_{k,\cdot,s}
-\zeta_{k,\cdot,\cdot}
-\zeta_{\cdot,\cdot,s}
+\zeta
\right)^2
}_{\text{Outcome} \times \text{Mediator 2 interaction}} \\
& +&
\underbrace{
\sum_{p,s} q_p \gamma_s
\left(
\zeta_{\cdot,p,s}
-\zeta_{\cdot,p,\cdot}
-\zeta_{\cdot,\cdot,s}
+\zeta
\right)^2
}_{\text{Mediator 1} \times \text{Mediator 2 interaction}} \\
&+&
\underbrace{
\sum_{k,p,s} w_k q_p \gamma_s
\left(
\zeta_{k,p,s}
-\zeta_{k,p,\cdot}
-\zeta_{k,\cdot,s}
-\zeta_{\cdot,p,s}
+\zeta_{k,\cdot,\cdot}
+\zeta_{\cdot,p,\cdot}
+\zeta_{\cdot,\cdot,s}
-\zeta
\right)^2
}_{\text{Outcome} \times \text{Mediator 1} \times \text{Mediator 2 interaction}}.
\end{eqnarray*}

where $\zeta_{k,\cdot,\cdot}= \sum_{p} \sum_{s} q_p \gamma_s \zeta_{k,p,s}$ and $\zeta_{\cdot,p,\cdot}=\sum_{k} \sum_{s} w_k \gamma_s \zeta_{k,p,s}$, $\zeta_{\cdot,\cdot,s}=\sum_{k} \sum_{p} w_k q_p \zeta_{k,p,s}$ , $\zeta_{k,p,\cdot}= \sum_{s}  \gamma_s \zeta_{k,p,s}$ , $\zeta_{k,\cdot,s}= \sum_{p}  q_p w_k \zeta_{k,p,s},$  $\zeta_{\cdot,p,s}= \sum_{k}  \omega_k q_p \zeta_{k,p,s}$ 

These interaction terms quantify the extent to which heterogeneity attributable to one source depends on the choice of another source. In particular, the two-way interactions capture pairwise departures from additivity, while the three-way interaction represents the residual variation that cannot be explained by the three main effects or the pairwise interactions alone.

\section{Shared Nuisance Parameters Across Sites: One-Stage Meta-Analysis}

This section presents alternative identification results and estimation strategies for data-fusion meta-analysis under progressively stronger assumptions on the censoring and missingness mechanisms. In contrast to the main analysis, where no restriction is imposed on the way these mechanisms vary across study sites, the approaches considered here rely on assumptions that allow some degree of information sharing across studies. These additional assumptions can improve identification and efficiency, but they require stronger modeling restrictions.

More specifically, the extent to which the missingness mechanisms are assumed to be shared across studies defines a hierarchy of possible settings. At one extreme, neither the outcome nor the mediator missingness mechanisms are assumed to be common across studies. This corresponds to the most flexible framework, which is the setting developed in the main article. At the other extreme, one may assume that both the outcome and mediator mechanisms and censoring and missingness mechanisms are shared across all study sites.  This leads to a one-stage meta-analysis framework, in which information is pooled across studies through common missingness models. We develop this setting in the following section.

Between these two extremes lie intermediate cases. For example, one may assume that the outcome mechanism is common across studies while allowing the mediator mechanism to remain study-specific (or conversely). In such cases, identification relies on a partially pooled structure, resulting in the estimation of (K) distinct estimands corresponding to the study-specific component that remains unrestricted.

The intermediate scenarios can be naturally derived from the two limiting cases described above. They therefore provide a flexible framework for incorporating varying degrees of cross-study information sharing while maintaining a clear connection with the fully study-specific and fully pooled approaches.

\subsection{Identification}
\label{proof_identification_shared}

\begin{align*}
\mathbb{E}\left[
    I\{T(a,M(a^{*}))>\nu\}
    \mid S=0
\right]&
\\
\quad=
\mathbb{E}_{L}\left[
    \mathbb{E}\left[
        I\{T(a,M(a^{*}))>\nu\}
        \mid L
    \right]
    \,\middle|\, S=0
\right]&
\\
\intertext{\textit{Cross-world independence:}}
\quad=
\mathbb{E}_{L}\Biggl[
    \sum_{m\in\mathcal{M}}
    \mathbb{E}\left[
        I\{T(a,M(a^{*}))>\nu\}
        \,\middle|\,
        L,M(a^{*})=m
    \right]
\nonumber\\[-0.2em]
\hspace{12em}\times
    \mathbb{P}\left(
        M(a^{*})=m
        \mid L
    \right)
    \,\Biggm|\, S=0
\Biggr]&
\\
\intertext{\textit{Outcome transportability:}}
\quad=
\mathbb{E}_{L}\Biggl[
    \sum_{m\in\mathcal{M}}
    \mathbb{E}\left[
        I\{T(a,M(a^{*}))>\nu\}
        \,\middle|\,
        \substack{L,M(a^{*})=m,\\ S>0}
    \right]
\nonumber&\\[-0.2em]
\hspace{12em}\times
    \mathbb{P}\left(
        M(a^{*})=m
        \mid L
    \right)
    \,\Biggm|\, S=0
\Biggr]&
\\
\intertext{\textit{Conditional outcome exchangeability:}}
\quad=
\mathbb{E}_{L}\Biggl[
    \sum_{m\in\mathcal{M}}
    \mathbb{E}\left[
        I\{T(a,M(a^{*}))>\nu\}
        \,\middle|\,
        \substack{L,M(a^{*})=m,\\ A=a,S>0}
    \right]&
\nonumber \\[-0.2em]
\hspace{12em}\times
    \mathbb{P}\left(
        M(a^{*})=m
        \mid L
    \right)
    \,\Biggm|\, S=0
\Biggr]&
\\
\intertext{\textit{Outcome consistency:}}
\quad=
\mathbb{E}_{L}\Biggl[
    \sum_{m\in\mathcal{M}}
    \mathbb{E}\left[
        I\{T>\nu\}
        \,\middle|\,
        L,M=m,A=a,S>0
    \right]&
\nonumber\\[-0.2em]
\hspace{12em}\times
    \mathbb{P}\left(
        M(a^{*})=m
        \mid L
    \right)
    \,\Biggm|\, S=0
\Biggr]&
\\
\intertext{\textit{Mediator transportability:}}
\quad=
\mathbb{E}_{L}\Biggl[
    \sum_{m\in\mathcal{M}}
    \mathbb{E}\left[
        I\{T>\nu\}
        \,\middle|\,
        L,M=m,A=a,S>0
    \right]&
\nonumber\\[-0.2em]
\hspace{12em}\times
    \mathbb{P}\left(
        M(a^{*})=m
        \mid L,S>0
    \right)
    \,\Biggm|\, S=0
\Biggr]&
\\
\intertext{\textit{Conditional mediator exchangeability:}}
\quad=
\mathbb{E}_{L}\Biggl[
    \sum_{m\in\mathcal{M}}
    \mathbb{E}\left[
        I\{T>\nu\}
        \,\middle|\,
        L,M=m,A=a,S>0
    \right]&
\nonumber\\[-0.2em]
\hspace{12em}\times
    \mathbb{P}\left(
        M(a^{*})=m
        \mid L,A=a^{*},S>0
    \right)
    \,\Biggm|\, S=0
\Biggr]&
\\
\intertext{\textit{Mediator consistency:}}
\quad=
\mathbb{E}_{L}\Biggl[
    \sum_{m\in\mathcal{M}}
    \mathbb{E}\left[
        I\{T>\nu\}
        \,\middle|\,
        L,M=m,A=a,S>0
    \right]&
\nonumber\\[-0.2em]
\hspace{12em}\times
    \mathbb{P}\left(
        M=m
        \mid L,A=a^{*},S>0
    \right)
    \,\Biggm|\, S=0
\Biggr] & \\
\quad=
    \sum_{(m,l)\in\mathcal{M} \times \mathcal{L}}
    \mathbb{E}\left[
        I\{T>\nu\}
        \,\middle|\,
        L=l,M=m,A=a,S>0
    \right]
\nonumber\\[-0.2em]
\hspace{12em}\times
    \mathbb{P}\left(
        M=m
        \mid L=l,A=a^{*},S>0    \right)
    \mathbb{P}\left( L=l\,|\, S=0  \right).&
\end{align*}

\medskip
\noindent
\textbf{Assumption 3bis (Shared Missing-at-Random Mechanism).}
\emph{
Among the contributing sites, $R \perp (M,T,C)\mid A,L,S>0,$
where $R$ denotes the missingness indicator for the mediator, with
$R=1$ if $M$ is missing and $R=0$ otherwise. In particular,
\[
\mathbb{P}(R=0\mid A,L,S=k)
=
\mathbb{P}(R=0\mid A,L,S=k')
\]
for all $k,k'\in\{1,\ldots,K\}$.
}
\medskip

\noindent
\textbf{Assumption 4bis (Shared Conditional Independent Censoring).}
\emph{
Among individuals with an observed mediator, $C\perp T \mid A,L,M,R=0,S>0.$
In particular,
\[
\mathbb{P}(C>t\mid A,L,M,R=0,S=k)
=
\mathbb{P}(C>t\mid A,L,M,R=0,S=k')
\]
for all $t$ and all $k,k'\in\{1,\ldots,K\}$.
}

\begin{remark}
Assumption 3bis can be formulated in alternative ways depending on the type of information sharing assumed across study sites. For instance, one may consider the following pair of assumptions:
\[
R \perp (T,C)\mid A,L,S>0,
\]
and
\[
R \perp M\mid A,L,S=k,\quad \text{for all } k.
\]

Under this specification, the missingness mechanism for the mediator is shared across studies, while the mediator distribution may remain study-specific. In particular, this implies that
\[
\mathbb{P}(R=0\mid A,L,S=k)
=
\mathbb{P}(R=0\mid A,L,S=k')
\]
for all study sites $k$ and $k'$. However, this assumption does not imply that the mediator distribution among individuals with observed mediator values is representative of the full mediator distribution. Specifically, in general,
\[
\mathbb{P}(M\mid A,L,S>0)
\neq
\mathbb{P}(M\mid A,L,S>0,R=0).
\]

Therefore, the estimand considered in the main analysis is no longer directly identified under this alternative assumption. Instead, one needs to define a study-specific estimand, given by
\begin{align*}
\theta_p &=
\sum_{(m,l)\in\mathcal{M}\times\mathcal{L}}
\mathbb{E}\left[
I\{T>\nu\}
\,\middle|\,
L=l,M=m,A=a,S>0
\right]
\mathbb{P}\left(
M=m
\mid L=l,A=a^{*},S=p
\right)  \\ & \qquad \qquad \times
\mathbb{P}\left(
L=l\mid S=0
\right) 
\end{align*}
where $p$ denotes the study site providing information on the mediator distribution.

Alternatively, one may recover the mediator distribution in the target population using the law of total probability:
\[
\mathbb{P}(M\mid A,L,S>0)
=
\sum_p
\mathbb{P}(M\mid A,L,R=0,S=p)
\mathbb{P}(S=p\mid A,L).
\]

This formulation illustrates that partial sharing of missingness mechanisms leads to intermediate identification strategies between the fully study-specific framework considered in the main article and the fully pooled one-stage meta-analysis setting.
\end{remark}

\begin{align*}
\sum_{(m,l)\in\mathcal{M} \times \mathcal{L}} &
    \mathbb{E}\left[
        I\{T>\nu\}
        \,\middle|\,
        L=l,M=m,A=a,S>0
    \right]
\nonumber\\[-0.2em]
&\hspace{12em}\times
    \mathbb{P}\left(
        M=m
        \mid L=l,A=a^{*},S>0    \right)
    \mathbb{P}\left( L=l\,|\, S=0  \right).
\end{align*}

By Assumption 3bis, we have:
$$\mathbb{P}\left(
        M=m \mid L, A=a^{*},S>0\right) = \mathbb{P}\left(M=m
        \mid R=0, L, A=a^{*},S>0\right).$$

As in the site-specific setting, we can develop IPCW, IPW, One-Step, and TMLE estimators for this estimand.

\subsection{Parametric Estimators}

We define the survival function, censoring survival function, and corresponding hazard functions on the population $S>0$:

\begin{align*}
g^a(t\mid L,M,R=0)
&=
\mathbb{P}\left(
    C=t
    \mid
    C\geq t,A=a,L,M,S>0,R=0
\right),
\\
G^a(t\mid L,M,R=0)
&=
\mathbb{P}\left(
    C>t
    \mid
    A=a,L,M,S>0,R=0
\right)
\nonumber\\
&=
\prod_{u\leq t}
\left\{
    1-g^a(u\mid L,M,R=0)
\right\}.
\end{align*}

Similarly, define the event hazard and survival functions as
\begin{align*}
\lambda^a(t\mid L,M,R=0)
&=
\mathbb{P}\left(
    T=t
    \mid
    T\geq t,A=a,L,M,S>0,R=0
\right),
\\
S^a(t\mid L,M,R=0)
&=
\mathbb{P}\left(
    T>t
    \mid
    A=a,L,M,S>0,R=0
\right)
\nonumber\\
&=
\prod_{u\leq t}
\left\{
    1-\lambda^a(u\mid L,M,R=0)
\right\}.
\end{align*}

From the identification result, we have the identifying functional:
\begin{align*}
\theta^{a,a^{*}}(P)
&=
\sum_{(m,l)\in\mathcal{M} \times \mathcal{L}}
S^{a}(\nu\mid l,m,R=0)
\,
\mathbb{P}\left(
    M=m
    \mid L=l,A=a^{*},R=0,S>0
\right)
\nonumber\\[-0.2em]
&\hspace{12em}\times
\mathbb{P}(L=l\mid S=0).
\end{align*}

For a continuous mediator, the sum over \(m\) is replaced by the
corresponding integral. This representation forms the basis of the parametric
estimators introduced below.

We now describe three parametric estimation strategies for
\(\theta^{a,a^{*}}(P)\). In contrast with the site-specific setting, the
event, censoring, and missingness nuisance parameters are shared across the
contributing sites and are therefore estimated using the pooled population
\(S>0\).

To simplify the presentation, let \(I(v)=I(V=v)\) denote the indicator
function for a generic variable \(V\), and define the collection of nuisance
parameters as
\[
\bm{\eta}
=
\left(
    U_s,V_k,\pi_k^{a},\tau_k^{a},\rho^{a},
    \lambda^{a}(t),g^{a}(t),S^{a}(t),G^{a}(t),
    \mu^{a^{*}},b^{a,a^{*}}
\right).
\]

For \(s\in\{0,\ldots,K\}\) and \(k\in\{1,\ldots,K\}\), the site
participation and treatment mechanisms are defined by
\begin{align*}
U_s(L)
&=
\mathbb{P}(S=s\mid L),
&
V_k(L,M)
&=
\mathbb{P}(S=k\mid L,M,R=0),
\\
\pi_k^{a}(L)
&=
\mathbb{P}(A=a\mid L,S=k),
&
\tau_k^{a}(L,M)
&=
\mathbb{P}(A=a\mid L,M,R=0,S=k).
\end{align*}
The shared missingness mechanism and mediator distribution are respectively
given by
\[
\rho^{a}(L)
=
\mathbb{P}(R=0\mid A=a,L,S>0)
\qquad\text{and}\qquad
\mu^{a^{*}}(m\mid L)
=
\mathbb{P}(M=m\mid L,A=a^{*},R=0,S>0).
\]
Finally, define
\[
b^{a,a^{*}}(L)
=
\mathbb{E}\left[
    S^{a}(\nu\mid L,M,R=0)
    \mid A=a^{*},L,R=0,S>0
\right].
\]

The quantities \(\rho^{a}\), \(g^{a}\), \(\lambda^{a}\), \(G^{a}\), and
\(S^{a}\) do not carry a site index because they are assumed to be shared
across all contributing sites.

\paragraph{Event and censoring nuisance parameters}

Within the pooled population \(S>0\), the conditional censoring hazard is
defined as
\[
g^a(t\mid L,M,R=0)
=
\mathbb{P}\left(
    C=t
    \mid C\geq t,A=a,L,M,S>0,R=0
\right).
\]
The corresponding censoring survival function is
\[
G^a(t\mid L,M,R=0)
=
\mathbb{P}\left(
    C>t
    \mid A=a,L,M,S>0,R=0
\right)
=
\prod_{u\leq t}
\left\{
    1-g^a(u\mid L,M,R=0)
\right\}.
\]

Similarly, the conditional event hazard is defined as
\[
\lambda^a(t\mid L,M,R=0)
=
\mathbb{P}\left(
    T=t
    \mid T\geq t,A=a,L,M,S>0,R=0
\right),
\]
and the corresponding event survival function is
\[
S^a(t\mid L,M,R=0)
=
\mathbb{P}\left(
    T>t
    \mid A=a,L,M,S>0,R=0
\right)
=
\prod_{u\leq t}
\left\{
    1-\lambda^a(u\mid L,M,R=0)
\right\}.
\]

\paragraph{Total treatment and participation probabilities}

The treatment mechanisms conditional on \(S>0\) can be expressed through
site-specific quantities. More precisely, for \(c\in\{a,a^{*}\}\), the law
of total probability gives
\begin{align*}
\sum_{k=1}^{K}
\pi_k^{c}(L)U_k(L)
&=
\sum_{k=1}^{K}
\mathbb{P}(A=c\mid L,S=k)
\mathbb{P}(S=k\mid L)
\\
&=
\mathbb{P}(A=c,S>0\mid L),
\end{align*}
while, conditional on the mediator,
\begin{align*}
\sum_{k=1}^{K}
\tau_k^{c}(L,M)V_k(L,M)
&=
\sum_{k=1}^{K}
\mathbb{P}(A=c\mid L,M,R=0,S=k)
\\[-0.2em]
&\hspace{7em}\times
\mathbb{P}(S=k\mid L,M,R=0)
\\
&=
\mathbb{P}(A=c,S>0\mid L,M,R=0).
\end{align*}

These total-probability sums are used directly in the weighting functions
below, rather than introducing separate propensity scores conditional on
\(S>0\).

\paragraph{Parametric plug-in estimator}

A first strategy consists of directly substituting estimators of the nuisance
parameters into the identifying functional. This yields the following
parametric representation:
\begin{align*}
\theta^{a,a^{*}}
&=
\mathbb{E}\left[
    \left.
    \sum_{m\in\mathcal{M}}
    S^{a}(\nu\mid L,m,R=0)
    \mu^{a^{*}}(m\mid L)
    \,\right\|_2\,S=0
\right]
\nonumber\\
&=
\mathbb{E}\left[
    b^{a,a^{*}}(L)
    \mid S=0
\right].
\end{align*}

In fact there are two ways of considering the estimation :

Let \(n_0=\sum_{i=1}^{n}I(S_i=0)\) denote the number of individuals in the
target population. The corresponding empirical plug-in estimator is
\begin{align*}
\widehat{\theta}_{\mathrm{G}}^{a,a^{*}}
&=
\frac{1}{n_0}
\sum_{i:S_i=0}
\sum_{m\in\mathcal{M}}
\widehat{S}^{a}(\nu\mid L_i,m,R=0)
\widehat{\mu}^{a^{*}}(m\mid L_i)
\nonumber\\
&=
\frac{1}{n_0}
\sum_{i:S_i=0}
\widehat{b}^{a,a^{*}}(L_i).
\end{align*}

\paragraph{IPCW estimator}

Alternatively, the parameter can be represented using inverse probability
weights that jointly account for transport, treatment, missingness, and
censoring. Define the pooled IPCW weight as
\begin{align*}
\Omega^{a,a^{*}}(L,M)
&=
\frac{
    U_0(L)
}{
    \rho^{a^{*}}(L)
    \displaystyle\sum_{k=1}^{K}
    \pi_k^{a^{*}}(L)U_k(L)
}
\nonumber\\[-0.2em]
&\quad\times
\frac{
    \displaystyle\sum_{k=1}^{K}
    \tau_k^{a^{*}}(L,M)V_k(L,M)
}{
    \displaystyle\sum_{k=1}^{K}
    \tau_k^{a}(L,M)V_k(L,M)
}.
\end{align*}

Using this weight, the IPCW representation of the target parameter is
\begin{align*}
\theta_{\mathrm{IPCW}}^{a,a^{*}}(P)
&=
\frac{1}{\mathbb{P}(S=0)}
\mathbb{E}\Biggl[
    \frac{
        \Omega^{a,a^{*}}(L,M)
    }{
        G^{a}(\nu\mid L,M,R=0)
    }
\nonumber\\[-0.2em]
&\hspace{7em}\times
    I(\widetilde{T}>\nu)
    I(A=a,S>0,R=0)
\Biggr].
\end{align*}

The corresponding empirical estimator pools the observations from all
contributing sites and is written as
\begin{align*}
\widehat{\theta}_{\mathrm{IPCW}}^{a,a^{*}}
&=
\frac{1}{n_0}
\sum_{k=1}^{K}
\sum_{i:S_i=k}
\frac{
    \widehat{\Omega}^{a,a^{*}}(L_i,M_i)
}{
    \widehat{G}^{a}(\nu\mid L_i,M_i,R_i=0)
}
\nonumber\\[-0.2em]
&\hspace{7em}\times
I(\widetilde{T}_i>\nu)
I(A_i=a,R_i=0),
\end{align*}
where
\begin{align*}
\widehat{\Omega}^{a,a^{*}}(L_i,M_i)
&=
\frac{
    \widehat{U}_0(L_i)
}{
    \widehat{\rho}^{a^{*}}(L_i)
    \displaystyle\sum_{j=1}^{K}
    \widehat{\pi}_j^{a^{*}}(L_i)
    \widehat{U}_j(L_i)
}
\nonumber\\[-0.2em]
&\quad\times
\frac{
    \displaystyle\sum_{j=1}^{K}
    \widehat{\tau}_j^{a^{*}}(L_i,M_i)
    \widehat{V}_j(L_i,M_i)
}{
    \displaystyle\sum_{j=1}^{K}
    \widehat{\tau}_j^{a}(L_i,M_i)
    \widehat{V}_j(L_i,M_i)
}.
\end{align*}

\paragraph{IPW estimator based on the event survival model}

A third representation replaces the observed survival indicator by the
estimated conditional event survival function. For this purpose, define the
pooled IPW weight as
\begin{align*}
W^{a,a^{*}}(L)
&=
\frac{
    U_0(L)
}{
    \rho^{a^{*}}(L)
    \displaystyle\sum_{k=1}^{K}
    \pi_k^{a^{*}}(L)U_k(L)
}.
\end{align*}

The resulting representation of the parameter is
\begin{align*}
\theta_{\mathrm{IPW}}^{a,a^{*}}(P)
&=
\frac{1}{\mathbb{P}(S=0)}
\mathbb{E}\Biggl[
    I(A=a^{*},S>0,R=0)
    W^{a,a^{*}}(L)
\nonumber\\[-0.2em]
&\hspace{11em}\times
    S^{a}(\nu\mid L,M,R=0)
\Biggr].
\end{align*}

Its empirical counterpart is
\begin{align*}
\widehat{\theta}_{\mathrm{IPW}}^{a,a^{*}}
&=
\frac{1}{n_0}
\sum_{k=1}^{K}
\sum_{i:S_i=k}
I(A_i=a^{*},R_i=0)
\widehat{W}^{a^{*}}(L_i)
\nonumber\\[-0.2em]
&\hspace{11em}\times
\widehat{S}^{a}(\nu\mid L_i,M_i,R_i=0),
\end{align*}
where
\[
\widehat{W}^{a^{*}}(L_i)
=
\frac{
    \widehat{U}_0(L_i)
}{
    \widehat{\rho}^{a^{*}}(L_i)
    \displaystyle\sum_{j=1}^{K}
    \widehat{\pi}_j^{a^{*}}(L_i)
    \widehat{U}_j(L_i)
}.
\]

Thus, the three estimators rely on the same collection of pooled nuisance
parameters but differ in how the event survival distribution is incorporated:
the plug-in estimator averages the fully modeled identifying functional, the
IPCW estimator uses the observed survival indicator corrected for censoring,
and the IPW estimator averages the modeled event survival function over the
weighted mediator distribution.

All denominators are assumed to be strictly positive on the support of the
target population.

\subsection{One-Step and TMLE Estimators}

The development of One-Step corrected and TMLE (Targeted Maximum Likelihood Estimation) estimators follows along similar lines and is omitted for brevity. The EIF in the nonparametric statistical model is the following:

\begin{align*}
\varphi(O)
&=
-
\frac{
    I\{A=a,S>0, R=0\}
}{
    \mathbb{P}(S=0)
}
\Omega^{a,a^{*}}(L,M)
\nonumber\\[-0.2em]
&\quad\times
\sum_{t_l\leq\nu}
I\{\widetilde T\geq t_l\}
\mathcal{G}_{a}^{(t_l)}(\nu\mid L,M)
\left\{
    I\{\widetilde T=t_l,\delta=1\}
    -
    \lambda_{a}(t_l\mid L,M)
\right\}
\nonumber\\[0.4em]
&\quad+
\frac{ I\{A=a*,S>0, R=0\}W^{a^* }(L)}{\mathbb{P}(S=0)}
\left\{
    S_{a}(\nu\mid L,M, R=0)
    -
    b^{a,a^{*}}(\nu\mid L)
\right\}
\nonumber\\[0.4em]
&\quad+
\frac{I\{S=0\}}{\mathbb{P}(S=0)}
\left\{
    b^{a,a^{*}}(\nu\mid L)
    -
    \theta^{a,a^*}(P)
\right\}.
\end{align*}

The One-Step estimator can be constructed following the same principle as in the main
paper, except that the shared event hazard $\lambda^a$ and the regression function
$b^{a,a^*}$ are estimated in a federated manner using data from all contributing
source populations, $S>0$.

The TMLE requires an additional modification because the shared event hazard must
itself be targeted in a federated manner. Since the source population is distributed
across $K$ servers, the empirical EIF equation involves contributions from all source
sites and cannot be solved locally at a single site. However, the corresponding
fluctuation likelihood decomposes additively across sites, allowing the targeting step
to be implemented through federated logistic regression.

For each time point $t_j \leq \nu$, the targeted update of the event hazard takes the
form
\begin{equation}
\operatorname{logit}
\left\{
\widehat{\lambda}^{a,(1)}
(t_j \mid L_i,M_i,R_i=0)
\right\}
=
\operatorname{logit}
\left\{
\widehat{\lambda}^{a}
(t_j \mid L_i,M_i,R_i=0)
\right\}
+
\epsilon^{(j)}
\widehat{w}^{(j)}_i,
\end{equation}
where $\widehat{w}^{(j)}_i
=
\widehat{\Omega}^{a,a^*}(L_i,M_i)
\widehat{\mathcal G}^{(t_j)}_{a}
(\nu \mid L_i,M_i,R_i=0)$
is the corresponding clever covariate.

The fluctuation parameter $\epsilon^{(j)}$ is estimated by logistic regression with outcome $I(\widetilde T_i=t_j,\delta_i=1),$
offset $\operatorname{logit}
\left\{
\widehat{\lambda}^{a}
(t_j\mid L_i,M_i,R_i=0)
\right\},$ and covariate $\widehat{w}^{(j)}_i$, using observations satisfying $A_i=a,\qquad S_i>0,\qquad R_i=0,\qquad
\widetilde T_i\geq t_j.$
Because the corresponding log-likelihood is additive across source sites, each site
can compute its local likelihood or gradient contribution and the fluctuation
parameter can be estimated using a standard federated optimization procedure such as
FedAvg.

As in the site-specific TMLE, the event hazards are targeted recursively backward in time, from the last follow-up time point to the first, so that the clever covariate at each time point depends only on hazard estimates that have already been updated. The TMLE therefore requires $h$ sequential federated targeting steps, one for each discrete follow-up time point. In this setting, the recursive strategy is preferable to a fully iterative targeting procedure, which would require repeatedly fitting $h$ logistic fluctuation models in a decentralized manner at each iteration, thereby substantially increasing the communication and computational burden.

\begin{algorithm}[H]
\small
\caption{Federated Targeted Updating of the Hazard Function
$\{\hat{\lambda}^{a}(t_j \mid L,M, R=0)\}_{j=1}^{h}$}
\label{alg:tmle_hazard_update_2}
\begin{algorithmic}[1]
\Require Data
$\{(\tilde T_i,\delta_i,L_i,M_i,A_i,S_i,R_i)\}_{i=1}^{n}$;
initial estimators
$\{\hat{\lambda}^{a}(t_j\mid L,M, R=0)\}_{j=1}^{h}$,
$\{\hat g(t_j\mid L,M, R=0)\}_{j=1}^{h}$,
and $\hat{\Omega}^{a,a^\ast}$.

\Ensure Updated hazard estimates
$\{\hat{\lambda}^{a,(1)}(t_j\mid L,M, R=0)\}_{j=1}^{h}$.

\vspace{0.2cm}
\State \textbf{Step 1: Targeting at the final time point $t_h$}

\State Compute the clever covariate
$
\hat w_i^{(h)}
=
\frac{\hat{\Omega}^{a,a^\ast}}
{\prod_{j < h}
\{1-\hat g(t_j\mid L_i,M_i)\}}.
$

\State Estimate the fluctuation parameter
$\hat\alpha^{(h)}$
from the logistic regression in a federated way
$$
\text{logit}\{\lambda^a(t_h\mid L,M, R=0)\}
=
\text{logit}\{\hat\lambda^a(t_h\mid L,M, R=0)\}
+
\alpha^{(h)}\hat w_i^{(h)},
$$
using only observations satisfying
$A_i=a$, $R_i=0$ and
$\tilde T_i\ge t_h$.
\State Share the $\hat\alpha^{(h)}$ with every server

\State Update the hazard estimate
$
\hat\lambda^{a,(1)}(t_h\mid L,M, R=0)
=
\text{expit}
\Bigl(
\text{logit}\{\hat\lambda^a(t_h\mid L,M, R=0)\}
+
\hat\alpha^{(h)}\hat w_i^{(h)}
\Bigr).
$

\State Recompute the corresponding survival quantity
$
\hat{\mathcal G}^{(t_{h-1},1)}(\nu\mid L,M, R=0).
$

\vspace{0.2cm}
\State \textbf{Step 2: Backward recursive targeting}

\For{$j=h-1,\ldots,1$}

    \State Compute the clever covariate
    $
    \hat w_i^{(j)}
    =
   \hat{\Omega}^{a,a^\ast}(L_i,M_i)\,
    \hat{\mathcal G}_{a}^{(t_j,1)}
    (\nu\mid L_i,M_i, R_i=0).
    $

    \State Estimate the fluctuation parameter
    $\hat\alpha^{(j)}$ in federated learning

    \State Update
    $
    \hat\lambda^{a,(1)}(t_j\mid L,M, R=0)
    =
    \text{expit}
    \Bigl(
    \text{logit}\{\hat\lambda^a(t_j\mid L,M, R=0)\}
    +
    \hat\alpha^{(j)}\hat w_i^{(j)}
    \Bigr).
    $

    \State Recompute
    $
    \hat{\mathcal G}^{(t_{j-1},1)}(\nu\mid L,M, R=0).
    $

\EndFor

\Return
$\{\hat{\lambda}^{a,(1)}(t_j\mid L,M, R=0)\}_{j=1}^{h}$
\end{algorithmic}
\end{algorithm}

\section{Data Application}

\paragraph{Study population} The study population included adults ($\geq18$ years) with psoriasis who initiated a biologic therapy during the study period. To ensure that patients were biologic-naïve at treatment initiation, we required no biologic treatment during the 1-year period preceding the index date. Patients with inflammatory bowel disease (IBD) were excluded from the study.

\paragraph{Covariates}

Covariates were assessed during the 1-year period preceding the index date. Comorbidity burden was measured using the Charlson Comorbidity Index (CCI). We also assessed the presence of other inflammatory diseases, including ankylosing spondylitis (AS), psoriatic arthritis (PsA), and other inflammatory diseases excluding AS, PsA, and IBD. In addition, we considered other comorbidities not captured by the CCI, including hypertension, dyslipidemia, and cardiovascular disease.

Medication use during the year preceding the index date was also assessed. The medications considered included systemic corticosteroids, nonsteroidal anti-inflammatory drugs (NSAIDs), methotrexate, ciclosporin, and acitretin.

\paragraph{Exposure}

The exposure of interest was the class of biologic therapy initiated at the index date. Biologic treatments were categorized into tumor necrosis factor inhibitors (TNFis), including adalimumab, etanercept, infliximab, and certolizumab, and interleukin-12/23 inhibitors (IL-12/23is), represented by ustekinumab.

\paragraph{Computation details}

For nuisance parameter estimation, we used random forests for the binary classification tasks involved in estimating $\lambda_k^a$, $g_k^a$, and $b_{k,p}^{a,a^*}$. These models were implemented in \texttt{scikit-learn} with 500 trees, a maximum tree depth of 10, a minimum of 20 observations per leaf, a minimum of 10 observations required to split an internal node, and $\sqrt{p}$ candidate features considered at each split. For continuous-outcome regression tasks, we used honest regression forests implemented with \texttt{EconML}. The nuisance parameters $\rho_p^{a^*}$, $\phi_k^a$, and $\tau_k^a$ were estimated using $L_2$-regularized logistic regression. For $U_p$ and $V_k$, we fitted $L_2$-regularized multinomial logistic regression models with a softmax link using FedAvg (Algorithm~1), without imposing a fixed maximum number of communication rounds. The federated optimization used a fixed learning rate of $10^{-3}$. As the analysis was primarily intended for illustration, all machine-learning hyperparameters were prespecified rather than selected by cross-validation. Finally, standardized indirect effect estimates were combined using the meta-analytic procedure described in Section~2 with uniform weights.

\begin{figure}
    \centering
    \includegraphics[angle=90,width=0.6\linewidth]{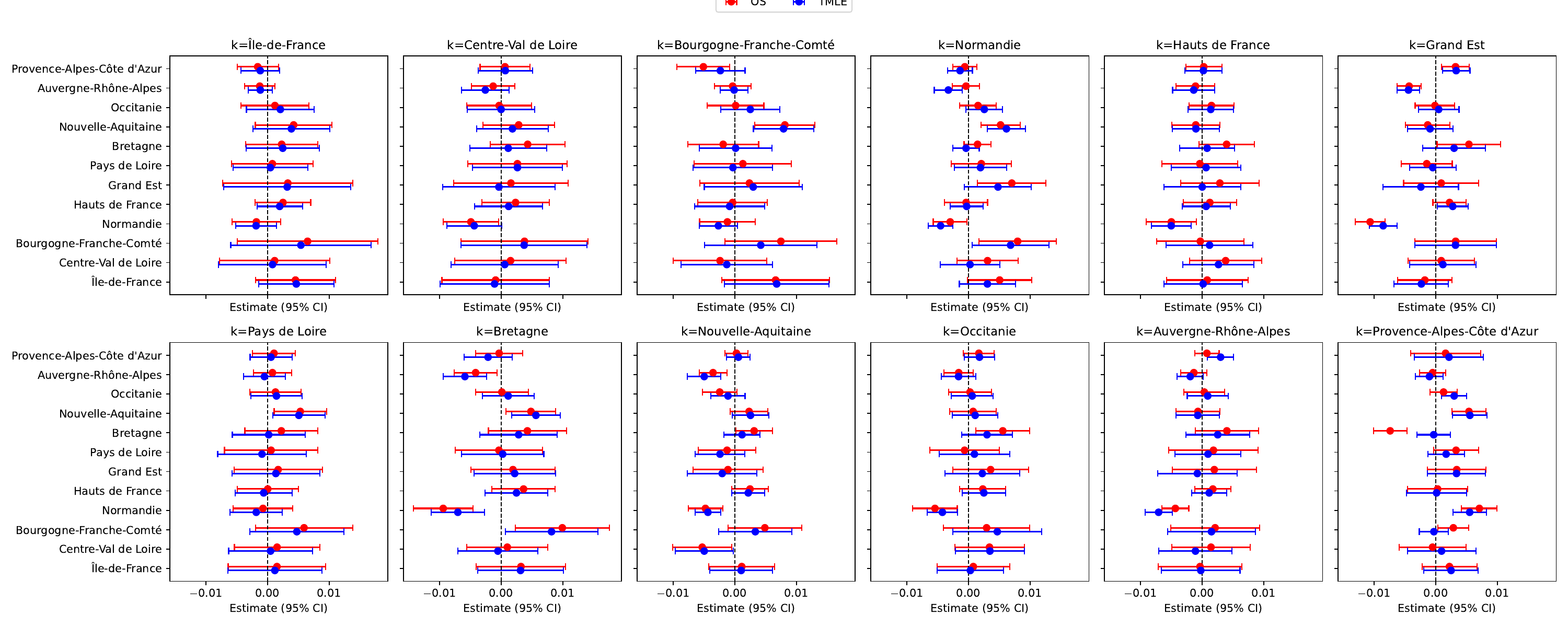}
    \caption{Treatment-Mediator heterogeneity across French administrative regions. Forest plot of the estimated
effects for each region using TMLE and OS}
    \label{fig:placeholder}
\end{figure}

\begin{figure}
    \centering
    \includegraphics[angle=90,width=0.6\linewidth]{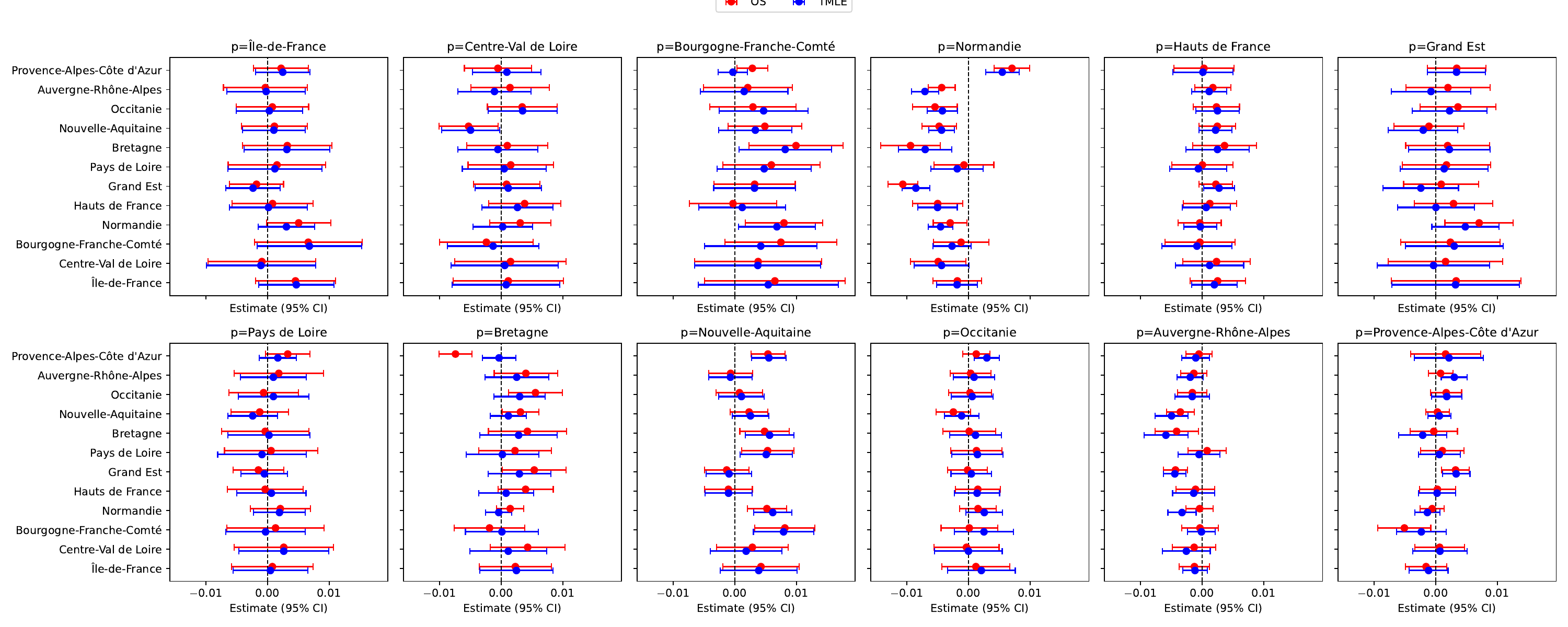}
    \caption{Treatment-Outcome heterogeneity across French administrative regions. Forest plot of the estimated
effects for each region using TMLE and OS}
    \label{fig:placeholder_2}
\end{figure}

\end{document}